\documentclass[a4paper]{article}
\newif\iflong
\newif\ifshort

\longtrue

\iflong
\else
\shorttrue
\fi

\usepackage{color}
\usepackage[utf8]{inputenc}
\usepackage[T1]{fontenc}
\usepackage{lmodern}
\usepackage[numbers,sort&compress]{natbib} %
\usepackage{mathtools}
\usepackage{amsmath,amssymb}
\usepackage{enumerate}
\usepackage{graphicx}
\usepackage{verbatim,booktabs}
\usepackage{xcolor}
\usepackage{verbatim}
\usepackage{paralist}
\usepackage{multicol}

\usepackage{ifthen}
\usepackage{tikz}
\usetikzlibrary{graphs}
\usetikzlibrary{calc}
\usetikzlibrary{positioning}
\usetikzlibrary{decorations.pathreplacing}
\usetikzlibrary{decorations.pathmorphing}

\pgfdeclarelayer{bg}    %
\pgfdeclarelayer{fg}    %
\pgfsetlayers{bg,main,fg}  %

\tikzstyle{avertex} = [color=black,fill=white,draw,circle,inner sep=0pt,minimum size=6pt]
\tikzstyle{bvertex} = [color=black,fill=black,draw,circle,semithick,inner sep=0pt,minimum size=6pt]
\tikzstyle{cvertex} = [color=black,fill=cyan,circle,semithick,inner sep=0pt,minimum size=12pt]
\tikzstyle{ivertex} = [color=black,fill=gray,draw,circle,semithick,inner sep=0pt,minimum size=6pt]

\tikzstyle{dial} = [draw = black, dashed, fill = blue, fill opacity = 0.3, semithick]
\tikzstyle{selgad} = [draw = black, fill = green, fill opacity = 0.3, semithick]
\tikzstyle{forb} = [draw = black, fill = red, fill opacity = 0.3, semithick]

\usepackage[absolute]{textpos}  %

\usepackage{amsthm}

\theoremstyle{plain}
\newtheorem{theorem}{Theorem}[section]
\newtheorem{lemma}[theorem]{Lemma}
\newtheorem{observation}[theorem]{Observation}
\newtheorem{proposition}[theorem]{Proposition}
\newtheorem{rrule}[theorem]{Reduction Rule}
\newtheorem{fact}[theorem]{Fact}%
\newtheorem{invariant}[theorem]{Invariant}%

\theoremstyle{definition}
\newtheorem{definition}[theorem]{Definition}

\usepackage{fullpage}
\usepackage{authblk}

\newenvironment{lemenum}{\begin{compactenum}[(i)]}{\end{compactenum}}
\newcommand{\proofparagraph}[1]{\smallskip\emph{#1}}
\newcommand{\proofparagraphf}[1]{\emph{#1}}

\usepackage[capitalize]{cleveref}

\newlength{\capitalheight}
\newcommand{\app}{$\mathrm{(\spadesuit)}$}

\crefname{lemma}{Lemma}{Lemmas}
\crefname{rrule}{Reduction Rule}{Reduction Rules}
\crefname{brule}{Branching Rule}{Branching Rules}
\crefname{fact}{Fact}{Facts}
\crefname{invariant}{Invariant}{Invariants}
\crefname{observation}{Observation}{Observations}
\crefname{subsection}{Section}{Sections}
\crefname{section}{Section}{Sections}
\crefname{figure}{Fig.}{Figs.}

\newcommand{\nat}{\mathbb{N}}
\newcommand{\Oh}{\mathcal{O}}

\newcommand{\NP}{\text{\normalfont NP}}
\newcommand{\FPT}{\text{\normalfont FPT}}

\newcommand{\poly}{\ensuremath{\operatorname{poly}}}

\newcommand{\monopolar}{{\sc Monopolar Recognition}}
\newcommand{\deltacluster}{{\sc Cluster-$\Pi_\Delta$-Partition}}
\newcommand{\deltaclustertitle}{{\sc Cluster-Pi-Delta-Partition}}

\newcommand{\picluster}{\textsc{Cluster-$\Pi$-Partition}}
\newcommand{\pCIS}{\textsc{Colorful Independent Set}}

\newcommand{\pipartition}{cluster\nobreakdash-$\Pi$ partition}

\newcommand{\cNP}{\ensuremath{\mathsf{NP}}}
\newcommand{\containment}{\ensuremath{\cNP \subseteq \mathsf{coNP}/\mathsf{poly}}}
\newcommand{\nocontainment}{\ensuremath{\cNP \not\subseteq \mathsf{coNP}/\mathsf{poly}}}

\newcommand{\AT}{A_{\mathsf{true}}}
\newcommand{\BT}{B_{\mathsf{true}}}
\newcommand{\VL}{V(\Lambda)}
\newcommand{\EL}{E(\Lambda)}
\newcommand{\D}{dist_{\Lambda}}
\newcommand{\N}{N_{\Lambda}}

\def\ie{{\em i.e.}}
\def\etal{{\em et al.}}
\newcommand{\inter}{\ensuremath{\mathsf{inter}}}
\newcommand{\important}{\ensuremath{\mathsf{rep}}}
\newcommand{\even}{\ensuremath{\mathsf{even}}}
\newcommand{\edge}{\ensuremath{\mathsf{edge}}}
\newcommand{\fixed}{\ensuremath{\mathsf{fixed}}}

\begin{document}

\title{
Solving Partition Problems Almost Always Requires\\Pushing Many Vertices Around%
\footnote{This manuscript is the full version of an article in the
  \emph{Proceedings of the 26th Annual European Symposium on
    Algorithms (ESA~'18)}~\cite{kanj_solving_2018}. CK gratefully
  acknowledges support by the DFG, project MAGZ, KO 3669/4-1. MS
  gratefully acknowledges support by the People Programme (Marie Curie
  Actions) of the European Union's Seventh Framework Programme
  (FP7/2007-2013) under REA grant agreement number 631163.11, by the
  Israel Science Foundation (grant number 551145/14), and by the
  European Research Council (ERC) under the European Union’s Horizon
  2020 research and innovation programme under grant agreement number
  714704.}}

\author[1]{Iyad Kanj}
\affil[1]{School of Computing, DePaul University
Chicago, USA, \texttt{ikanj@cs.depaul.edu}}

\author[2]{Christian Komusiewicz}
\affil[2]{Fachbereich Mathematik und Informatik, Philipps-Universität Marburg, Germany, \texttt{komusiewicz@informatik.uni-marburg.de}}

\author[3,4]{Manuel Sorge}
\affil[3]{Department of Industrial Engineering and Management, Ben-Gurion University of the Negev, Beer Sheva, Israel}
\affil[4]{Faculty of Mathematics, Informatics and Mechanics, University of Warsaw, Warsaw, Poland, \texttt{manuel.sorge@mimuw.edu.pl}}

\author[5]{Erik Jan van Leeuwen}
\affil[5]{Department of Information and Computing Sciences, Utrecht University, {The Netherlands} \texttt{e.j.vanleeuwen@uu.nl}}

\maketitle

\begin{textblock}{20}(.63, 14.8)
\includegraphics[width=40px]{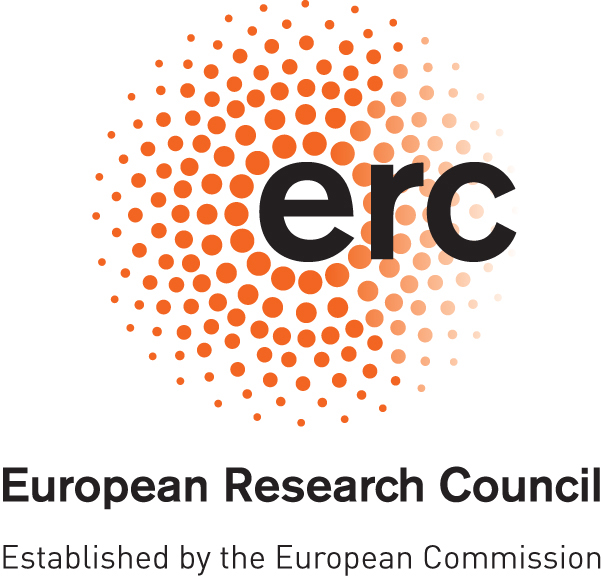}%
\end{textblock}
\begin{textblock}{20}(2.6, 14.78)
\includegraphics[width=60px]{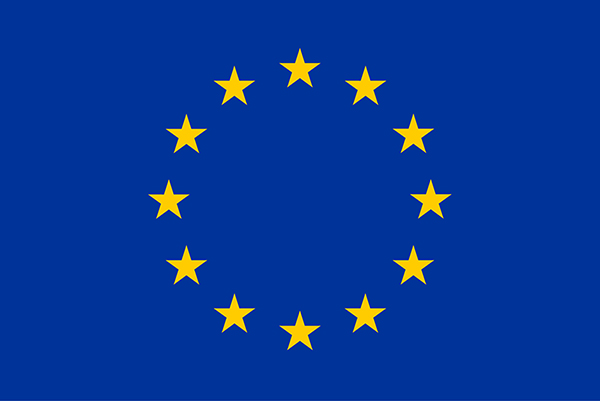}%
\end{textblock}

\begin{abstract}
\looseness=-1
A fundamental graph problem is to recognize whether the vertex set of a graph $G$ can be bipartitioned into sets $A$ and $B$ such that $G[A]$ and $G[B]$ satisfy properties $\Pi_A$ and $\Pi_B$, respectively. This so-called \textsc{$(\Pi_A,\Pi_B)$-Recognition} problem generalizes amongst others the recognition of $3$-colorable, bipartite, split, and monopolar graphs. %

In this paper, we study whether certain fixed-parameter tractable \textsc{$(\Pi_A,\Pi_B)$-Recognition} problems admit polynomial kernels. In our study, we focus on the first level above triviality, where $\Pi_A$ is the set of $P_3$-free graphs (disjoint unions of cliques, or cluster graphs), the parameter is the number of clusters in the cluster graph $G[A]$, and~$\Pi_B$ is characterized by a set $\mathcal{H}$ of connected forbidden induced subgraphs. We prove that, under the assumption that \nocontainment{}, \textsc{$(\Pi_A,\Pi_B)$-Recognition} admits a polynomial kernel if and only if $\mathcal{H}$ contains a graph with at most~$2$ vertices. In both the kernelization and the lower bound results, we exploit the properties of a \emph{pushing process}, which is an algorithmic technique used recently by Heggerness~\etal~and by Kanj \etal~to obtain fixed-parameter algorithms for many cases of \textsc{$(\Pi_A,\Pi_B)$-Recognition}, as well as several other problems.

\end{abstract}

\section{Introduction}\label{sec:intro}
Given two (induced-)hereditary graph properties $\Pi_A$ and $\Pi_B$, a graph $G$ is a \emph{$(\Pi_A,\Pi_B)$-graph} if $V(G)$ can be partitioned into two sets $A,B$ such that $G[A] \in \Pi_A$ and $G[B] \in \Pi_B$. We call $(A,B)$ a \emph{$(\Pi_{A},\Pi_{B})$-partition} of $G$. The \textsc{$(\Pi_{A},\Pi_{B})$-Recognition} problem is to recognize whether a given graph is a $(\Pi_A,\Pi_B)$-graph. This generic problem captures a wealth of famous problems, including the recognition of $3$-colorable, bipartite, co-bipartite, and split graphs, and \textsc{$\Pi$-Vertex Deletion}, which asks for a partition~$(A,B)$ such that~$G[A]\in \Pi$ and~$G[B]$ has order at most~$k$ for some given~$k$.\footnote{The order of a graph is its number of vertices.}

Note that, since $\Pi_A$ and $\Pi_B$ are hereditary, they are characterized by a (not necessarily finite) set of forbidden induced subgraphs.
In the most interesting cases, the characterization of one of the two properties $\Pi_A$ and $\Pi_B$ includes only forbidden induced subgraphs whose order is at least three.
Indeed, if the characterizations of $\Pi_A$~and~$\Pi_B$ both include a forbidden induced subgraph of order exactly~$2$, then~\textsc{$(\Pi_{A},\Pi_{B})$-Recognition} can be solved in linear time: The property~$\Pi_A$, and similarly~$\Pi_B$, is either the class of complete graphs or the class of edgeless graphs.
Thus, the set of~$(\Pi_{A},\Pi_{B})$-graphs is either the set of split graphs, which can be recognized in linear time~\cite{FH77}, the set of bipartite graphs, or the set of co-bipartite graphs, which both are also well known to be recognizable in linear time.
In this manuscript we thus focus on the case where the forbidden induced subgraph characterization of $\Pi_A$ or $\Pi_B$ does not include graphs of size at most two.

When~$\Pi_A$~or~$\Pi_B$ is characterized by larger forbidden induced subgraphs, then \textsc{$(\Pi_{A},\Pi_{B})$-Re\-cog\-ni\-tion} is more computationally complex~\cite{Ach97,Far04,KS97}.
In particular, \textsc{$(\Pi_{A},\Pi_{B})$-Re\-cog\-ni\-tion} is NP-hard for all $\Pi_A$ and $\Pi_B$ that are characterized by a set of \emph{connected} forbidden induced subgraphs as long as at least one of~$\Pi_A$ and~$\Pi_B$ has only forbidden induced subgraphs of order at least three~\cite{Far04}.
The restriction to connected forbidden induced subgraphs in this statement is necessary in order to get such a sweeping classification result. If we allow disconnected forbidden induced subgraphs in the characterizations of $\Pi_A$ or $\Pi_B$, then there are some polynomial-time solvable cases.
Consider, for example, the set of \emph{unipolar graphs} which can be recognized in polynomial time~\cite{FH77,EW14,MY15,TC85}.
Unipolar graphs are $(\Pi_A,\Pi_B)$-graphs wherein~$\Pi_A$ is the set of complete graphs, which is characterized by forbidding the graph containing two nonadjacent vertices, and~$\Pi_B$ is the set of cluster graphs, which is characterized by forbidding the (simple) path $P_3$ on three vertices.
Indeed, the general NP-hardness result from above does not apply to unipolar graphs, because, while forbidden subgraph characterization for~$\Pi_B$ does not contain graph of order at most two, the set of complete graphs is characterized by a \emph{disconnected} forbidden induced subgraph.
There are further similar polynomial-time solvable cases of \textsc{$(\Pi_{A},\Pi_{B})$-Re\-cog\-ni\-tion}~\cite{KKSL17}.
We are here concerned with the NP-hard cases of \textsc{$(\Pi_{A},\Pi_{B})$-Re\-cog\-ni\-tion} and thus we focus on the cases where both $\Pi_A$ and $\Pi_B$ are characterized only by connected forbidden subgraphs.
Note that, equivalently, $\Pi_A$ and~$\Pi_B$ are each closed under the disjoint union of graphs.\footnote{That is, the disjoint union of two graphs that each satisfy $\Pi_A$ (resp.~$\Pi_B$) also satisfies $\Pi_A$ (resp.~$\Pi_B$).}

Many \textsc{$(\Pi_{A},\Pi_{B})$-Recognition} problems that can be characterized by forbidden induced subgraphs were shown to be fixed-parameter tractable, for instance, when $\Pi_{A}$ is the class of graphs that is a disjoint union of $k$ cliques and $\Pi_{B}$ is the set of edgeless graphs or the set of cluster graphs~\cite{KKSL17}, where $k$ is the parameter.
We thus aim here to complement these results by studying the kernelization complexity of \textsc{$(\Pi_{A},\Pi_{B})$-Recognition}.

By the discussion above, the first interesting case of \textsc{$(\Pi_{A},\Pi_{B})$-Recognition} to study with respect to kernelization is when the two properties are closed under the disjoint union and one of the two properties, say $\Pi_A$, is characterized by connected forbidden induced subgraphs of size no less than three.
We thus consider the first level above triviality of \textsc{$(\Pi_{A},\Pi_{B})$-Recognition}, by letting $\Pi_A$ be the hereditary class characterized by a single forbidden induced subgraph, which, in addition, is the most simple connected graph with at least three vertices: the path $P_3$.
This leads to the following problem:
 
\begin{quote}
  \picluster\\
  \textbf{Input:} A graph $G=(V,E)$.\\
  \textbf{Question:} Can $V(G)$ be partitioned into two parts $A$ and $B$ such that $G[A]$~is a cluster
  graph and~$G[B]\in \Pi$?
\end{quote}
\picluster{} generalizes the recognition problem of many graph classes, such as the recognition of \emph{monopolar graphs}~\cite{BHK14,CH14,CH14b,LN14} ($\Pi$ is the set of edgeless graphs), \emph{$2$\nobreakdash-subcolorable} graphs~\cite{BFNW02,FJBS03,Gimbel2003,Sta08} ($\Pi$ is the set of cluster graphs), and several others~\cite{AFM15,BO17,CC86}. %
\picluster{} is NP-hard in these special cases, and by the general hardness result mentioned above, it is NP-hard for all $\Pi$ that are characterized by a set of connected forbidden induced subgraphs~\cite{Ach97,Far04,KS97}. To cope with this hardness, we consider the number~$k$ of clusters in the cluster graph~$G[A]$ as a parameter. As mentioned above, this parameter led to a number of tractability results for \picluster. In particular, recognizing monopolar graphs and recognizing 2-subcolorable graphs are fixed-parameter tractable with respect to this parameter~\cite{KKSL17}.

\subparagraph{Our Results}
The result we obtain gives a complete characterization of the kernelization complexity of {\picluster}:

\begin{theorem} \label{thm:main}
Let $\Pi$ be a graph property characterized by a (not necessarily finite) set~$\mathcal{H}$ of connected forbidden induced subgraphs. Then unless \containment, {\picluster} parameterized by the number $k$ of clusters in the cluster graph~$G[A]$ admits a polynomial kernel if and only if $\mathcal{H}$ contains a graph of order at most~$2$.
\end{theorem}
The positive result in \cref{thm:main} corresponds to the recognition of monopolar graphs. Indeed, the graph properties with forbidden induced subgraphs of order~$2$ are ``being edgeless'' and ``being nonedge-less'', but the latter is not characterized by connected forbidden induced subgraphs. Moreover, one application of \cref{thm:main} is to show that the recognition of 2-subcolorable graphs does not admit a polynomial kernel parameterized by the number of clusters in~$G[A]$.

One might also consider two other parameters: the size of a largest cluster in~$G[A]$ and the size of one of the sides.
The size of a largest cluster in~$G[A]$ will not lead to tractability, as~{\picluster} is NP-hard on subcubic graphs, even %
when~$\Pi$ is the set of edgeless graphs~\cite{LN14}.
Thus, we consider the number~$k$ of vertices in the graph $G[B]$, even for the broader \textsc{$(\Pi_{A},\Pi_{B})$-Recognition} problem. We previously proved a general fixed-parameter tractability result in this case~\cite{KKSL17}. Here, we observe a very general kernelization result: %

\begin{theorem}\label[theorem]{thm:general-kernel}
\textsc{$(\Pi_{A},\Pi_{B})$-Recognition} parameterized by $k$, the maximum size of $B$, admits a polynomial kernel with $\Oh((d+1)! (k+1)^{d})$ vertices, when $\Pi_A$ can be characterized by a collection $\mathcal{H}$ of forbidden induced subgraphs, each of size at most~$d$, and $\Pi_B$ is hereditary.
\end{theorem}
We obtain a better bound on the number of vertices in the kernel for {\deltacluster}, the restriction of {\picluster} to the case when all graphs containing a vertex of degree at least~$\Delta+1$ are forbidden induced subgraphs of $\Pi$.

\begin{theorem}\label{thm:kernelis}
{\deltacluster} parameterized by $k$, the maximum size of $B$, admits a polynomial kernel with~$\Oh((\Delta^2 +1) \cdot k^2)$ vertices.
\end{theorem}

\subparagraph{Our Techniques}
We obtain our main kernelization result, stated in Theorem~\ref{thm:main}, by studying an algorithmic method that we call the \emph{pushing process}. It was employed in~\cite{HKLRS13,KPRS16,KKSL17} to obtain fixed-parameter tractability results for several problems. In the context of \textsc{$(\Pi_{A},\Pi_{B})$-Recognition}, the pushing process was employed as part of an algorithmic technique, referred to as \emph{inductive recognition}, that works as follows~\cite{KKSL17}. The algorithm for the \textsc{$(\Pi_{A},\Pi_{B})$-Recognition} problem under consideration empties the input graph, and adds vertices back one by one while maintaining a valid partition. However, adding a vertex might invalidate a previously valid partition. To remedy this, the pushing process comes in: vertices are \emph{pushed} from one part of the partition to the other part in the hope of obtaining a valid partition again. Earlier, Heggernes \etal~\cite{HKLRS13} had employed a pushing process, as part of an algorithmic technique, referred to as \emph{iterative localization}, which works very similarly to inductive recognition, to show the fixed-parameter tractability of computing the cochromatic number of perfect graphs and the stabbing number of disjoint rectangles with axes-parallel lines (using the standard parameters). Kolay \etal~\cite{KPRS16} also applied it in follow-up work related to the cochromatic number. All aforementioned results, for the various problems under consideration, rely on iteratively/inductively maintaining a valid partition of the input instance, as the instance elements (vertices, rectangles, etc.) are added one at a time. After each element is added, which may invalidate the partition, a pushing process is applied in order to try to repair the partition. This process results in pushing some elements from each part of the partition to the other part, and may have a cascading~effect.

A crucial ingredient in applying the pushing process is to understand the \emph{ava\-lanches} caused by this process. For \textsc{$(\Pi_{A},\Pi_{B})$-Recognition}, an avalanche is triggered when a vertex is pushed to~$A$; this may imply that several other vertices must be pushed to~$B$, which, in turn, triggers the pushing of yet more vertices to~$A$, and so on. Similar effects are visible in the aforementioned cochromatic number and rectangle stabbing number problems~\cite{HKLRS13,KPRS16}. The contribution of the previous works~\cite{HKLRS13,KKSL17,KPRS16} was to bound the depth of this process by some function of the parameter, leading to fixed-parameter algorithms. However, such a bound does not provide an answer to the question of which vertices trigger avalanches and their continued rolling, and whether the number of vertices relevant to avalanches can somehow be limited.

This question can be naturally formalized in terms of the kernelization complexity of problems to which the pushing process applies. A kernel reduces the size of the graph and thus directly reduces the number of vertices triggering or being affected by avalanches when an algorithm based on the pushing process is applied to the kernelized instance. In previous work, Kolay \etal~\cite{KPRS16} studied the kernelization complexity of computing the cochromatic number of a perfect graph $G$, which is the smallest number $k = r + \ell$ such that $V(G)$ can be partitioned into $r$ sets that each induces a clique and $\ell$ sets that each induces an edgeless graph. This problem has a parameterized algorithm using a pushing process~\cite{HKLRS13}, but Kolay \etal~\cite{KPRS16} showed that, unless \containment, this problem does not admit a polynomial kernel parameterized by $r+\ell$. This suggests that, for this problem, one cannot control the number of vertices affected by avalanches. The kernelization complexity of \textsc{$(\Pi_{A},\Pi_{B})$-Recognition}, however, has not been studied so far. Hence, it is open whether avalanches can be controlled to affect few vertices in this case.

In this work, we study the role that the pushing process plays in characterizing the kernelization complexity of \textsc{$(\Pi_{A},\Pi_{B})$-Recognition}.
The pushing process, and a deeper understanding of the avalanches it causes, turn out to be central to both directions of Theorem~\ref{thm:main}. Indeed, we show that, while for a specific $\Pi$ the pushing process can be used to witness a small vertex set of size $k^{\Oh(1)}$ containing the vertices affected by avalanches, for all other $\Pi$, such a set of polynomial size is unlikely to~exist.

To prove the positive result in the theorem, we first perform a set of data reduction rules to identify some vertices that are part of~$A$ or~$B$ in \emph{any} partition~$(A,B)$ of $V(G)$ such that $G[A]$ is a cluster graph with at most $k$ clusters and $G[B]$ is edgeless. More importantly, these rules restrict the combinatorial properties of the graph induced by the remaining vertices. With these restrictions, it becomes possible to represent the structure of the avalanches that occur using an auxiliary bipartite graph. This graph enables two further reduction rules that lead to the polynomial kernel.

For the negative result, we observe that the bipartite graph constructed in the kernel is closely tied to the deterministic behavior of the pushing process for monopolar graphs: when an edge in~$G[B]$ is created by pushing a vertex to~$B$, the other endpoint of the edge must be pushed to~$A$ (recall that $G[B]$ must become edgeless). This limits the avalanches. However, for more complex properties $\Pi_B$, such a simple correspondence no longer exists. In particular, when the forbidden induced subgraphs have order at least~$3$, pushing a vertex to~$B$ may create a forbidden induced subgraph in~$G[B]$ that can be repaired in at least two different ways. Then the pushing process starts to behave nondeterministically, and the avalanches grow beyond control. We exploit this intuition to exclude the existence of a polynomial kernel, unless \containment, by providing a cross-composition.

\section{Preliminaries}
\label{sec:prelim}
For $\ell \in \mathbb{N}$, we use $[\ell]$ to denote $\{1, 2, \ldots, \ell\}$.

\subparagraph{Graphs} We follow standard graph-theoretic notation~\cite{Die12}. Let $G$ be a graph. By $V(G)$ and $E(G)$ we denote the vertex-set and the edge-set of $G$, respectively. Throughout the paper, we use~$n:=|V(G)|$ to denote the number of vertices in~$G$ and~$m:=|E(G)|$ to denote its number of edges. We also say that $G$ is of \emph{order} $|V(G)|$. We assume~$n=\Oh(m)$ since isolated vertices can be safely removed in the problems that we consider. For $X \subseteq V(G)$, $G[X] = (X, \{e \mid e \in E(G) \cap X\}$ denotes the \emph{subgraph of $G$ induced by $X$}. For a vertex $v \in G$, $N(v) = \{u \mid \{u, v\} \in E(G)\}$ and $N[v] = N(v) \cup \{v\}$ denote the \emph{open neighborhood} and the \emph{closed neighborhood} of $v$, respectively. For $X \subseteq V(G)$, we define $N(X) := (\bigcup_{v \in X}N(v))\setminus X$ and $N[X] := \bigcup_{v \in X}N[v]$, and for a family $\mathcal{X}$ of subsets $X \subseteq V(G)$, we define  $N(\mathcal{X}) := (\bigcup_{X \in \mathcal{X}}N(X)) \setminus (\bigcup_{X \in \mathcal{X}}X)$ and $N[\mathcal{X}] := \bigcup_{X \in \mathcal{X}}N[X]$.

We say that a vertex $v$ is \emph{adjacent to a subset $X \subseteq V(G)$ of vertices}  if $v$ is adjacent to at least one vertex in $X$. Similarly, we say that two vertex sets~$X\subseteq V(G)$ and~$Y\subseteq V(G)$ are adjacent if there exist~$x\in X$ and~$y\in Y$ that are adjacent. If $X$ is any set of vertices in $G$, we write $G-X$ for $G- X$. For a vertex $v \in V(G)$, we write $G-v$ for $G - \{v\}$.

\subparagraph{Graph Partitions}
We say a partition $(A, B)$ of $V(G)$ is a \emph{cluster-$\Pi$ partition} if (1) $G[A]$ is a cluster graph and (2) $G[B]\in \Pi$.
A \emph{monopolar partition} of a graph $G$ is a partition of $V(G)$ into a cluster graph and an independent set. The \monopolar{} problem is defined as follows:

\begin{quote}
  \monopolar\\
  \textbf{Input:} A graph $G=(V,E)$ and an integer $k$.\\
  \textbf{Question:} Does $G$ admit a monopolar partition $(A,B)$ such that the number of clusters in the cluster graph $G[A]$ is at most $k$?
\end{quote}

For an instance $(G,k)$ of {\monopolar}, a monopolar partition of $G$ is \emph{valid} if the number of clusters in the cluster graph of the partition is at most $k$.

\subparagraph{Parameterized Complexity}
A {\em parameterized problem} is a tuple~$(P, \kappa)$, where $P \subset \Sigma^*$ is a language over some finite alphabet~$\Sigma$ and $\kappa \colon \Sigma^* \to \mathbb{N}$ is a \emph{parameterization}. For a given instance $x \in \Sigma^*$, we also say $\kappa(x)$ is the \emph{parameter}.
A parameterized problem $(P, \kappa)$ is {\it fixed-parameter tractable} (\FPT), if there exists an algorithm that on input $x \in \Sigma^*$
decides if $x$ is a yes-instance of $P$, that is, $x \in P$, and that runs in time $f(\kappa(x))n^{\Oh(1)}$,
where $f$ is a computable function independent of $n = |x|$.
A parameterized problem is {\em kernelizable}
if there exists a polynomial-time reduction that maps an instance $x$ of
the problem to another instance $x'$ such that: (1) $|x'| \leq \lambda(\kappa(x))$ for
some computable function $\lambda$, (2) $\kappa(x') \leq \lambda(\kappa(k))$, and (3) $x$ is a yes-instance
of the problem if and only if $x'$ is. The instance
$x'$ is called the {\em kernel} of~$x$, and $|x'|$ is the \emph{kernel size}.  It is well known that a parameterized problem is \FPT{} if and only if it is kernelizable~\cite{DF13,kernelbook}, and a natural question to ask for an \FPT{} problem is whether or not it has a kernel of polynomial size. %
We refer the reader to~\cite{kernelbook} for an in-depth discussion about kernelization.%

Let $Q\subseteq \Sigma^*$ be a language and $(P, \kappa)$ a parameterized problem, \ie, $P$ is a language and $\kappa \colon \Sigma^* \to \mathbb{N}$ a parameterization. An \emph{or-cross-composition} from~$Q$
into $(P, \kappa)$ is a polynomial-time
algorithm that, given $t$~instances $q_1, \ldots, q_t \in \Sigma^*$ of
$Q$, computes an instance~$r \in \Sigma^*$ %
such that \[\kappa(r) \leq \poly\left(\log
t + \max_{i = 1}^t|q_i|\right),\] and $r \in P$ if and only if $q_i \in Q$ for some
$i \in [t]$. %
We have the following:

\begin{theorem}[\cite{BJK14}] \label{thm:crosscomp} Let $Q \subseteq \Sigma^*$ be an NP-hard language and $(P, \kappa)$ be a parameterized problem. If there is an or-cross-composition from~$Q$ into $(P, \kappa)$ and $(P, \kappa)$ admits a
polynomial-size problem kernel, then \containment.
\end{theorem}

\noindent For more discussion on parameterized complexity, we refer to the literature~\cite{DF13,CFK+15}.

\section{A Polynomial Kernel for Monopolar Recognition Parameterized by the Number of Clusters}
\label{sec:monopolar-cluster}
\looseness=-1\noindent The outline of the kernelization algorithm is as follows.  First, we compute a
decomposition of the input graph into sets of vertex-disjoint maximal cliques which we
call a \emph{clique decomposition}. This decomposition is used and updated throughout the data-reduction procedure. We also maintain sets of vertices that are determined to belong to~$A$
or~$B$. We first apply a sequence of reduction rules whose aim is roughly to bound the number
of cliques and the number of edges between the cliques in the decomposition, and to
restrict the structure of edges between cliques. Then, we build an auxiliary graph
to model how the placement of a vertex in~$A$ or~$B$ implies an avalanche of placements
of vertices in~$A$ and~$B$. If this avalanche creates too many clusters in~$A$, then this determines the
placement of certain vertices in~$A$ or $B$, and triggers another reduction rule. If this reduction rule
does not apply anymore, then the size of the auxiliary graph is bounded, which in turn
helps bounding the size of the instance.

\subsection{Clique Decompositions}

Say that a clique $C$ is a {\em large clique} if $|C| \geq 3$, an \emph{edge clique} if $|C|=2$ (\ie, $C$ is an edge), and a \emph{vertex clique} if $|C|=1$  (\ie, $C$ consists of a single vertex).
Let $(G, k)$ be an instance of \monopolar. Suppose that $\AT \subseteq V(G)$ and $\BT \subseteq V(G)$ are subsets of vertices that have been determined to be in $A$ and $B$, respectively, in any valid monopolar partition of $(G, k)$. We define a decomposition $(C_1, \ldots, C_r)$ of $V(G) \setminus (\AT \cup \BT)$, referred to as a \emph{nice clique decomposition}, that partitions this set into vertex-disjoint cliques $C_1, \ldots, C_r$, $r \geq 1$, such that the tuple $(C_1, \ldots, C_r)$ satisfies the following properties (see \cref{fig:clique-decomp} for an illustration):
\begin{compactenum}[(i)]%
\item In the decomposition tuple $(C_1, \ldots, C_r)$, the large cliques appear before the edge cliques, and the edge cliques, in turn, appear before the vertex cliques; that is, for each large clique $C_i$  and for each edge or vertex clique~$C_j$  we have $i < j$, and for each edge clique $C_i$ and for each vertex clique $C_j$ we have~$i < j$.

\item Each clique $C_i$, $i \in [r-1]$, is maximal in $\bigcup_{j=i}^{r} C_j$; that is, there does not exist a vertex $v \in \bigcup_{j=i+1}^{r} C_j$ such that $C_i \cup \{v\}$ is a clique.

\item The subgraph of $G$ induced by the union of the edge cliques and vertex cliques does not contain any large clique.
\end{compactenum}
The following fact is implied by property~(ii) above:

\begin{figure}[t]
  \centering
  \begin{tikzpicture}[xscale=0.8,yscale=0.8]

    \node[bvertex] (a1) at (1.2,-0.5) {}; \node[bvertex] (a2) at (2,-1.8) {}; \node[bvertex]
    (a3) at (2.8,-0.5) {};

    \draw[semithick] (a1)--(a2)--(a3)--(a1);

    \foreach \i in {1,3}{ \node[bvertex] (a\i1) at (3.2+2*\i,-0.5) {}; \node[bvertex] (a\i2) at
      (3.2+2*\i,-2.1) {}; \node[bvertex] (a\i3) at (4.8+2*\i,-0.5) {}; \node[bvertex] (a\i4) at
      (4.8+2*\i,-2.1) {};

      \draw[semithick] (a\i1)--(a\i2)--(a\i3)--(a\i4)--(a\i3)--(a\i2)--(a\i4)--(a\i1)--(a\i3);
    }

    \node[bvertex,label=above:$C_{16}$] at (4.8+2,-0.5) {};

    \node at (1.35,-1.5) {$C_1$};
    \node at (4.5,-1.5) {$C_2$};
    \node at (11.3,-1.5) {$C_3$};
    \node at (7.75,-1.5) {$C_4$};

    \foreach \x in {2,3,7}{ \node[bvertex] (b\x) at (1.25*\x,1) {}; }

    \foreach \x in {0,1,4,5,6,8,9,10}{
      \pgfmathsetmacro\y{int(\x+5)}
      \node[bvertex, label=above:$C_{\y}$] (b\x) at (1.25*\x,1) {};
    }

    \draw[semithick] (a1)--(b0); \draw[semithick] (a1)--(b1); \draw[semithick] (a1)--(b2);
    \draw[semithick] (a3)--(b3); \draw[semithick] (a2)--(b2); \draw[semithick] (a3)--(b2);

    \draw[semithick] (a11)--(b2); \draw[semithick] (a13)--(b2); \draw[semithick] (a12)--(b3);
    \draw[semithick] (a11)--(b3); \draw[semithick] (a14)--(b7); \draw[semithick] (a12)--(b5);
    \draw[semithick] (a14)--(b4); \draw[semithick] (a34)--(b5); \draw[semithick] (a31)--(b6);
    \draw[semithick] (a32)--(b8); \draw[semithick] (a32)--(b9); \draw[semithick] (a33)--(b10);

        \begin{pgfonlayer}{bg}
          { \draw [cyan, rounded corners, fill] (2,-1.8) -- (1.2,-0.5) -- (2.5, 1) --
            (2.8,-0.5) -- (2,-1.8);
          } { \draw [cyan, rounded corners, fill] (1.25*3,1) -- (3.2+2,-0.5) -- (3.2+2,-2.1);
          } { \draw [cyan, rounded corners, fill] (3.2+2*3,-2.1) -- (3.2+2*3,-0.5) --
            (4.8+2*3,-0.5) -- (4.8+2*3,-2.1); } { \draw [cyan, line width=4pt ] (1.25*7,1) --
            (4.8+2,-2.1); } { \foreach \x in {0,1,4,5,6,8,9,10}{ \node[cvertex] (c\x) at
              (1.25*\x,1) {}; } \node[cvertex] (c13) at (4.8+2,-0.5) {}; }

        \end{pgfonlayer}

      \end{tikzpicture}

      \caption{A monopolar graph with a clique decomposition. The large cliques are~$C_1$,~$C_2$, and~$C_3$; the edge clique is~$C_4$; all other cliques are small cliques. The small cliques form an independent set.  }
      \label[figure]{fig:clique-decomp}
    \end{figure}
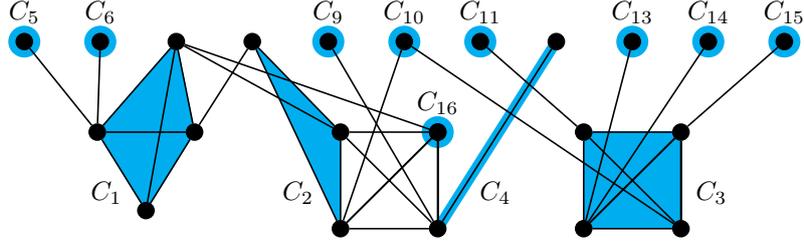

    \begin{fact}\label[fact]{fact:is}
      The vertex cliques in a nice clique decomposition form an independent set in~$G$.
\end{fact}
A nice clique decomposition of $V(G) \setminus (\AT \cup \BT)$ can be computed as follows. Let $V'=V(G) \setminus (\AT \cup \BT) \neq \emptyset$. We check whether~$G[V']$ contains a clique~$C$ of size at least three. If this is the case, then we find a maximal clique $C'\supseteq C$ in $G[V']$, add~$C'$ as a large clique to the decomposition, set $V' \leftarrow V' -C'$ and repeat. Otherwise, $G[V']$ does not contain any clique of size 3, we check whether~$G[V']$ contains an edge clique~$C$ (\ie, two endpoints of an edge), add~$C$ to the decomposition, set $V'\leftarrow V'-C$ and repeat. If no edge clique exists in $G[V']$, then the remaining vertices in $V'$ form an independent set, and we add each one of them to the decomposition as a vertex clique. This process can be seen to run in polynomial time, but we will use the following more precise bound.
\begin{lemma}\label[lemma]{lem:cluster-dec-time}
  A nice clique decomposition of~$G$ can be computed in~$\Oh(nm)$ time.
\end{lemma}
\begin{proof}
  First, in~$\Oh(nm)$ time, compute a list of \emph{all} triangles in~$G$. Then, label all vertices
  as \emph{free}. Let~$G'$ denote the graph~$G[V']$. Process the list from head to
  tail; that is, consider each triangle in the list. If one vertex of the triangle is not
  labeled as free, then continue with the next triangle. If all vertices in this triangle
  are labeled as free, then compute a maximal clique in~$G'$ containing this
  triangle and consisting only of free vertices. This can be done in~$\Oh(m)$ time~\cite{Skie08}.  Add the maximal clique to the
  decomposition as described above, remove all vertices of the maximal clique from~$G'$,
  and unlabel all vertices of the maximal clique. Overall this step takes~$\Oh(nm)$ time,
  since we encounter at most~$n/3$ triangles whose vertices are labeled free. Once all
  triangles in the list are processed, compute a set of edge cliques
  in~$\Oh(m)$ time by computing a maximal matching in~$G[V']$. Finally, add all remaining
  vertices as vertex cliques in~$\Oh(n)$ time.
\end{proof}

Let $(G, k)$ be an instance of \monopolar. We initialize $\AT =\BT = \emptyset$, $V' =V(G) \setminus (\AT \cup \BT)$, and we compute a nice clique decomposition $(C_1, \ldots, C_r)$ of $V'$. We will then apply reduction rules to simplify the instance $(G, k)$. During this process, we may identify vertices in $V'$ to be added to $\AT$ or $\BT$. At any point in the process, we will maintain a partition $(\AT, \BT, C_1, \ldots, C_r)$ of $V(G)$ such that (1) $\AT \subseteq A$ and $\BT \subseteq B$ for any valid monopolar partition $(A, B)$ of $V(G)$, and (2) $(C_1, \ldots, C_r)$ is a nice clique decomposition of $V'=V(G) \setminus (\AT \cup \BT)$. We call such a partition $(\AT, \BT, C_1, \ldots, C_r)$ a \emph{normalized partition} of $V(G)$.
\subsection{Basic Reduction Rules}
We now describe our basic set of reduction rules. After the application of a reduction rule, a normalized partition may change as the result of moving vertices from $\bigcup_{i=1}^{r}C_i$ to $\AT \cup \BT$, and we will need to compute a nice clique decomposition of the resulting (new) set $V(G) \setminus (\AT \cup \BT)$. However, a vertex that has been moved to $\AT$ (resp.~$\BT$) will remain in $\AT$ (resp.~$\BT$). When a reduction rule is applied, we assume that no reduction rule preceding it\iflong, with respect to the order in which the rules are listed,\else\fi~is applicable.

The following rule is straightforward:

\begin{rrule}\label[rrule]{rule:0}
Let $(\AT, \BT, C_1, \ldots, C_r)$ be a normalized partition of $V(G)$.  If $\AT$ is not a cluster graph with at most $k$ clusters, or $\BT$ is not an independent set, then reject the instance $(G, k)$.
\end{rrule}

The following rule is correct because, for every monopolar partition~$(A, B)$ of~$G$, $\BT \subseteq B$ and $B$ is an independent set.

\begin{rrule}\label[rrule]{rule:0.1}
Let $(\AT, \BT, C_1, \ldots, C_r)$ be a normalized partition of $V(G)$. If there is a vertex $v \in V(G) \setminus (\AT \cup \BT)$ that is adjacent to~$\BT$, then set $\AT= \AT \cup \{v\}$.
\end{rrule}

The following rule is correct, since $\AT \subseteq A$ and the cluster graph~$G[A]$ is~$P_3$-free for every monopolar partition $(A, B)$ of $G$:

\begin{rrule}\label[rrule]{rule:0.5}
Let $(\AT, \BT, C_1, \ldots, C_r)$ be a normalized partition of $V(G)$. If there is a vertex $v \in V(G) \setminus (\AT \cup \BT)$ such that $G[\AT \cup \{v\}]$ contains an induced $P_3$,
then set $\BT= \BT \cup \{v\}$.
\end{rrule}

After the exhaustive application of the above rules, we have that $G[\AT]$ is a cluster graph, $G[\BT]$ fulfills~$\Pi$, and any vertex in~$V(G)\setminus (\AT\cup \BT)$ can be moved to either side of the partition without creating a forbidden induced subgraph.

\noindent\looseness=-1 The next two reduction rules restrict the number and type of edges incident to large cliques.

\begin{rrule}\label[rrule]{rule:3}
Let $(\AT, \BT, C_1, \ldots, C_r)$ be a normalized partition of $V(G)$. If there exists a vertex $v \in V(G) \setminus (\AT \cup \BT)$ and a large clique $C_i$ such that $ 1 < |N(v) \cap C_i| \leq |C_i|-1$, then set $\AT= \AT \cup (N(v) \cap C_i)$.
\end{rrule}
\begin{proof}[Correctness Proof]
Since $1 < |N(v) \cap C_i| \leq |C_i|-1$, $v$ has at least two neighbors $u, w \in C_i$ and at least one nonneighbor $x \in C_i$.  If a vertex $z \in N(v) \cap C_i$ is in $B$, for any valid monopolar partition $(A, B)$ of $V(G)$,  then since $B$ is an independent set, it follows that $C_i - \{z\} \subseteq A$. In particular, $v$ is in $A$, at least one of $u, w$, say $u$, is in $A$, and $x$ is in $A$. But this implies that $(v, u, x)$ forms an induced $P_3$ in $A$, contradicting that $A$ is a cluster graph.
\end{proof}

\begin{rrule}\label[rrule]{rule:4}
Let $(\AT, \BT, C_1, \ldots, C_r)$ be a normalized partition of $V(G)$, and let $C_i, C_j$, $i < j$, be two cliques such that $C_i$ is a large clique and $C_j$ is either a large clique or an edge clique. If there are at least two edges between $C_i$ and $C_j$, then one of the following reductions, considered in the listed order, is applicable:
\begin{compactenum}[{Case} (1)]
\item There are two edges $uu'$ and $vv'$, where $u, v \in C_i$ and $u', v' \in C_j$, such that $u \neq v$ and $u' \neq v'$. Let $w \in C_i$ be such that $w \notin \{u, v\}$ (note that $w$ exists because $|C_i| \geq 3$). Set $\AT= \AT \cup \{w\}$.

\item $N(C_j) \cap C_i =\{v\}$. Set $\BT= \BT \cup \{v\}$.
\end{compactenum}

\end{rrule}
\begin{proof}[Correctness Proof]
We first prove that either case (1) or case (2) applies. Suppose that case (1) does not apply, and we show that case (2) does.

By maximality of $C_i$ in $\bigcup_{j \ge i} C_j$ (property (ii) in the definition of a nice clique decomposition), no vertex in $C_j$ can be adjacent to all vertices in $C_i$. It follows from this fact and from the inapplicability of \cref{rule:3} that each vertex in $C_j$ has at most one neighbor in $C_i$.  Since case (1) does not apply, the vertices in $C_j$ that have a neighbor in $C_i$ must all have the same neighbor, which proves that case (2) applies.

Now suppose that case (1) applies, and we will show the correctness of the reduction rule in this case. Let $(A, B)$ be any valid monopolar partition of $(G, k)$. Since at most one of $u', v'$ can be in $B$, at least one of $u', v'$, say $u'$, is in $A$. Suppose, to get a contradiction, that $w \in B$. Then both $u$ and $v$ must be in $A$. By maximality of $C_i$ in $\bigcup_{j \ge i} C_j$, $u'$ cannot be adjacent to all vertices in $C_i$. Since~\cref{rule:3} is not applicable, $u$ must be the only neighbor of $u'$ in $C_i$. But then $(v, u, u')$ is an induced $P_3$ in $A$, contradicting that $A$ is a cluster graph.

Suppose that case (2) applies, and suppose to get a contradiction that $v \in A$ in some valid monopolar partition $(A, B)$ of $(G, k)$. Since there are at least two edges between $C_i$ and $C_j$, $v$ has at least two neighbors $u', v' \in C_j$. Again, observe that at least one of $u', v'$, say $v'$, must be in $A$. Since $|C_i| \geq 3$, at least one vertex in $C_i$, say $w$, must be in $A$. Since $v$ is the only neighbor of $v'$ in $C_i$ by the premise of case (2), it follows that $(w, v, v')$ is an induced $P_3$ in $A$, contradicting that $A$ is a cluster graph.
\end{proof}
After applying Reduction Rules~\ref{rule:3} and~\ref{rule:4}, it holds that no two large cliques and no large clique and an edge clique can merge into a single cluster. Hence, any large clique will create a new cluster in~$A$, and edge cliques cannot join any large clique in~$A$. Based on that, we can now upper bound the number of large cliques and edge cliques in a yes-instance of the problem:

\begin{rrule}
\label[rrule]{rule:5}
Let $(G, k)$ be an instance of \monopolar, and suppose that $(\AT, \BT, C_1, \ldots, C_r)$ is a normalized partition of $V(G)$. If the number of large cliques among $C_1, \ldots, C_r$ is more than $k$, or the number of large cliques plus the number of edge cliques is more than $2k$, then reject the instance $(G, k)$.
\end{rrule}
\begin{proof}[Correctness Proof]
Let $(A, B)$ be any monopolar partition of $V(G)$. Since a large clique $C$ has size at least 3, at least $|C| -1 \geq 2$ vertices from $C$ must belong to the same cluster in $A$. By \cref{rule:4}, the number of edges between any large clique and any other large or edge clique is at most 1. It follows from the aforementioned statements that two vertices from two different large cliques, or from a large clique and an edge clique, must belong to different clusters in $A$. Consequently, if the number of large cliques in $(C_1, \ldots, C_r)$ is more than $k$, then for any monopolar partition $(A, B)$ of $G$, the number of clusters in $A$ is more than $k$, and hence $(G, k)$ is a no-instance of \monopolar.

Suppose now that the number of large cliques in $(C_1, \ldots, C_r)$ is $\ell \leq k$, and that the number of edge cliques is $\ell'$. From above, for any monopolar partition $(A, B)$, no vertex from an edge clique can belong to a cluster in $A$ containing a vertex from a large clique. Let $C_i$ and $C_j$, $i < j$, be any two edge cliques.  Since $B$ is an independent set, at least one vertex from each edge clique must be in $A$. By property (iii) of a nice decomposition, no cluster in $A$ can contain three vertices from three different edge cliques in $(C_1, \ldots, C_r)$. It follows from the aforementioned two statements that the number of clusters in $A$ that contain vertices from edge cliques in $(C_1, \ldots, C_r)$ is at least $\ell'/2$. Now the set of clusters in $A$ containing vertices from large cliques is disjoint from that containing vertices from edge cliques, and hence the number of clusters in $A$ is at least $\ell + \ell'/2$. If the number of large cliques plus the number of edge cliques is more than $2k$, then $\ell+\ell' > 2k$, and hence $\ell + \ell'/2 \geq \ell/2 + \ell'/2 > k$. This means that for any monopolar partition $(A, B)$ of $G$, the number of clusters in $A$ is more than $k$. It follows that $(G, k)$ is a no-instance of \monopolar.
\end{proof}

Next, we sanitize the connections between already determined clusters in~$\AT$ and the remaining cliques in the normalized partition.

\begin{rrule}\label[rrule]{rule:6}
Let $(\AT, \BT, C_1, \ldots, C_r)$ be a normalized partition of $V(G)$, let $C$ be a cluster in $\AT$, and let $C_i$, $i \in [r]$, be a large clique. If $v \in C_i$ is such that: (1) $v$ is the only vertex in $C_i$ that is adjacent to $C$, or (2) $v$ is the only vertex in $C_i$ that is not adjacent to $C$, then set $\BT= \BT \cup \{v\}$.
\end{rrule}
\begin{proof}[Correctness Proof]
To prove the correctness of the reduction rule in case (1) holds, suppose that $v$ is the only vertex in $C_i$ that is adjacent to $C$. Let $(A, B)$ be any monopolar partition of $G$. Let $w$ be any vertex in $C$ that is adjacent to $v$.  Since $C_i$ is a large clique, there exists a vertex $u \in C_i$, with $u \neq v$, such that $u \in A$. Since $v$ is the only vertex in $C_i$ that is adjacent to $C$, $u$ is not adjacent to $w$. Now if $v$ were in $A$, then since $C \subseteq A$ and hence $w \in A$, $(u, v, w)$ would be an induced $P_3$ in $A$, contradicting that $A$ is a cluster graph. It follows that $v \in B$ for any monopolar partition $(A, B)$ of $G$.

\looseness=-1 To prove the correctness of the reduction rule in case (2) holds, suppose that $v$ is the only vertex in $C_i$ that is not adjacent to $C$. Let $(A, B)$ be any monopolar partition of $G$.   Since $C_i$ is a large clique, there exists a vertex $u \in C_i$, with $u \neq v$, such that $u \in A$. Since $v$ is the only vertex in $C_i$ that is not adjacent to $C$, $u$ is adjacent to some vertex $w \in C$. Now if $v$ were in $A$, then since $C \subseteq A$ and hence $w \in A$, $(v, u, w)$ would be an induced $P_3$ in $A$, contradicting that $A$ is a cluster graph. It follows that $v \in B$ for any monopolar partition $(A, B)$ of $G$.
\end{proof}
Suppose that none of the above reduction rules applies~to the instance $(G, k)$. Then, the following lemma holds:
\begin{lemma}\label[lemma]{lem:mergeable}
Let $(\AT, \BT, C_1, \ldots, C_r)$ be a normalized partition of $V(G)$, let $C$ be a cluster in $\AT$, and let $C_i$, $i \in [r]$, be a large clique such that $C_i$~is adjacent to $C$. If~$G$ admits a monopolar partition, then $C \cup C_i$~induces a clique in~$G$.
\end{lemma}
\begin{proof}
Suppose, to get a contradiction, that $C \cup C_i$ does not induce a clique, and hence, there exists a vertex $x_i \in C_i$ such that $x_i$ is not adjacent to some vertex in $C$. Since $C$ and $C_i$ are adjacent, there exist vertices $y_i \in C_i$ and $v \in C$ such that $v$ and $y_i$ are adjacent.  Since~\cref{rule:0.5} is not applicable, $x_i$ is not adjacent to any vertex in $C$, and $y_i$ is adjacent to every vertex in $C$. Since cases (1) and (2) of~\cref{rule:6} are not applicable, there exist vertices $y'_i \neq y_i$ and $x'_i \neq x_i$ in $C_i$ such that $y'_i$ is adjacent to $C$ and $x'_i$ is not adjacent to $C$. Since~\cref{rule:0.5} is not applicable, $y'_i$ is adjacent to every vertex in $C$. Now for any monopolar partition $(A, B)$ of $G$, since $B$ is an independent set, at least one vertex $w \in \{y_i, y'_i\}$ is in $A$, and at least one vertex of $u \in \{x_i, x'_i\}$ is in $A$. But then $(v, w, u)$ is an induced $P_3$ in $A$, contradicting that $A$ is a cluster graph.
\end{proof}

The above structure allows us to simplify the instance by shrinking
already determined clusters in~$\AT$.

\begin{rrule}\label[rrule]{rule:8}
Let $(G, k)$ be an instance of \monopolar, and let the tuple $(\AT, \BT, C_1, \ldots, C_r)$ be a normalized partition of $V(G)$. If either (1) $\BT$ contains more than $k+1$ vertices or (2) there exists a cluster in $\AT$ that is not a singleton, then reduce the instance $(G, k)$ to an instance~$(G',k)$ with~$G'$ constructed as follows. Let $V(G') = V_1 \cup V_2 \cup V_3$, where $V_1= \{u_{C} \mid C \ \mbox{is a cluster in}\  \AT\}$, $V_2=\{v_1, \ldots, v_{k+1}\}$, and $V_3= C_1 \cup \cdots \cup C_r$; and $E(G') = \{vu_C \mid v \in V_2 \wedge u_C \in V_1 \} \cup \{vu_C \mid v \in V_3  \wedge u_C \in V_1 \wedge v \ \mbox {is adjacent to C}\}$. That is, $G'$ is constructed from $G$ by introducing $k+1$ new vertices, replacing each cluster $C$ in $\AT$ (if any) by a single vertex $u_C$ whose neighborhood is the neighborhood of $C$ in $C_1, \ldots, C_r$ plus the $k+1$ new vertices, and keeping $C_1, \ldots, C_r$ the~same.
\end{rrule}
\begin{proof}[Correctness Proof]
To prove the correctness of the reduction rule, we need to show that $(G, k)$ is a yes-instance of \monopolar~if and only if $(G', k)$ is. First, observe that by \cref{rule:0.1}, no vertex in $C_1 \cup \cdots \cup C_r$ is adjacent to any vertex in $\BT$.

If $\AT = \emptyset$, then the reduction rule consists of removing the vertices in $\BT$ from~$G$, and replacing them with $k+1$ isolated vertices $v_1, \ldots, v_{k+1}$. Since $\AT = \emptyset$ and no vertex in $C_1 \cup \cdots \cup C_r$ is adjacent to any vertex in $\BT$, the vertices in $\BT$ are isolated vertices in $G$. Therefore, the reduction rule in this case essentially consists of removing some isolated vertices from~$\BT$ and~$G$, and thus is obviously correct.

Assume now that $\AT \neq \emptyset$. It is easy to see that if $(G, k)$ is a yes-instance of \monopolar~then so is $(G', k)$. This can be seen as follows. If $(A, B)$ is a valid monopolar partition of $(G, k)$, then the above reduction rules guarantee that $\AT \subseteq A$, and hence each cluster of $\AT$ must be a subset of a single cluster in $A$. If we (i) remove the vertices in $\BT$ and add $k+1$ vertices to $B$ that induce an independent set, and (ii) replace each cluster $C$ in $\AT$ by a single vertex $u_C$ connected to the $k+1$ new vertices in $B$ and to the vertices of the cluster that $C$ belongs to $A$, we still get a valid monopolar partition of $G$.

To prove the converse, suppose that $(G', k)$ is a yes-instance of \monopolar, and let $(A', B')$ be a valid monopolar partition of $V(G')$. Since $(A', B')$ is a valid monopolar partition of $V(G')$, and every vertex $u_C$, $C$ is a cluster in $\AT$, is adjacent to the $k+1$ independent vertices $v_1, \ldots, v_{k+1}$, we have $u_C \in A'$ for every cluster $C$ in $\AT$, and $\{v_1, \ldots, v_{k+1}\} \subseteq B'$. Let $B = B' \setminus \{v_1, \ldots, v_{k+1}\} \cup \BT$. Since (1) $B'$ induces an independent set, (2) every vertex $u_C$, $C$ is a cluster in $\AT$, is in $A'$, and (3) no vertex in $C_1 \cup \cdots \cup C_r$ is adjacent to any vertex in $\BT$, it follows that $B$ is an independent set. Let $A$ be the set of vertices obtained from $A'$ by replacing each vertex $u_C$ by the vertices in the cluster $C$ in $\AT$. We claim that $A$ is a cluster graph with at most $k$ clusters. Suppose that a vertex $u_C$ is replaced in $A'$ by the vertices in cluster $C$ in $\AT$; assume that $u_C$ belongs to cluster $C'$ in $A'$. Each vertex in $C'$, other than $u_C$, must be a vertex in $V_3= C_1 \cup \cdots \cup C_r$. Let $v' \in C' \setminus \{u_C\}$ be chosen arbitrarily. Since $v'$ and $u_C$ belong to the same cluster $C'$, by definition of $G'$, $v'$ must be adjacent to $C$ in $G$. By \cref{rule:0.5}, $v'$ must be adjacent to all vertices in $C$. Since $v'$ was an arbitrarily chosen vertex in $C' \setminus \{u_C\}$, $(C' \setminus \{u_C\}) \cup C$ induces a cluster in $A$. It remains to show that no two clusters in $A$ are adjacent. Suppose, to get a contradiction, that this is not the case. Since $A'$ induces a cluster graph, there must exist two vertices $u_{C_1}$ and $u_{C_2}$ in $A'$, that belong to clusters $C'_1$ and $C'_2$ in $A'$, respectively, such that cluster $C_1 \cup (C'_1\setminus\{u_{C_1}\})$ is adjacent to cluster $C_2 \cup (C'_2 \setminus \{u_{C_2}\})$. Since $A'$ is a cluster graph, this implies that either: (1) $C_1$ is adjacent to $C_2$, (2) $C_1$ is adjacent to $C'_2 \setminus \{u_{C_2}\}$, or (3) $C_2$ is adjacent to $C'_1\setminus\{u_{C_1}\}$. This leads to a contradiction in each of the three cases above: (1) would contradict that $\AT$ is a cluster graph (\cref{rule:0}), (2) would imply (by the construction of $G'$) that $u_{C_1}$, and hence $C'_1$, is adjacent to $C'_2$ in $A'$, and (3) would imply (by the construction of $G'$) that $u_{C_2}$, and hence $C'_2$, is adjacent to $C'_1$ in $A'$. It follows from the above that the constructed partition $(A, B)$ is a valid monopolar partition for $V(G)$. Finally, the number of clusters in $A$ is the same as that in $A'$, which is at most $k$.
\end{proof}
If \cref{rule:8} is applied, then after its application, we set $\AT$ to $V_1$  and $\BT$ to $\{v_1, \ldots, v_{k+1}\}$. Observe that this implies that~$|\AT|\le k$ and~$|\BT|\le k+1$.
Note that in any valid monopolar partition $(A, B)$ of the graph resulting from the application of \cref{rule:8}, each vertex in $V_1$ must be in $A$, being adjacent to the $k+1$ independent set vertices $v_1, \ldots, v_{k+1}$, whereas the vertices $v_1, \ldots, v_{k+1}$ can be safely assumed to be in $B$ since their only neighbors are in $V_1 \subseteq A$.
\subsection{Modelling the Pushing Process by a Bipartite Graph}
We have now arrived at a stage where we have bounded the number of large and edge cliques,
and the size of~$\AT$ and~$\BT$. It remains to bound the size of
the large cliques and the number of vertex cliques to obtain a polynomial-size problem
kernel. The challenge here is that we need to identify vertices such that putting them
in~$A$ or~$B$ will eventually, after a series of pushes, lead either to the creation of
too many clusters in~$A$, or to the addition of two adjacent vertices in~$B$. To describe the structure of the
avalanche of pushes to~$A$ or~$B$, we introduce the following auxiliary graph.

\begin{definition}\label[definition]{def:auxiliary}
 For a normalized partition $(\AT, \BT, C_1, \ldots, C_r)$ of $V(G)$, we define the auxiliary bipartite graph $\Lambda$ as follows. The vertex set of $\Lambda$ is $\VL=(V_C, V_I)$, where $V_C$ is the set of all vertices in the large cliques in $C_1, \ldots, C_r$, and $V_I$ is the set of all vertices in the vertex cliques in $C_1, \ldots, C_r$. The edge set of $\Lambda$ is $\EL =\{uv \in E(G) \mid u \in V_C \ \mbox{and} \ v \in V_I\}$; that is, $\EL$ consists of precisely the edges in $E(G)$ that are between $V_C$ and $V_I$.
\end{definition}

Recall that $V_I$ is an independent set in $G$ by Fact~\ref{fact:is}. For a vertex $v \in \VL$, we write  $\N(v):=N(v) \cap \VL$ for the neighbors of $v$ in $\Lambda$. To bound the maximum degree in~$\Lambda$, we apply the following reduction rule. The correctness proof  is straightforward, after recalling that the vertex cliques induce an independent set in $G$ (\cref{fact:is}), and observing that no two vertices of an independent set can belong to the same cluster in a cluster graph:

\begin{rrule}\label[rrule]{rule:1}
Let $(\AT, \BT, C_1, \ldots, C_r)$ be a normalized partition of $V(G)$. If there is a vertex $v \in V(G) \setminus (\AT \cup \BT)$ with more than $k$ neighbors that are vertex cliques, then set $\AT= \AT \cup \{v\}$.
\end{rrule}

After performing all rules up to this point, we have the following lemma:

\begin{lemma}\label[lemma]{lem:lambda}
Let $(\AT, \BT, C_1, \ldots, C_r)$ be a normalized partition of $V(G)$ and consider the auxiliary graph $\Lambda = (\VL=(V_C, V_I), \EL)$. Then the maximum degree of $\Lambda$, $\Delta(\Lambda)$, is at most~$k$.
\end{lemma}

\begin{proof}
For every vertex $v \in V_C$, we have $|\N(v)| \leq k$ because~\cref{rule:1} is inapplicable.
By property (ii) of a nice decomposition and the inapplicability of~\cref{rule:3}, every vertex clique that is adjacent to a large clique $C$ is adjacent to exactly one vertex in $C$. Since by~\cref{rule:5} the number of large cliques is at most $k$, every vertex in $V_I$, which is a vertex clique by definition of $V_I$, has at most $k$ neighbors in $V_C$. Therefore, for every vertex $v \in V_I$, we have $|\N(v)| \leq k$.
\end{proof}

Using the following lemma, we now observe that the auxiliary
graph~$\Lambda$ captures some of the avalanches emanating from
vertices in large or vertex cliques. Namely, pushing a vertex~$v$ in a
large clique to $B$ (or in a vertex clique to $A$) will also require
pushing each vertex reachable (in the auxiliary graph) from~$v$ from
$A$ to $B$ or vice versa.

For two vertices $u, v \in \VL$, write $\D(u, v)$ for the length of a shortest path between $u$ and $v$ in $\Lambda$. For a vertex $v \in \VL$ and $i \in \{0, \ldots, n\}$, define $N^i(v) =\{u \in \VL \mid \D(u,v)=i\}$.
Write $\bar{0}_n$ for the set of even integers in $\{0, \ldots, n\}$, and $\bar{1}_n$ for the set of odd integers in $\{0, \ldots, n\}$.

\begin{lemma}\label[lemma]{lem:alternating}
Let $(\AT, \BT, C_1, \ldots, C_r)$ be a normalized partition of $V(G)$, let $\Lambda = (\VL, \EL)$ be the associated auxiliary graph where $\VL = (V_C, V_I)$, and let $(A, B)$ be any valid monopolar partition of~$G$.
\begin{compactenum}[(i)]
\item For any vertex $v \in V_C$: If $v \in B$ then $\N(v) \subseteq A$.
\item For any vertex $v \in V_I$: If $v \in A$ then $\N(v) \subseteq B$.
\item For any vertex $v \in V_C$: If $v \in B$ then $\N^i(v) \subseteq B$ for $i \in \bar{0}_n$, and $\N^i(v) \subseteq A$ for $i \in \bar{1}_n$.
\item For any vertex $v \in V_I$: If $v \in A$ then $\N^i(v) \subseteq A$ for $i \in \bar{0}_n$, and $\N^i(v) \subseteq B$ for $i \in \bar{1}_n$.
\end{compactenum}
\end{lemma}

\begin{proof}
  (i): This trivially follows because $B$ is an independent set.

(ii): Suppose that $v \in V_I$ is in $A$, and let $u \in \N(v)$. Then $u \in V_C$ because $\Lambda$ is bipartite, and hence, by definition, $u$ belongs to a large clique $C_i$ for some $i \in [r]$. Suppose, to get a contradiction, that $u \in A$.
            Since $C_i$ is a large clique, and hence $|C_i| \geq 3$, there exists a vertex $w \neq u$ in $C_i$ such that $w \in A$.  By property (ii) of the nice decomposition $(C_1, \ldots, C_r)$ and the inapplicability of~\cref{rule:3}, $v$ is not a neighbor of $w$ in $G$. But this implies that $(v, u, w)$ is an induced $P_3$ in $A$, contradicting that $A$ is a cluster graph. It follows that $\N(v) \subseteq B$.

            (iii): This follows by repeated alternating applications of (i) and (ii) above.

(iv): This follows by repeated alternating applications of (ii) and (i) above.
\end{proof}

The above lemma about the avalanches captured by the auxiliary graph
allows us to identify vertices whose push to one side of the partition
would lead to avalanches that, in turn, would lead to too many clusters in~$A$ or to
two adjacent vertices in~$B$. We can hence fix them in the
corresponding part.

\begin{rrule}\label[rrule]{rule:7}
Let $(\AT, \BT, C_1, \ldots, C_r)$ be a normalized partition of $V(G)$, and let $\Lambda = (\VL=(V_C, V_I), \EL)$ be the associated auxiliary graph.
\begin{compactenum}[(i)]
\item For any vertex $v \in V_C$: If either $\bigcup_{i \in \bar{0}_n} \N^i(v)$ contains two adjacent (in $G$) vertices or we have $|\bigcup_{i \in \bar{1}_n} \N^i(v)| > k$, then set $\AT = \AT \cup \{v\}$.
\item For any vertex $v \in V_I$: If either $|\bigcup_{i \in \bar{0}_n} \N^i(v)| > k$ or $\bigcup_{i \in \bar{1}_n} \N^i(v)$ contains two adjacent (in $G$) vertices, then set $\BT = \BT \cup \{v\}$.
\end{compactenum}
\end{rrule}

\begin{proof}[Correctness Proof]
  (i) Let $v \in V_C$, and suppose that either $\bigcup_{i \in \bar{0}_n} \N^i(v)$ contains two adjacent vertices or $|\bigcup_{i \in \bar{1}_n} \N^i(v)| > k$. If $v \in B$ for any valid partition $(A, B)$ of $G$, then by part (iii) of \cref{lem:alternating}, it would follow that $\bigcup_{i \in \bar{0}_n} \N^i(v) \subseteq B$ and $\bigcup_{i \in \bar{1}_n} \N^i(v) \subseteq A$. In either case this contradicts that $(A, B)$ is valid partition of $G$: If $\bigcup_{i \in \bar{0}_n} \N^i(v)$ contains two adjacent vertices, then $B$ is not an independent set, and if $|\bigcup_{i \in \bar{1}_n} \N^i(v)| > k$ then $A$ contains more than $k$ clusters since $\bigcup_{i \in \bar{1}_n} \N^i(v)$ induces an independent set in $G$.

  \ifshort%
  (ii) The proof follows along the same lines as the proof of~(i) ($\spadesuit$).
  \else%
  (ii) Let $v \in V_I$, and suppose that either $|\bigcup_{i \in \bar{0}_n} \N^i(v)| > k$ or $\bigcup_{i \in \bar{1}_n} \N^i(v)$ contains two adjacent vertices. If $v \in A$ for any valid partition $(A, B)$ of $G$, then by part (iv) of \cref{lem:alternating}, it would follow that $\bigcup_{i \in \bar{0}_n} \N^i(v) \subseteq A$ and $\bigcup_{i \in \bar{1}_n} \N^i(v) \subseteq B$. In either case this contradict that $(A, B)$ is valid partition of $G$: If $|\bigcup_{i \in \bar{0}_n} \N^i(v)| > k$ then $A$ contains at more than $k$ clusters since the vertices in $\bigcup_{i \in \bar{0}_n} \N^i(v)$ induce an independent set in $G$, and if $\bigcup_{i \in \bar{1}_n} \N^i(v)$ contains two adjacent vertices, then $B$ is not an independent set.
  \fi
\end{proof}
We are now ready to define a set of representative vertices which already capture the remaining structure of avalanches in the instance.
\begin{definition}
\label[definition]{def:important}
Let $(\AT, \BT, C_1, \ldots, C_r)$ be a normalized partition of $V(G)$, and let $\Lambda = (\VL=(V_C, V_I), \EL)$ be the associated auxiliary graph. From each large clique $C_i$, $i \in [r]$, fix three vertices $u_i, v_i, w_i$, and define $V_{\fixed} = \{u_i, v_i, w_i \mid C_i \ \mbox{is a large clique} \}$ to be the set of all fixed vertices. Define \\ $V_{\edge} =\{C_i \mid C_i \ \mbox{is an edge clique}\}$ to be the set of vertices of the edge cliques, define $N_{\edge} = N(V_{\edge}) \cap \VL$ to be the neighbors of $V_{\edge}$ in $G$ that are also in $\VL$, and define $N_{\edge}^{\cup} = \bigcup_{v \in N_{\edge}} \bigcup_{i \leq n}\N^i(v)$ to be the set of all vertices in $\VL$ that are reachable in $\Lambda$ from the vertices in $N_{\edge}$.
(Note that $N_{\edge} \subseteq N_{\edge}^{\cup}$.) Define $V_{\inter} =\{u, v \mid u \in C_i \wedge  v \in C_j \wedge  i \neq j \wedge uv \in E(G)  \wedge (C_i, C_j \ \mbox{are large cliques})\}$ to be the set of endpoints of edges between large cliques, and define $N_{\inter}^{\cup} = \bigcup_{v \in V_{\inter}} \bigcup_{i \leq n}\N^i(v)$ to be the set of all vertices in $\VL$ that are reachable in $\Lambda$ from the vertices in $V_{\inter}$. (Note that $V_{\inter} \subseteq N_{\inter}^{\cup}$.) Finally, let $V_{\important} = \AT \cup \BT \cup V_{\fixed} \cup N_{\inter}^{\cup} \cup V_{\edge} \cup N_{\edge}^{\cup}$.
\end{definition}

The next reduction rule shrinks the instance to the set of representative vertices defined above.

\begin{rrule}
\label[rrule]{rule:9}
Let $(G, k)$ be an instance of \monopolar, and let the tuple $(\AT, \BT, C_1, \ldots, C_r)$ be a normalized partition of $V(G)$. Let $V_{\important}$ be as defined in \cref{def:important}. Set $G=G[V_{\important}]$.
\end{rrule}
\begin{proof}[Correctness proof]
To prove the correctness of the reduction rule, let $G'=G[V_{\important}]$ and we need to show that the two instances $(G, k)$ and $(G', k)$ of \monopolar~are equivalent. Since $G'$ is a subgraph of $G$ and the property of
having a valid monopolar partition is a hereditary graph property, it follows that if $(G, k)$ is a yes-instance of \monopolar~then so is $(G', k)$. Therefore, it suffices to show the converse, namely that if $(G', k)$ is a yes-instance of \monopolar~then so is $(G, k)$.

Suppose that $(G', k)$ is a yes-instance of \monopolar, and let $(A, B)$ be a valid monopolar partition of $(G', k)$. Let $(\AT, \BT, C_1, \ldots, C_r)$ be the normalized partition of $V(G)$ with respect to which $V_{\important}$, and hence $G'=G[V_{\important}]$, were defined, and let $V_{\fixed}, V_{\inter}, N_{\inter}^{\cup}, V_{\edge}, N_{\edge}^{\cup}$ be as in \cref{def:important}.  Let $v$ be an arbitrary vertex in $V(G)\setminus V(G')$. It suffices to show that $G[V(G') \cup \{v\}]$ has a valid monopolar partition, as we can repeatedly add vertices, one after the other, and the same proof applies.  Since the set of vertices forming the edge cliques, $V_{\edge}$, is a subset of $V(G')$, and $\AT \cup \BT \subseteq V(G')$, vertex $v$ is either a vertex of a large clique of $G$ that is not in $V_{\fixed}$, or $v$ is a vertex clique in $G$. We distinguish these two cases.

\medskip\noindent {\bf Case 1.} $v \in V(C_i) \setminus V_{\fixed}$, for some large clique $C_i$, where $i \in [r]$. Since three vertices from $C_i$ are in $V_{\fixed}$, at least two of these vertices must belong to a cluster, say $C'_i$, in part $A$ of the partition $(A, B)$. Note that since $\AT \subseteq A$, if $C_i$ has a neighbor in $\AT$, which must be a neighbor of all the vertices in $C_i$, including $v$, by \cref{lem:mergeable}, then this neighbor must be  in $C'_i$. We first claim that $C'_i \cup \{v\}$ is a clique. Observe that since $C'_i$ contains two vertices from $V_{\fixed}$, and hence from $C_i$, by~\cref{rule:4}, $C'_i$ does not contain any vertices from a large clique other than $C_i$ or from an edge clique. Moreover, by property (ii) of the nice decomposition $(C_1, \ldots, C_r)$ and~\cref{rule:3}, $C'_i$ does not contain any vertex clique. Therefore, $C'_i$ consists only of vertices in $C_i$, plus possibly a single vertex in $\AT$ that must be adjacent to all the vertices in $C_i$. Since $v \in C_i$, it follows that $C'_i \cup \{v\}$ is a clique.

Let $S$ be the set of vertex cliques in $A$, and note that $S$ is an independent set. Define the following layered structure.  The root of this structure is $v$. The first layer contains the set of vertices (possibly empty), denoted $N_1(v)$, that are the neighbors of $v$ in $S$, that is, $N_1(v) =N(v) \cap S$; and the second layer contains the set of vertices, denoted $N_2(v)$, that are the neighbors in $B$ of the vertices of $N_1(v)$, that is $N_2(v) = N(N_1(v)) \cap B$. For $i > 2$, layer $i$ contains the set of vertices $N_i(v) = N(N_{i-1}(v)) \cap S$ if $i$ is odd, and the set of vertices $N_i(v) = N(N_{i-1}(v)) \cap B$ if $i$ is even. Let $N_0(v) = \{v\}$. We claim that the partition $(A', B')$ obtained from $(A, B)$ by placing $v$ in $A$, moving the vertices in $N_i(v)$ for even $i \geq 2$ from $B$ to $A$, and moving the vertices in $N_i(v)$ for odd $i$ from $A$ to $B$, is a valid monopolar partition; that is, $(A', B')$, where $A'= (A \cup \bigcup_{i \in \bar{0}_n}N_i(v)) \setminus \bigcup_{i \in \bar{1}_n}N_i(v)$ and $B' = (B \cup \bigcup_{i \in \bar{1}_n}N_i(v)) \setminus \bigcup_{i \in \bar{0}_n}N_i(v)$ is a valid monopolar partition of $G[V(G') \cup \{v\}]$. Since $S$ is an independent set, so is $\bigcup_{i \in \bar{1}_n}N_i(v) \subseteq S$. Since the set of neighbors of $\bigcup_{i \in \bar{1}_n}N_i(v)$ is precisely $\bigcup_{i \in \bar{0}_n}N_i(v)$ and $B$ is an independent set, $B' = (B \cup \bigcup_{i \in \bar{1}_n}N_i(v)) \setminus \bigcup_{i \in \bar{0}_n}N_i(v)$ is an independent set as well. Therefore, to show that $(A', B')$ is a valid monopolar partition, it suffices to show that $A'$ is a cluster graph of at most $k$ clusters.

First, we claim that each vertex in $N_{\even} = \bigcup_{i \in \bar{0}_n}N_i(v)$ belongs to a large clique in $C_1, \ldots, C_r$. This is certainly true for the vertex $v$, which is in $C_i$, where $C_i$ is a large clique.
Now for any other vertex $u \in N_{\even}$, by construction, $u$ is the neighbor of a vertex clique in $S$. Since the set of all vertex cliques induces an independent set, $u$ itself cannot be a vertex clique, being a neighbor of a vertex clique. Since $u \in N_i(v)$, for some $i$, and hence $v$ is reachable from $u$, $u$ cannot be an endpoint of an edge clique; otherwise, $v$ would belong to $N_{\edge}^{\cup}$ and hence, would belong to $V_{\important}$. Since $\AT \subseteq A$, and no vertex in $V(G) \setminus (\AT \cup \BT)$ is adjacent to a vertex in $\BT$ by \cref{rule:0.1}, $u \notin \AT \cup \BT$. It follows from the preceding that each vertex in $N_{\even}$ belongs to a large clique in $C_1, \ldots, C_r$. From each large clique in $C_1, \ldots, C_r$, at least two fixed vertices are in $A'$; denote by $C'_j$ the cluster in $A'$ that contains the two fixed vertices from a large clique $C_j$.  As shown at the beginning of {\bf Case 1} about $C'_i$, the same holds true for any $C'_j$: $C'_j$ consists of a subset of $C_j$, plus possibly a vertex in $\AT$ that is adjacent to all vertices in $C_j$. Now add each vertex in $N_{\even}$ that belongs to a (large clique) $C_j$ in $G$ to the corresponding cluster $C'_j$ in $A'$. We claim that the resulting partition is a valid monopolar partition. Since each vertex $u$ in $N_{\even}$ was added to the cluster $C'_j$ such that $u \in C_j$ and $C'_j$ consists of a subset of $C_j$ plus possible a vertex in $\AT$ that is adjacent to all vertices in $C_j$, $C'_j \cup \{u\}$ is a clique. Moreover, since each vertex $u \in N_{\even}$ was added to an existing cluster, this addition does not increase the number of clusters in $A$, and hence, $A'$ has at most $k$ clusters. It remains to show that this addition does not create an edge between two different clusters. Suppose that this is not the case. Since $N_{\even} \subseteq B$ is an independent set, this implies that there exists a vertex $u \in N_{\even}$ that is added to a cluster $C'_j$ in $A$ such that $u$ is adjacent to some vertex $w$ in $A$. Since all the neighbors of $u$ that are vertex cliques are in
$\bigcup_{i \in \bar{1}_n}N_i(v) \subseteq B'$, $w$ is not a vertex clique. Since $v$ is reachable from $u$, and hence from $w$, and $v \notin N_{\edge}^{\cup}$, $w$ cannot be a vertex of an edge clique. By the same token, since $v$ is reachable from $w$ and $v \notin  N_{\inter}^{\cup}$, $w$ cannot be a vertex of a large clique. Finally, $w$ cannot be in $\BT$ because no vertex in $V(G) \setminus (\AT \cup \BT)$ is adjacent to a vertex in $\BT$, and $w$ cannot be in $\AT$ because $w$ would be adjacent to all vertices of $C'_j$. This completes the proof of {\bf Case 1}.

\medskip\noindent {\bf Case 2.} $v$ is a vertex clique.  The treatment of this case is very similar to {\bf Case 1}. We define $N_{0}(v) = \{v\}$, $N_i(v) = N(N_{i-1}(v)) \cap B$ if $i \geq 1$ is odd, and $N_i(v) = N(N_{i-1}(v)) \cap S$ if $i \geq 2$ is even, where $S$ is the set of vertex cliques in $A$. It can then be shown---using very similar arguments to those made in {\bf Case 1}---that the partition $(A', B')$, where $A'= (A \cup \bigcup_{i \in \bar{1}_n}N_i(v)) \setminus \bigcup_{i \in \bar{0}_n}N_i(v)$ and $B' = (B \cup \bigcup_{i \in \bar{0}_n}N_i(v)) \setminus \bigcup_{i \in \bar{1}_n}N_i(v)$ is a valid monopolar partition of $G[V(G') \cup \{v\}]$. The proof is omitted to avoid repetition.
\end{proof}
With this reduction rule we have finally bounded the size of~$V(G)\setminus \AT\cup \BT$ and may now give the polynomial kernel whose existence was promised in \cref{thm:main}.

{
\begin{theorem}
\label{thm:mainkernel}
\monopolar~has a polynomial kernel with at most $9k^4 + 9k+1$ vertices which can be computed in~$\Oh(n^2m)$ time.
\end{theorem}
}
\begin{proof}
  Given an instance $(G, k)$ of \monopolar, we apply Reduction Rules \ref{rule:0}--\ref{rule:9} exhaustively to $(G, k)$. Clearly, the above rules can be applied in polynomial time. Let $(G', k')$ be the resulting instance, let $(\AT, \BT, C_1, \ldots, C_r)$ be a normalized partition of $V(G')$ with respect to which none of Reduction Rules~\ref{rule:0}--\ref{rule:9} applies, and let $\Lambda = (\VL=(V_C, V_I), \EL)$ be the auxiliary graph. Note that, by \cref{rule:9}, $V(G') =V_{\important}= \AT \cup \BT \cup V_{\fixed}  \cup N_{\inter}^{\cup} \cup V_{\edge} \cup N_{\edge}^{\cup}$. By \cref{rule:5}, the number of large cliques is at most $k$, and the number of edge cliques is at most $2k$. It follows that $|V_{\fixed}| \leq 3k$ and $|V_{\edge}| \leq 4k$. For a vertex $v \in V_{\edge}$, by \cref{rule:1}, $v$ has at most $k$ neighbors in $V_I$. Moreover, by \cref{rule:4}, $v$ can have at most $k$ neighbors in $V_C$, and therefore, $|\N(v)| \leq 2k$, and $|N_{\edge}| \leq 4k \cdot 2k = 8k^2$. Since \cref{rule:7} does not apply and $\Delta(\Lambda) \leq k$ by \cref{lem:lambda}, we have that, for any $v \in \VL$, we have $|\bigcup_{i \leq n}\N^i(v)| \leq \Delta(\Lambda) \cdot k \leq k^2$. This implies that $|N_{\edge}^{\cup}| \leq 8k^2 \cdot k^2 \leq 8k^4$. Now since the number of large cliques is at most $k$, by \cref{rule:4}, it follows that $|V_{\inter}| \leq {k \choose 2} < k^2$.  Since for a vertex $v \in \VL$ we have $|\bigcup_{i \leq n}\N^i(v)| \leq k^2$ as argued above, it follows that $|N_{\inter}^{\cup}| \leq k^4$. Since $|\AT| \leq k$ and $|\BT| \leq k+1$, putting everything together, we conclude that the number of vertices in $V(G')$, $|V_{\important}|$, is at most $k+k+1 + 3k+k^4 +4k+8k^4 \leq 9k^4 + 9k+1$. Hence, the number of edges in~$G'$ is~$\Oh(k^8)$ and thus the kernel has polynomial size. \ifshort
 The running time proof is deferred \app.%
 \else

It remains to show the running time. First, observe that the overall number of applications of the reduction rules
is~$\Oh(n)$, since each application either moves a vertex to~$\AT$ or~$\BT$, or
reduces the number of vertices in~$G$. To obtain the overall running time bound, we first
bound the time to check the applicability of each reduction rule.

For
Reduction Rules~\ref{rule:0}--\ref{rule:1}, it is obvious that their applicability can be checked
in~$\Oh(m)$ time (recall that we assume $n \in \Oh(m)$).

For \cref{rule:3}, its applicability can be checked in~$\Oh(m)$ time, if we
assign to each vertex a label indicating the number of its clique and an additional
``counter-variable'' for each cluster. Then, we consider the vertices of the clique
decomposition one by one. When considering a vertex~$v$, we reset all clique counters
to 0. Then we scan through the adjacency list of~$v$, and increase the counter
of a clique~$C_i$ for each edge between~$v$ and~$C_i$ (the cluster for each edge can be
checked in~$\Oh(1)$ time using the clique labels of the vertices). After scanning through
the adjacency list, we check for each clique~$C_i$ that $v$ is adjacent to, if the
number of edges between $v$ and $C_i$ and the size of~$C_i$ meet the conditions in \cref{rule:3}.

For \cref{rule:4}, we create once in~$\Oh(n^2)$ time an~$n'\times n'$ matrix~$M$ where $n'$ is the number of large and edge cliques.
Entry~$M[i,j] = M[j, i]$ is used to count the number of edges between the $i$th large or edge clique~$C_i$ and the $j$th large or edge clique~$C_j$. Before and
after we test the applicability of the rule,~$M[i,j]=0$ for all~$i,j\in [n']$. To test applicability, we scan
through a list containing each edge of~$G$ exactly once and increment~$M[i,j]$ each time
we encounter an edge between~$C_i$ and~$C_j$. If at some point in this procedure~$M[i,j]=2$ for some~$i$ and~$j$, then the
rule applies. After the check, we reset~$M$ to~$0$ by keeping a list of all pairs of
modified matrix indices.

It is clear that we can check in $\Oh(m)$ time whether \cref{rule:5} applies.

\cref{rule:6} can be checked in a similar manner to \cref{rule:4}. We use a $n_1 \times n_2$ matrix~$M'$ where $n_1$ is the number of large cliques in the current normalized partition and $n_2$ is the number of clusters in~$\AT$. We use entry~$M[i, j]$ to count the number of vertices in large clique $C_i$ adjacent to the $j$th cluster~$D_j$ in~$\AT$. After a one-time $\Oh(n^2)$ initialization, we will ensure that before and after the test of applicability~$M[i, j] = 0$ for all $i \in [n_1]$, $j \in [n_2]$. Additionally, we use a vertex labeling for all vertices in~$G$, which we initialize for every vertex as \emph{uncounted}. We iterate in $\Oh(m)$ time over the list of edges in~$G$ and whenever we encounter an edge such that one endpoint~$v$ of which is uncounted and in large clique~$C_i$, and the other endpoint is in the $j$th cluster in $\AT$, then we increment~$M[i, j]$ and remove the labeling from~$v$. If, after processing some edge, $M[i, j]$ now equals 1, or $|C_i| - 1$, then we label~$C_i$ as \emph{amenable} and otherwise remove the amenable-label from~$C_i$ (if any). \cref{rule:6} applies if and only if a large clique is amenable after processing each edge. As before, after the check for applicability, we restore all entries $M[i, j] = 0$ by tracking the pairs of indices which changed during the applicability test.

\cref{rule:8} can obviously be
checked in~$\Oh(m)$ time. For the remaining reduction rules, it is necessary to compute the
auxiliary bipartite graph~$\Lambda$ which can be done in $\Oh(m)$~time by iterating over all edges and checking whether they are incident with a large clique or vertex clique. To check whether \cref{rule:7} applies, it is enough to compute the connected components of~$\Lambda$, and
compute for each component the size of each part and the subgraph of~$G$ that is induced
by each part. This can clearly be done in~$\Oh(m)$ time. For \cref{rule:9}, we first need to compute~$V_{\important}$ in~$\Oh(m)$
time---we iterate over all edges and check whether one of the corresponding conditions applies to the endpoints---and then compute the subgraph~$G[V_{\important}]$ also in~$\Oh(m)$ time.

The time to perform each reduction rule is~$\Oh(m)$, plus the time needed to update
the clique decomposition. We update the clique decomposition~$\Oh(n)$ times, by
\cref{lem:cluster-dec-time}, this takes~$\Oh(nm)$ time. Thus, the latter step has a
total running time of~$\Oh(n^2m)$, which gives the overall running time for computing the
kernel.
\fi
\end{proof}

\section{Kernel-size Lower Bound}

\newcommand{\selcon}{\ensuremath{\mathsf{selection}}}
This section is dedicated to proving the ``only if'' direction of \cref{thm:main}, which, together with \cref{thm:mainkernel}, completes the proof of \cref{thm:main}. More precisely, we prove the following:
{

\begin{theorem}\label{thm:lowerbound}
Let $\Pi$ be a graph property characterized by a (not necessarily finite) set~$\mathcal{H}$ of connected forbidden induced subgraphs, each of order at least~$3$. Then unless \containment, {\picluster} parameterized by the number $k$ of clusters in the cluster graph $G[A]$ does not admit a polynomial kernel.
\end{theorem}
}
\looseness=-1\noindent Throughout, let $\Pi$ be any graph property satisfying the conditions of \cref{thm:lowerbound}.
We show \cref{thm:lowerbound} by giving a cross-composition from the \NP-hard problem \pCIS~\cite{FHRV09}, defined below.
\begin{quote}
  \pCIS\\
  \textbf{Input:} A graph $G=(V,E)$, $k \in \nat$, and a proper $k$-coloring $c \colon V \to \{1, \ldots, k\}$.\\
  \textbf{Question:} Is there an independent set with $k$ vertices in $G$ that contains exactly one vertex of each color?
\end{quote}
In the remainder of this section, we explain the construction behind
the cross-com\-position and prove its correctness. We start by
describing the intuition behind the construction, and why the
avalanches in the case of properties~$\Pi$ as above cannot be
contained.

In contrast to \monopolar, the avalanches caused by the pushing process for the general \picluster\ problem are much more uncontrollable: If some push to the $\Pi$-side~$B$ creates a forbidden induced subgraph~$M$ for $\Pi$ in $G[B]$, we can repair the partition and ``break'' $M$ by pushing a vertex of~$M$ to the cluster graph side~$A$.
The crucial point here is that, because the forbidden induced subgraphs have order at least three, there are at least two possibilities for choosing a vertex from $M$ to push.
Now each distinct push of a vertex of $M$ from~$B$ to~$A$ may lead---through further necessary pushes from~$A$ to~$B$---to distinct forbidden induced subgraphs in~$G[B]$, again with multiple possible ways of breaking them in order to repair the partition.
These avalanches cannot be contained, and lead to many possible paths along which they can be repaired.

It is precisely the above-described behavior of avalanches that we exploit to obtain a cross-composition: The main gadgets select a \pCIS\ instance and independent-set vertices within that instance.
The gadgets are combined in a way that gives a trivial \pipartition\ with one caveat: The overall number of clusters in $G[A]$ is one more than allowed and there is exactly one singleton cluster which can be pushed to the $\Pi$-side~$B$.
We call the vertex of this singleton cluster an {\em activator} vertex since it activates the instance selection gadget.
Pushing the activator vertex into $B$ creates a forbidden induced subgraph for~$\Pi$, requiring further pushes that propagate along a root-leaf path in a binary-tree-like structure.
In the end, exactly one vertex corresponding to a leaf in the instance selection gadget will be pushed from~$A$ to~$B$.
This vertex corresponds to one \pCIS{} instance and its push transmits the choice of this instance to further gadgets that check whether this instance has an independent set of $k$ vertices.

Next, in \cref{sec:nopk-example}, we give an example of a problem instance that shows these binary-tree-like pushes---we later generalize the construction underlying this instance and use it in selection gadgets.
Based on this intuition, we outline the construction in \cref{sec:nopk-outline}.
In \cref{sec:nopk-basicops}, we begin with the concrete description of the cross-composition: We give scaffolds for the construction and some basic operations that we need, such as fixing vertices into one of the two parts of the partition, and invariants that we need to maintain for the correctness proofs.
Afterwards, the construction proceeds in an incremental fashion: In \cref{sec:nopk-instsel,sec:nopk-vertsel}, we show how to construct a selection gadget and then use it to create instance-selection and vertex-selection gadgets, and we add them to the constructed graph one-by-one.
Finally, in \cref{sec:nopk-verif} we construct verification gadgets that ensure that the selected vertices in the selected instance form an independent set. %

\subsection{Example}\label{sec:nopk-example}An example of the instance
Figure~\ref{fig:instance-select} shows an example of the selection gadget for~\picluster\ with~$\Pi$ being the set of cluster graphs.
The forbidden induced subgraph for both~$A$ and~$B$ is the~$P_3$ in this case, and in the example we look for a \pipartition{} with at most~$4$ clusters in~$G[A]$.
Observe that, as highlighted in Figure~\ref{fig:instance-select}, we build an instance with some initial \pipartition~$(A,B)$ where~$G[A]$ has one cluster more than allowed.

In the example, we aim to select one of four instances~$(G_1,k)$, $(G_2,k)$, $(G_3,k)$, and~$(G_4,k)$, and the selection of instance~$(G_i,k)$, for $i \in [4]$, is triggered---as shall be seen later---by a push of the vertex~$w_i$ from~$A$ to~$B$.
We ensure that exactly one of these vertices is pushed to~$B$ as follows.
Of the initially five clusters in~$G[A]$, all except one contain a special vertex called \emph{anchor vertex}~$a^2_i$.
(The superscript is not necessary here and used only for consistency with the formal construction given later on.)
We add vertices to~$B$ (shown above each anchor vertex) to ensure that an anchor vertex~$a^2_i$ cannot be moved to~$B$: Pushing~$a^2_i$ to~$B$ creates five~$P_3$s that only intersect in~$a^2_i$ and are otherwise independent.
Consequently, to destroy these forbidden induced subgraphs in~$B$, one needs to push five independent vertices to~$A$ which would create too many clusters.
Hence, the only viable push to decrease the number of clusters in~$G[A]$ is pushing~$v^*$ from~$A$ to~$B$.
This, however, creates a~$P_3$ in~$B$ since~$v^*$ has two nonadjacent neighbors in~$B$.
Consequently, one of these two neighbors needs to be pushed to~$A$.
Now this vertex is adjacent to the anchor vertex~$a^2_1$ and to a vertex in the cluster containing the anchor vertex~$a^2_2$.
Since~$a^2_1$ cannot be pushed to~$B$ as discussed above, we need to push the other vertex from~$A$ to~$B$ to ensure that~$G[A]$ is~$P_3$-free.
This results in a new~$P_3$ in~$G[B]$ and again we have two possibilities to destroy this~$P_3$ by pushing a vertex to~$A$.
In each case, the vertex that is pushed to~$A$ is adjacent to~$a^2_3$ and to a vertex in the cluster of~$a^2_4$.
This vertex will be pushed to~$A$.
Depending on the possible choices of pushes from~$B$ to~$A$, this vertex could be~$w_1$,~$w_2$, $w_3$, or~$w_4$.

\begin{figure}[t]
  \centering
  \begin{tikzpicture}[xscale=0.8,yscale=0.8]

    \node[bvertex,label=above:$v^*$] (a1) at (2,0) {};

    \node[avertex,label=right:$a^2_1$] (a2) at (4,1) {};

    \node[avertex,label=right:$a^2_2$] (a3) at (7,1) {};
    \node[bvertex] (a31) at (6.25,0) {};
    \node[bvertex] (a32) at (7.75,0) {};
    \draw[] (a31)--(a3)--(a32)--(a31);

    \node[avertex,label=right:$a^2_3$] (a4) at (10,1) {};

    \node[avertex,label=right:$a^2_4$] (a5) at (13,1) {};
    \node[ivertex,label = below:$w_1$] (a51) at (11.5,0) {};
    \node[ivertex,label = below:$w_2$] (a52) at (12.5,0) {};
    \node[ivertex,label = below:$w_3$] (a53) at (13.5,0) {};
    \node[ivertex,label=below:$w_4$] (a54) at (14.5,0) {};
    \draw[] (a51)--(a52)--(a53)--(a54);
    \draw[] (a51)--(a5)--(a52);
    \draw[] (a53)--(a5)--(a54);
    \draw[-]  (a51)  [bend left=32] to (a53);
    \draw[-]  (a52)  [bend left=32] to (a54);
    \draw[-]  (a51)  [bend left=32] to (a54);

    \node[bvertex] (b11) at (3.5,-1.5) {};
    \node[bvertex] (b12) at (4.5,-1.5) {};

    \node[bvertex] (b21) at (8.25,-1.5) {};
    \node[bvertex] (b22) at (9.25,-1.5) {};
    \node[bvertex] (b23) at (10.75,-1.5) {};
    \node[bvertex] (b24) at (11.75,-1.5) {};

    \draw[very thick,red] (a1)--(b12)--(a32)--(b24) edge [bend right=13] (a54);

    \draw[] (a1)--(b11);
    \draw[] (a2)--(b11);
    \draw[] (a2)--(b12);
    \draw[] (a31)--(b11);
    \draw[] (a31)--(b21);
    \draw[] (a31)--(b22);
    \draw[] (a32)--(b23);
    \draw[] (a4)--(b21);
    \draw[] (a4)--(b22);
    \draw[] (a4)--(b23);
    \draw[] (a4)--(b24);
    \draw (a51) edge [bend left=13] (b21);
    \draw[] (a52) edge [bend left=13] (b22);
    \draw[] (a53) edge [bend left=13] (b23);

    \foreach \i in {2,3,4,5}{
      \node[bvertex] (c\i0) at (3*\i -3,2) {};
      \node[bvertex] (d\i0) at (3*\i -3,3) {};
      \node[bvertex] (c\i1) at (3*\i -2.5,2) {};
      \node[bvertex] (d\i1) at (3*\i -2.5,3) {};
      \node[bvertex] (c\i2) at (3*\i -2,2) {};
      \node[bvertex] (d\i2) at (3*\i -2,3) {};
      \node[bvertex] (c\i3) at (3*\i -1.5,2) {};
      \node[bvertex] (d\i3) at (3*\i -1.5,3) {};
      \node[bvertex] (c\i4) at (3*\i -1,2) {};
      \node[bvertex] (d\i4) at (3*\i -1,3) {};
      \draw[] (a\i)--(c\i0)--(d\i0);
      \draw[] (a\i)--(c\i1)--(d\i1);
      \draw[] (a\i)--(c\i2)--(d\i2);
      \draw[] (a\i)--(c\i3)--(d\i3);
      \draw[] (a\i)--(c\i4)--(d\i4);
    }

        \begin{pgfonlayer}{bg}
           { \draw [cyan, rounded corners, fill] (10,3.5) -- (14.5,3.5) -- (14.5, 1.5) --
             (2.5,1.5)--(2.5,3.5)--(10,3.5);
             \draw [cyan, rounded corners, fill] (10,-1) -- (14.5,-1) -- (14.5, -2) --
             (2.5,-2)--(2.5,-1)--(10,-1);
           }
           \input{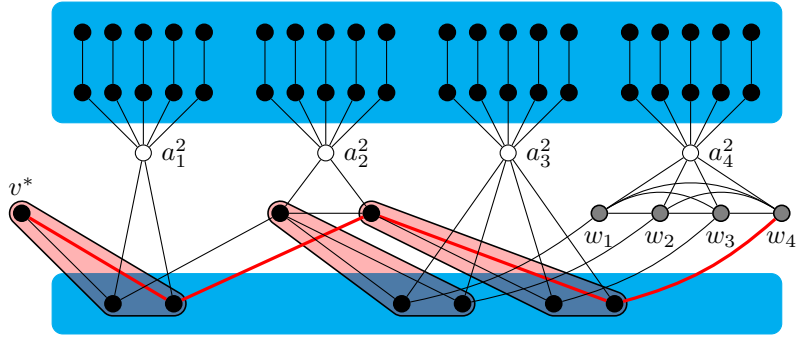}
           \newcommand\forbr{2mm}
           \draw[forb]
           \hedgeiii{a1}{b12}{b11}{\forbr}
           \hedgeiii{a31}{b22}{b21}{\forbr}
           \hedgeiii{a32}{b24}{b23}{\forbr};
        \end{pgfonlayer}

      \end{tikzpicture}

      \caption{An example of the instance selection in the case where~$\Pi$ is the set of
        cluster graphs and one of four instances needs to be selected. The blue background
        highlights vertices that are initially in~$B$. Anchor vertices are white. The vertices~$w_1$, $w_2$, $w_3$, and $w_4$, corresponding to the instances from which one is to be selected, are grey. A possible path of pushes starting with~$v^*$ and ending with~$w_4$ is drawn with thick red edges. The red regions with solid outline show the $P_3$s that we use as forbidden induced subgraphs for $G[B]$ in the argument that at least one $w_i$-vertex is pushed to~$B$.}
      \label{fig:instance-select}
    \end{figure}

\subsection{Construction Outline}\label{sec:nopk-outline}
Let $t$ instances of \pCIS\ with graphs $G_1, \ldots, G_t$ be given; we will refer to the instances by their indices $1, \ldots, t$.
The construction is divided into three main parts:
\begin{itemize}
\item an instance selection part, whose purpose is selecting one of the $t$ \pCIS\ instances;%
\item a vertex selection part, whose purpose is selecting from each color in the selected instance a vertex into the independent set; and
\item a verification part, whose purpose is ensuring that the selected vertices form an independent set.
\end{itemize}
All these parts share a common scaffold, which is given by the so-called anchor vertices as outlined below.
See \cref{fig:composition-struc} for a schematic view of the construction.
For illustrative purposes, by~$(A, B)$ we refer to a \pipartition\ for the constructed instance.
Next, we explain the different components of the construction.

\subsubsection*{Anchors}
The construction starts in \cref{sec:nopk-basicops} by introducing a special set of vertices, which we call anchors, introduced before any of the selection and verification gadgets.
We then introduce small gadgets that ensure that the \pipartition~$(A, B)$ has the property that each anchor is part of a distinct cluster in the cluster-graph side~$G[A]$ and that each cluster contains an anchor (similar to the example given in \cref{sec:nopk-example}).
In other words, each cluster in $G[A]$ ``grows around'' some anchor.
The remaining gadgets will be attached to the anchors and contain different parts of the clusters around an anchor.
We introduce the anchors at the beginning since they need to be \emph{shared} by several gadgets. The reason being that the number of anchors will be equal to the number of clusters, the parameter in the constructed instance, which we need to keep small in order for the cross-composition to work.

\subsubsection*{Instance Selection}
In \cref{sec:nopk-instsel}, we introduce an instance-selection gadget, in analogy to the example from \cref{sec:nopk-example}.
This gadget is represented by the leftmost triangle in \cref{fig:composition-struc}.
Two of its main features are the activator vertex (on the left tip) and the set of choice vertices (on the right side).
In the example in \cref{sec:nopk-example} the choice vertices were labeled $w_1$ to~$w_4$.
The choice vertices correspond in a one-to-one fashion to the instances of \pCIS\ given as input.
The activator vertex is not adjacent to any anchor, and since the number of clusters in $G[A]$ equals the number of anchors, the activator vertex~$v^*$ cannot be in $A$.
By the properties explained in \cref{sec:nopk-example}, this means that one of the choice vertices is in $B$; this corresponds to selecting one instance.
The remaining gadgets verify that the selected instance is a yes-instance.

\subsubsection*{Vertex Selection}
The next set of gadgets, described in \cref{sec:nopk-vertsel}, aims to select a vertex from each color of the previously-selected instance.
This works similarly to the instance selection, by introducing a gadget---analogous to the example from \cref{sec:nopk-example}---for each color and each instance.
These gadgets are shown in the center of \cref{fig:composition-struc}.
As before, each gadget consists of an activator vertex (on its left tip), and a set of choice vertices (on its right side).
For each instance $i \in [t]$ and color $\ell \in [k]$, the choice vertices in the gadget instance $G_i$ correspond in a one-to-one fashion to the vertices of color $\ell$ in $G_i$.

The choice of instance---as done before---is transferred to the vertex-selection gadgets using a small additional gadget that has the following effect: If a choice vertex of some instance is in $B$, then in \emph{each} vertex-selection gadget for that instance (\ie, for all colors) the activator vertex is in~$B$.
This construction is indicated by the dashed lines in \cref{fig:composition-struc}.
Again, by the properties of the selection gadget explained in \cref{sec:nopk-example}, this means that, for each color~$\ell$, there is a vertex in $G_i$ that is selected by a choice vertex in the gadget for color~$\ell$.

A crucial part to the cross-composition construction is obtaining an upper bound on the number of clusters in $G[A]$ that is polynomial in $k$, the independent set size, and logarithmic in $t$, the number of instances.
Note that each of the selection gadgets uses a number of clusters, and hence, it is imperative (for the reasons explained in the previous sentence) that the selection gadgets share clusters.
\Cref{fig:selcons} shows schematically which gadgets share clusters.
The corresponding cliques will be merged into one large clique containing exactly one anchor.
We will call these large cliques \emph{dials}, supporting the intuition that exactly one of the gadgets sharing a dial may be active in a \pipartition.

\subsubsection*{Verification}
In the final part of the construction, given in \cref{sec:nopk-verif}, we ensure that it is impossible to simultaneously push into $B$ two choice vertices that correspond to two adjacent vertices in some input instance.
By itself, this would be easy to do: Simply introduce a forbidden induced subgraph for $\Pi$ that contains the two choice vertices.
However, we also need to allow pushing \emph{one} of the two choice vertices into~$B$. Since a vertex corresponding to a choice vertex may be adjacent to many other vertices in the corresponding input graph, the corresponding parts of the forbidden subgraphs may overlap and inadvertently introduce new forbidden subgraphs.
Hence, we need a slightly more involved construction that consists of more independent vertex and edge gadgets, shown schematically on the right in \cref{fig:composition-struc}; the construction is described in detail in \cref{sec:nopk-verif}.

The basic idea is to have an edge gadget for each edge $e = \{uv\}$ in an \pCIS\ instance and vertex gadgets corresponding to the endpoints $u, v$ of that edge.
If the choice vertex corresponding to endpoint $u$ (resp.\ $v$) is in $B$, then vertex gadgets and additional gadgetry will ensure that a vertex, $u_e$ (resp.\ $v_e$), corresponding to $u$ (resp.\ $v$) in the edge gadget is in $B$ as well.
Using additional gadgetry we make $u_e$ and $v_e$ part of a forbidden subgraph for $\Pi$ whose remaining vertices are fixed in $B$---this ensures that not both $u_e$ and $v_e$ can be in $B$, that is, not both $u$ and $v$ are chosen into the independent set.

\begin{figure}[p]
  \centering
  \begin{tikzpicture}[on grid, scale =.8, every edge/.style={draw, semithick}]
    \input{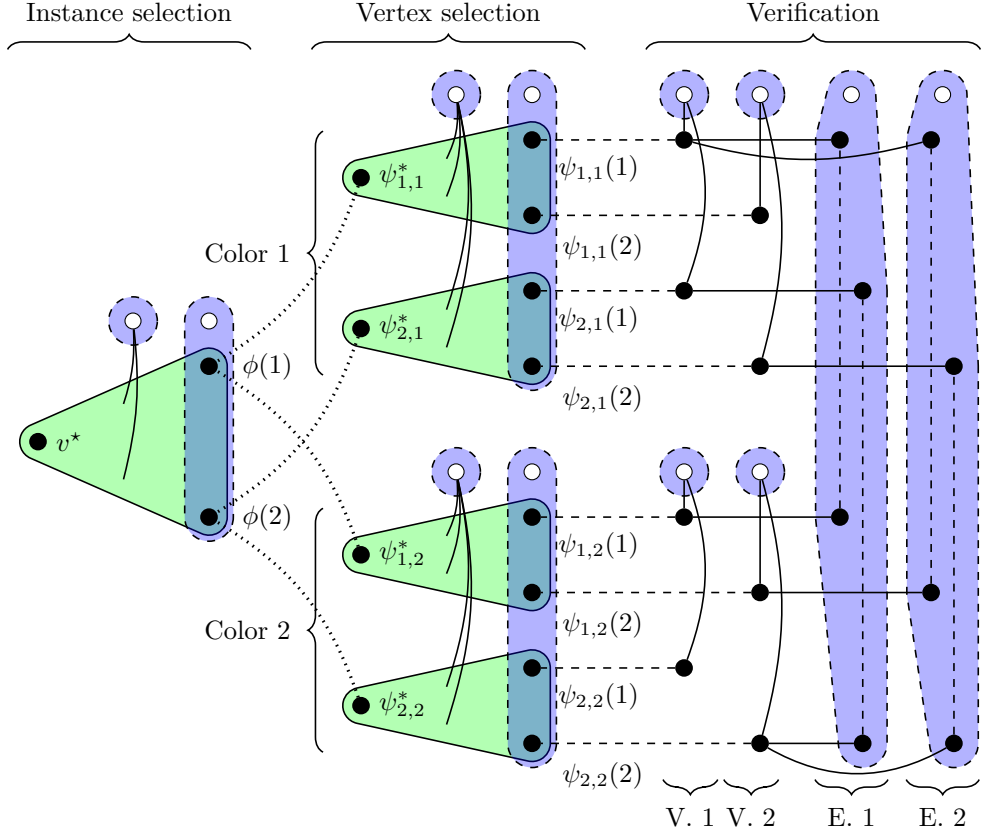}
    \newcommand\seladdlength{1.25cm}
    \newcommand\dialr{4mm}
    \newcommand\selgadr{3mm}
    \tikzstyle{dial} = [dashed, fill = blue, fill opacity = 0.3, semithick]
    \tikzstyle{selgad} = [fill = green, fill opacity = 0.3, semithick]

    \tikzstyle{brace} = [semithick, decorate, decoration = {brace, amplitude = 2mm}]
    \tikzstyle{bracemirr} = [semithick, decorate, decoration = {brace, mirror, amplitude = 2mm}]

    \tikzstyle{nota} = [dotted, very thick]

    \node (vs) [bvertex, label = right:$v^\star$] at (0, 0) {};

    \node (p1) [bvertex, above right = of vs, xshift = \seladdlength, label = {[label distance = 2mm]right:$\phi(1)$}] {};
    \node (p2) [bvertex, below right = of vs, xshift = \seladdlength, label = {[label distance = 2mm]right:$\phi(2)$}] {};

    \node (d1t) [avertex, above left = of p1, yshift = -4mm] {}
    edge [bend left = 10] ($ (vs) ! .5 ! (p1) $)
    edge [bend left = 10] ($ (vs) ! .5 ! (p2) $);
    \node (d1b) [left = of p2] {};

    \node (d2t) [avertex, above = of p1, yshift = -4mm] {};
    \node (d2b) at (p2) {};

    \begin{pgfonlayer}{bg}
      \draw[selgad] \hedgeiii{vs}{p1}{p2}{\selgadr};
      \draw[dial] \hedgei{d1t}{\dialr};
      \draw[dial] \hedgeii{d2t}{d2b}{\dialr};
    \end{pgfonlayer}

    \node (ss21) [bvertex, right = 2cm of p1, yshift = 5mm, label = right:$\psi^*_{2, 1}$] {}
    edge [nota, bend left = 20] (p2);
    \node (ss12) [bvertex, right = 2cm of p2, yshift = -5mm, label = right:$\psi^*_{1, 2}$] {}
    edge [nota, bend right = 20] (p1);

    \node (ss11) [bvertex, above = 2cm of ss21, label = right:$\psi^*_{1, 1}$] {}
    edge [nota, bend left = 20] (p1);
    \node (ss22) [bvertex, below = 2cm of ss12, label = right:$\psi^*_{2, 2}$] {}
    edge [nota, bend right = 20] (p2);

    \foreach \c in {1, 2} {     %
      \foreach \i in {1, 2} {   %
        \node (ss\i\c1) [bvertex, above right = of ss\i\c, yshift = -.5cm, xshift = \seladdlength, label = {[label distance = 2mm, yshift = 1.5mm]below right:$\psi_{\i, \c}(1)$}] {};
        \node (ss\i\c2) [bvertex, below right = of ss\i\c, yshift = .5cm, xshift = \seladdlength, label = {[label distance = 2.75mm, yshift = 1.5mm]below right:$\psi_{\i, \c}(2)$}] {};
        \node (d1t\i\c) [above left = of ss\i\c1, yshift = -4mm] {};
        \node (d1b\i\c) [left = of ss\i\c2] {};

        \node (d2t\i\c) [above = of ss\i\c1, yshift = -4mm] {};
        \node (d2b\i\c) at (ss\i\c2) {};
        \begin{pgfonlayer}{bg}
          \draw[selgad] \hedgeiii{ss\i\c}{ss\i\c1}{ss\i\c2}{\selgadr};
        \end{pgfonlayer}
      }
    }

    \foreach \c in {1, 2} {
      \node [avertex] at (d1t1\c) {}
      edge [bend left = 15] ($ (ss1\c) ! .5 ! (ss1\c1) $)
      edge [bend left = 15] ($ (ss1\c) ! .5 ! (ss1\c2) $)
      edge [bend left = 15] ($ (ss2\c) ! .5 ! (ss2\c1) $)
      edge [bend left = 15] ($ (ss2\c) ! .5 ! (ss2\c2) $);
      \node [avertex] at (d2t1\c) {};
      \begin{pgfonlayer}{bg}
        \draw[dial] \hedgei{d1t1\c}{\dialr};
        \draw[dial] \hedgeii{d2t1\c}{d2b2\c}{\dialr};
      \end{pgfonlayer}
    }

    \begin{scope}[transform canvas = {xshift = -2.75cm}]
      \draw [bracemirr]
      (ss111.north) -- (ss212.south)
      node [midway, left = 0.3cm] {Color 1};
      \draw [bracemirr]
      (ss121.north) -- (ss222.south)
      node [midway, left = 0.3cm] {Color 2};
    \end{scope}

    \foreach \c in {1, 2} {     %
      \foreach \i in {1, 2} {   %
        \node (v\i\c1) [bvertex, right = 2cm of ss\i\c1] {}
        edge [dashed] (ss\i\c1);
        \node (v\i\c2) [bvertex, right = 3cm of ss\i\c2] {}
        edge [dashed] (ss\i\c2);
        \node (d3l\i\c) [above = of v\i\c1, yshift = -4mm] {};
        \node (d3r\i\c) [above right = of v\i\c1, yshift = -4mm] {};
      }
    }

    \foreach \c in {1, 2} {     %
      \node [avertex] at (d3l1\c) {}
      edge (v1\c1)
      edge [bend left = 20] (v2\c1);
      \node [avertex] at (d3r1\c) {}
      edge (v1\c2)
      edge [bend left = 15] (v2\c2);
      \begin{pgfonlayer}{bg}
        \draw[dial] \hedgei{d3l1\c}{\dialr};
        \draw[dial] \hedgei{d3r1\c}{\dialr};
      \end{pgfonlayer}
    }

    \newcommand\edgeshifta{1.2cm}
    \newcommand\edgeshiftb{2.2cm}
    \newcommand\edgeshiftc{2.4cm}
    \newcommand\edgeshiftd{3.4cm}

    \node (ae1) [avertex, right = \edgeshifta of d3r11] {};

    \node (e111) [bvertex, right = \edgeshiftb of v111, xshift = -.15cm] {}
    edge (v111);
    \node (e112) [bvertex, right = \edgeshiftb of v121, xshift = -.15cm] {}
    edge (v121)
    edge [dashed] (e111);
    \node (e121) [bvertex, right = \edgeshiftb of v211, xshift = .15cm] {}
    edge (v211);
    \node (e122) [bvertex, right = \edgeshifta of v222, xshift = .15cm] {}
    edge (v222)
    edge [dashed] (e121);

    \begin{pgfonlayer}{bg}
      \draw [dial] \hedgem{ae1}{e121}{e122,e112,e111}{\dialr};
    \end{pgfonlayer}

    \node (ae2) [avertex, right = 1.2cm of ae1] {};

    \node (e211) [bvertex, right = \edgeshiftd of v111, xshift = -.15cm] {}
    edge [bend left = 15] (v111);
    \node (e212) [bvertex, right = \edgeshiftc of v122, xshift = -.15cm] {}
    edge (v122)
    edge [dashed] (e211);
    \node (e221) [bvertex, right = \edgeshiftc of v212, xshift = .15cm] {}
    edge (v212);
    \node (e222) [bvertex, right = \edgeshiftc of v222, xshift = .15cm] {}
    edge [bend left] (v222)
    edge [dashed] (e221);

    \begin{pgfonlayer}{bg}
      \draw [dial] \hedgem{ae2}{e221}{e222, e212, e211}{\dialr};
    \end{pgfonlayer}

    \begin{scope}[yshift = 6.5cm]
      \draw [brace] (-0.5, 0) -- (3.5, 0) node [midway, above = .25cm] {Instance selection};
      \draw [brace] (4.5, 0) -- (9, 0) node [midway, above = .25cm] {Vertex selection};
      \draw [brace] (10, 0) -- (15.5, 0) node [midway, above = .25cm] {Verification};
    \end{scope}

    \begin{scope}[xshift = .28cm, yshift = -5.6cm]
      \draw [bracemirr] (10, 0) -- (10.9, 0) node [midway, below = .25cm] {V.\ 1};
      \draw [bracemirr] (11, 0) -- (11.9, 0) node [midway, below = .25cm] {V.\ 2};
    \end{scope}

    \begin{scope}[xshift = .28cm, yshift = -5.6cm]
      \draw [bracemirr] (12.5, 0) -- (13.7, 0) node [midway, below = .25cm] {E.\ 1};
      \draw [bracemirr] (14, 0) -- (15.2, 0) node [midway, below = .25cm] {E.\ 2};
    \end{scope}

  \end{tikzpicture}
  \caption{Schematic of the arrangement of the gadgets for modeling the instance and vertex selection and how dials are shared among them.
    In this example, there are two instances of \pCIS.
    Each instance has two color classes, two vertices of each color class, and two edges.
    Referring to the instance-selection part,
    the green solid area represents a selection gadget similar to the example given in \cref{sec:nopk-example} (it is produced by a procedure called $\selcon$, given in \cref{sec:nopk-instsel}).
    The vertex on the left tip is the activator vertex and the vertices on the right side are the choice vertices.
    White vertices represent anchors and blue dashed areas represent dials (the corresponding edges are omitted).
    We show only two anchors and dials in the selection gadget but more may be present.
    Referring to the vertex-selection part,
    dotted edges indicate an additional construction that implies that if one of the endpoints is in $B$, then the other is as well;
    this transfers the choice of instance from the instance to the vertex selection gadgets.
    The vertex-selection part consists of one vertex-selection gadget for each color class and instance.
    Thus, in total there are four selection gadgets.
    The dials corresponding to each color are shared among the selection gadgets of that color, regardless of their instance.
    Referring to the verification part, dashed edges indicate an additional construction that implies that not both endpoints can be in $B$; we say that they have been made exclusive.
    This additional construction transfers the choice of vertices from the vertex-selection to the verification part.
    The verification part consists of vertex gadgets (left, labeled V.~1 and V.~2) and edge gadgets (right, labeled E.~1 and E.~2).
    Their construction is explained in \cref{sec:nopk-verif}.
    In essence, pushing a choice vertex to $B$ in a vertex-selection gadget necessitates pushing a corresponding vertex to~$B$ in all edge-verification gadgets corresponding to adjacent edges.
    For example, putting $\psi_{1, 1}(1) \in B$ requires that the corresponding vertices in edge gadgets to the right are in $B$.
    Edges are represented by pairs of exclusive vertices in the edge-verification gadgets.
    Thus, not both choice vertices corresponding to the endpoints of an edge can be put into~$B$.
  }
  \label[figure]{fig:composition-struc}
\end{figure}

\subsection{Scaffolding and Basic Subconstructions}\label{sec:nopk-basicops}
\label{sec:nopk-setup}

We now begin with the formal description of the cross-composition and prove its properties.
Let $t$ instances of \pCIS\ be given, with graphs $G_1, \ldots, G_t$, respectively.
Below, we use an instance and its index in $[t]$ interchangeably.
Without loss of generality, assume that the following properties hold; they can be achieved by simple padding techniques.

\begin{compactitem}
\item Each instance asks for an independent set of size~$k$ (otherwise, introduce new colors and isolated vertices as needed);
\item each color class in each graph has $n$ vertices and $n$ is a power of two (otherwise, in a color class that does not satisfy this property, add as needed new vertices that are adjacent to all vertices in all other color classes); and
\item $t$ is a power of two (otherwise, duplicate one of the instances as needed).
\end{compactitem}
In the following, let $m$ be the maximum number of
edges over all graphs~$G_i$, for $i \in [t]$.

We construct an instance of \picluster\ as described in this and the following sections.
The instance consists of the graph~$G$ and asks for a \pipartition~$(A, B)$ with at most $d$ clusters in~$G[A]$ (we specify $d$ below).
Throughout, we denote by $(A, B)$ an arbitrary fixed \pipartition\ of the (so-far constructed) graph~$G$.
We also fix $M$ to be a forbidden induced subgraph of $\Pi$ with minimum number of vertices.
By the properties of~$\Pi$, $M$ contains at least three vertices.
The graph~$G$ is constructed by first adding $d$ vertices which we call {\em anchors} (see below).
The clusters in any \pipartition~$(A, B)$ of $G$ with $d$ clusters in~$G[A]$ will extend these anchor vertices into larger cliques; we show below how to achieve that.
We then successively add gadgets that are attached to these anchors, as outlined in \cref{sec:nopk-outline}.

\subsubsection{Anchors}
As mentioned before, the construction begins by adding \emph{anchor} vertices.
In \cref{sec:nopk-notation} we will add gadgets to ensure that each anchor vertex is in~$A$.
We introduce $d$~anchor vertices, divided into $4 + 2k$ groups, as follows:
\begin{itemize}
\item Introduce the anchors $a^1_1, a^1_2$; these two anchors serve to fix vertices into~$B$ by making any such vertex adjacent to both $a^1_1$ and $a^1_2$.
\item Introduce the anchors $a^2_1, a^2_2, \ldots, a^2_{2\log t}$; these anchors will be used in the instance-selection gadget.
\item Introduce the anchors $a^3_1, a^3_2, \ldots, a^3_{k + 1}$; these anchors serve to connect the instance-selection gadget with the vertex-selection gadgets.
\item For each $i \in [k]$, introduce the anchors $a^{3 + i}_1, a^{3 + i}_2, \ldots, a^{3 + i}_{\log n}$; these anchors are used by the vertex-selection gadgets.
\item For each $i \in [k]$, introduce the anchors $a^{3 + k + i}_1, a^{3 + k + i}_2, \ldots, a^{3 + k + i}_n$; these anchors are used in vertex subgadgets of the verification gadgets.
\item Finally, introduce the anchors $a^{4 + 2k} _1, a^{4 + 2k} _2, \ldots, a^{4 + 2k}_{m}$; these anchors are used in edge subgadgets of the verification gadgets.
\end{itemize}
Hence, we define the number~$d$ of desired clusters to be $d \coloneqq 2 + 2 \log(t) + (k + 1) + k\log n + kn + 2m$.

\subsubsection{Helper, Dial, and Volatile Vertices, Fixing Anchors}\label{sec:nopk-notation}
Before explaining how to ensure that all anchors are in $A$, we introduce some notation.
Many of the gadgets will fix some vertices into $A$ or $B$ using some additional auxiliary vertices.
To avoid having to reason about the fixed and auxiliary vertices, we will name them and introduce invariants of the ensuing construction that allow us to ignore these vertices when constructing \pipartition s.

Throughout, we use the following notation.
The vertices that we introduce will be in three disjoint categories: \emph{helper} vertices, \emph{dial} vertices, and \emph{volatile} vertices.
Their meaning is as follows.
Helper vertices will always be contained in~$B$ and only serve to impose certain properties on other vertices.
Dial vertices are normally in~$A$ and belong to a cluster extending around an anchor; some of these vertices may be pushed to~$B$.
On the other hand, volatile vertices are normally in~$B$ and may be pushed to~$A$.
To gain some intuition, consider \cref{fig:instance-select}.
The vertices in the top blue area will be helper vertices, the vertices on white background in the middle (except for $v^*$) will be dial vertices, the vertices on blue background on the bottom will be volatile vertices.

We now fix each of the anchors into~$A$ by introducing, for each anchor~$a_i^j$, $d + 1$ copies of $M$ and, for each copy, identifying an arbitrary vertex of that copy with~$a_i^j$.
The vertices different from~$a_i^j$ in the copies of~$M$ are helper vertices.
See \cref{fig:instance-select} for the special case of $d = 5$ and $M$ being a $P_3$.
Suppose that $a_i^j \in B$.
Then out of each of the $d$ incident copies of $M$, at least one vertex is in~$A$, and since these vertices are pairwise nonadjacent, $G[A]$ would contain at least~$d + 1$ clusters, which is a contradiction.
Thus, each anchor is in $A$.

When we construct \pipartition s in the following we always tacitly assume that anchors are in~$A$ and all helper vertices are in~$B$.
More generally, we maintain the following invariant throughout the construction.
\begin{invariant}\label[invariant]{inv:anchor-helper}
  For each \pipartition~$(A, B)$ of $G$ with at most~$d$ clusters in~$G[A]$, all anchors are in~$A$ and all helper vertices are in~$B$.
  Each helper vertex was introduced as part of a copy of $M$ and made adjacent to both $a_1^1$ and $a_2^1$. No helper vertex has any other neighbors than those in the copy of~$M$ it was introduced in and the anchors $a_1^1$ and~$a_2^1$.%
\end{invariant}

\subsubsection{Dials and Joining Dials}
As mentioned before, the gadgets that we are about to construct share clusters in $G[A]$.
Hence, for a shared cluster, we need to ensure that its parts coming from different gadgets are mutually and completely adjacent.
To describe the gadgets in a self-contained way, we define dials: A \emph{dial} will be a set~$D$ of vertices that contains exactly one anchor, say $a$, and that induces a clique.
At some point in the construction, we may want to construct a new gadget that uses the anchor~$a$.
That is, the gadget needs to contain a subset~$S$ of vertices that shall form a cluster with $a$, together with all other vertices that have already/previously been added to the cluster containing $a$.
To do so, we let $S$ \emph{join the dial} $D$, whose anchor is $a$, by making $D \cup S$ a clique. %

Formally, we associate each anchor~$a_i^j$ with a vertex set~$D_i^j$ containing $a_i^j$.
(And $D_i^j$ is supposed to induce a clique in~$G$ throughout the construction.)
We say that $D_i^j$ is the \emph{dial} of~$a_i^j$.
Initially, we have $D_i^j = \{a_i^j\}$.
Later on, other vertices may join $D_i^j$, which is defined as follows.
By making a vertex $v$ \emph{join} a dial $D_i^j$, we mean that we put $v$ into $D_i^j$ and make $v$ adjacent to all other vertices in~$D_i^j$.
Vertex $v$ is then designated as a dial vertex.
Throughout, we maintain the following invariant:
\begin{invariant}\label[invariant]{inv:dials}
  \begin{inparaenum}[(i)]
  \item Each dial induces a clique in $G$.
  \item No two dial vertices in different dials are adjacent.
  \end{inparaenum}
\end{invariant}
Note that this is true so far, since the only vertices currently in dials are anchors.

Next, there will be two different flavors of dials, similar to what we showed in \cref{sec:nopk-example}: A dial may either serve to push a vertex from this dial into $B$, or to accept a (volatile) vertex which is pushed to $A$.
In order to simplify the reasoning about which vertices may be in a cluster in $G[A]$, we use the following invariant.
\begin{invariant}\label[invariant]{inv:volatile}
\begin{lemenum}
  \item The following dials~$D_i^j$ are singletons:
    \begin{inparaenum}[a)]\item %
      for each $j \in \{2\} \cup \{4, 5, \ldots, 3 + k\}$ and each odd~$i$ (these are the dials used in the instance- and vertex-selection gadgets), and
    \item %
      for each $j \in \{3 + k + 1, \ldots, 3 + 2k\}$ (and each $i$; these are the dials used in the vertex subgadgets of the verification gadgets).
    \end{inparaenum}
  \item For each anchor~$a_i^j$, each volatile vertex is either nonadjacent to~$a_i^j$ or adjacent to all vertices in~$a_i^j$'s dial~$D_i^j$.
  \end{lemenum}
\end{invariant}
A particular corollary will be that each cluster in~$G[A]$ either contains an anchor~$a_i^j$ with a dial~$D_i^j = \{a_i^j\}$, or contains only vertices of~$D_i^j$.
This will help in the correctness proof, where we build a \pipartition\ for~$G$ piece-by-piece.

\newcommand{\friendly}[1]{friendly with respect to the dials~\ensuremath{#1}}%
We introduce the following terminology:

\begin{definition}[Friendly partition]
  Let $(A, B)$ be a \pipartition\ for~$G$ and $\cal D$ be a set of
  dials. Partition~$(A, B)$ is \emph{friendly} %
  with respect to $\cal D$ if each
  singleton dial in~$\cal D$ is a singleton cluster in~$G[A]$.
\end{definition}

\subsubsection{Making Vertices Exclusive}
As a final prerequisite, we need the operation of making three vertices~$u, v, w$ exclusive.
Intuitively, this operation was the main tool used in \cref{sec:nopk-example} to fan out the possible pushes in avalanches according to a binary tree-like structure: When $u$ is pushed to $B$, either $v$ or $w$ can be pushed to $A$ to repair the partition.
We use this construction extensively in the selection gadget described~below.

\begin{definition}[Making three vertices exclusive]
  Let~$u, v, w \in V(G)$ be three vertices satisfying the following conditions:
  At most two of $u, v, w$ are dial vertices.
  Furthermore, any edge between two vertices $\{u, v, w\}$ is contained in a dial. %

  By \emph{making~$u$, $v$, and~$w$ exclusive} we mean:
  \begin{lemenum}
  \item introducing a copy of (the forbidden subgraph)~$M$ consisting of new vertices into~$G$;
  \item identifying three distinct vertices of~$M$ with~$u$, $v$, and~$w$---if there are two dial vertices among $u, v, w$, then we identify them with two adjacent vertices of $M$; and
  \item making each remaining vertex (if any) of $M$ (different from $u$, $v$,
    and $w$) adjacent to both anchors~$a^1_1$ and~$a^1_2$.
  \end{lemenum}
  The vertices in $V(M) \setminus \{u, v, w\}$ are helper~vertices.
\end{definition}
See \cref{fig:make-excl} for an illustration.

\begin{figure}[t]
  \centering
  \begin{tikzpicture}[scale = 1.2]
    \input{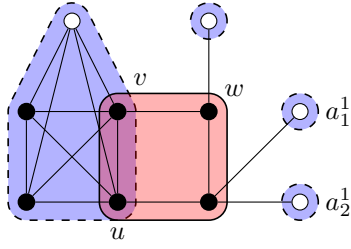}
    \node [avertex] (a21) at (0, 2) {};
    \node [avertex] (a22) at (-1.5, 2) {};
    \node [avertex, label = {[label distance = 1mm]right:$a^1_1$}] (a11) at (1, 1) {};
    \node [avertex, label = {[label distance = 1mm]right:$a^1_2$}] (a12) at (1, 0) {};

    \node [bvertex, label = {above right:$w$}] (w) at (0, 1) {};
    \node [bvertex] (x) at (0, 0) {};
    \node [bvertex, label = {[label distance = 1.3mm]70:$v$}] (v) at (-1, 1) {};
    \node [bvertex, label = {[label distance = 1mm]below:$u$}] (u) at (-1, 0) {};
    \node [bvertex] (d1) at (-2, 1) {};
    \node [bvertex] (d2) at (-2, 0) {};

    \graph{
      (v) -- (w) -- (x) -- (u);

      {[clique] (u), (v), (d1), (d2), (a22)};

      (x) -- {(a11), (a12)};

      (w) -- (a21);
    };

    \begin{pgfonlayer}{bg}
      \draw [dial] \hedgem{a22}{v}{u, d2, d1}{2mm};
      \draw [forb] \hedgem{v}{w}{x, u}{2mm};
      \draw [dial] \hedgei{a21}{2mm};
      \draw [dial] \hedgei{a11}{2mm};
      \draw [dial] \hedgei{a12}{2mm};
    \end{pgfonlayer}
  \end{tikzpicture}
  \caption{Making $u$, $v$, and $w$ exclusive. In this example, $M$ is a cycle with four vertices, $u$ and $v$ are part of a dial, and $w$ is adjacent to some anchor. Dials are represented by blue regions with dashed outlines and the introduced copy of $M$ by the red region with solid outline. Anchor vertices are white.}
  \label{fig:make-excl}
\end{figure}

Observe that step~(iii) entails that $V(M) \setminus \{u, v, w\} \subseteq B$, since otherwise, there would be a $P_3$ in $G[A]$ involving $a^1_1$ and $a^1_2$.
Hence, \cref{inv:anchor-helper} is maintained by this operation (clearly, \cref{inv:volatile} is maintained as well).
\cref{inv:dials} is maintained as well since all dial vertices among $u$, $v$, and~$w$ are contained in the same dial by the preconditions on these vertices.

Furthermore, not all three $u, v, w \in B$; otherwise, since $V(M) \setminus \{u, v, w\} \subseteq B$, $G[B]$ would contain a copy of~$M$.
That is, making vertex exclusive indeed imposes the constraint on $(A, B)$ that we are aiming for.
For further reference, we state this fact in the following lemma.
\begin{lemma}\label[lemma]{lem:exclusive-soundness}
  Let $G$ be the graph obtained at any point during the construction in which $u, v, w \in V(G)$ were made exclusive and let $(A, B)$ be a \pipartition\ for $G$ with at most $d$ clusters in $G[A]$. Then at least one of the vertices $u, v, w$ is not in~$B$.
\end{lemma}

To simplify arguing about the existence of \pipartition s,
we will always tacitly assume that $V(M) \setminus \{u, v, w\} \subseteq B$ and ignore the vertices in~$V(M) \setminus \{u, v, w\}$.
To simplify proving that the constructed partitions are indeed \pipartition s, we derive the following sufficient conditions.
\begin{lemma}\label[lemma]{lem:exclusive-completeness}
  Let $G$ be the graph obtained at any point during the construction in which $u, v, w \in V(G)$ were made exclusive using a copy of~$M$, denoted as $M$ in a slight abuse of notation, and let $(A, B)$ be a bipartition of~$V(G)$.
  Suppose that
  \begin{lemenum}
  \item $G[A]$ is a cluster graph with at most $d$ clusters,
  \item $G[B \setminus V(M)] \in \Pi$,
  \item at least one of~$u, v, w$ is in~$A$, and
  \item $u, v, w$ are each adjacent only to some subset of~$\{u, v, w\}$ in~$G[B]$.
  \end{lemenum}
  Then, $(A, B \cup (V(M) \setminus \{u, v, w\}))$ is a \pipartition\ for~$G$.
\end{lemma}
\begin{proof}
  Without loss of generality, by symmetry, we may assume that $u \in A$.
  Clearly, it suffices to prove that $G[B \cup (V(M) \setminus \{u, v, w\})] \in \Pi$.
  Suppose that there is a copy $M'$ of $M$ contained in $G[B \cup (V(M) \setminus \{u, v, w\})]$ as an induced subgraph.
  Since $G[B \setminus V(M)] \in \Pi$, graph $M'$ contains a vertex of $V(M)$.
  Since $u \in A$ and thus $|V(M) \setminus A| < |V(M')|$, graph $M'$ moreover contains a vertex of $B \setminus V(M)$.
  By condition~(iv), $v$ and $w$ are adjacent in $G[B]$ only to some subset of $\{v, w\}$. Since $M$ is connected, there is an edge $e$ in $M'$ between $V(M) \setminus \{v, w\}$ and $B \setminus V(M)$.
  Indeed, since $u \in A$, edge $e$ contains a vertex of $V(M) \setminus \{u, v, w\}$, that is, a helper vertex.
  By \cref{inv:anchor-helper}, helper vertices in $M$ do not receive further edges, and $e$ thus contains a helper vertex and an anchor vertex.
  Note that each anchor vertex~$a$ is in $A$: Each of $a$'s incident copies of $M$ consists otherwise only of helper vertices and by \cref{inv:anchor-helper} these copies are present in~$G$ as well.
  Thus, if $a$ was in $B$, then there would be more than $d$ clusters in $G[A]$, a contradiction to assumption~(i).
  Thus, $e$ contains a vertex in $A$, a contradiction to the fact that $e$ is in $M'$ which is a subgraph of $G[B]$.
\end{proof}

The operation of making vertices exclusive can naturally be applied to (only) two vertices:
\begin{definition}\label[definition]{def:two-exclusive}
  Let~$u, v \in V(G)$ be two vertices such that, if they both are dial vertices, then they are contained in the same dial.
  By \emph{making~$u$ and $v$ exclusive} we~mean:
  \begin{lemenum}
  \item introducing a new helper vertex~$x$;
  \item making $x$ adjacent to both $a^1_1$ and $a^1_2$; and
  \item making $u$, $v$, and $x$ exclusive.
  \end{lemenum}
\end{definition}
\noindent Note that \cref{lem:exclusive-soundness,lem:exclusive-completeness} hold analogously.

We next explain a generic selection gadget construction, and then use it to construct an instance-selection gadget and, for each instance and color, vertex-selection gadgets.

\subsection{Generic Selection Gadget and Instance Selection}\label{sec:nopk-instsel}

We now introduce a generic selection gadget that we use for both instance and vertex selection.
The inner workings of the gadget use the necessary pushes along a binary-tree-like structure as outlined in \cref{sec:nopk-example}; refer to \cref{fig:instance-select,fig:selcons} for examples.
That is, the gadget is constructed such that it allows a trivial \pipartition~$(A, B)$ for the resulting graph with $d + 1$ clusters in $G[A]$, one more than allowed.
This cluster is a singleton, called the activator vertex.
Pushing it to $B$ results in a forbidden induced subgraph in~$G[B]$ requiring subsequent pushes to~$A$.  Each of these pushes to~$A$ will create a~$P_3$ involving two anchors, meaning that the third vertex has to be pushed to~$B$.
This again creates a forbidden subgraph in $B$ and so on.
The leaves in the resulting tree-like structure correspond to the selection to be made.
That is, there is a set of dial vertices, which we call {\em choice} vertices below, which are normally in~$A$.
Through a path of pushes in the binary-tree-like structure, one of the choice vertices will be pushed to~$B$.
This push will in turn activate other gadgets.

For use as an instance-selection gadget, we need to take special care so that the number of clusters used is roughly logarithmic in the number of instances.
We achieve this by using only two clusters (represented by anchors and their dials) per level in the binary-tree-like structure of pushes; see \cref{fig:selcons}.
For use as a vertex-selection gadget, to bound the number of clusters in the size of the largest instance, we need to ensure that all the vertex-selection gadgets share their corresponding clusters.
We achieve this by grouping the gadgets according to the groups of anchors above; each gadget uses only anchors in their corresponding group and shares these anchors with all other gadgets in this group.
Essentially, the operation of vertices joining dials makes it possible to define the selection gadgets in a relatively local way.

\subsubsection*{Construction}
We use the following (generic) construction, called $\selcon(p, q)$, both for selecting an instance and for selecting the independent-set vertices in that instance.
For this purpose, fix two construction parameters~$p, q \in \mathbb{N}$, where $p$ specifies which anchors (and dials) we use when constructing the gadget and $q$~specifies how many possible choices shall be modeled.
Herein, we require that $q$~is a power of two.
For example, in the instance-selection gadget we will set $p = 2$ and~$q = t$.
Refer to \cref{fig:selcons} for an example of the construction.

\begin{figure}[t]
  \centering
  \begin{tikzpicture}[new set = alvl1, new set = blvl1, on grid, scale = .75]
    \input{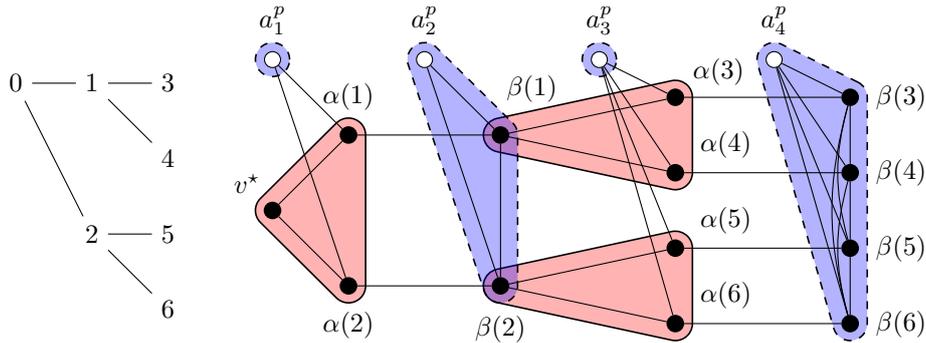}

    \begin{scope}[xshift = -4.5cm, yshift = 2.25cm]
      \graph[trie, simple]{
        0 -- {
          1 -- {
            3,
            4
          },
          2 -- {
            5,
            6
          }
        }
      };
    \end{scope}

    \newcommand\spacer{2.3cm}

    \node [bvertex, label = {[label distance = 0mm]110:$v^\star$}] (vs) at (0, 0) {};

    \node [bvertex, set = alvl1, above right = of vs, xshift = 0cm, label = {[label distance = 1mm]above:$\alpha(1)$}] (al1) {};
    \node [bvertex, set = alvl1, below right = of vs, xshift = 0cm, label = {[label distance = 1mm]below:$\alpha(2)$}] (al2) {};

    \node [avertex, above left = of al1, label = {[label distance = 1mm]above:$a^p_1$}] (an1) {};

    \node [bvertex, set = blvl1, right = 2cm of al1, label = {[label distance = 2mm, xshift = -.6mm]87:$\beta(1)$}] (be1) {};
    \node [bvertex, set = blvl1, right = 2cm of al2, label = {[label distance = 1.5mm]below:$\beta(2)$}] (be2) {};

    \node [avertex, above left = of be1, label = {[label distance = 1mm]above:$a^p_2$}] (an2) {};

    \node [avertex, right = \spacer of an2, label = {[label distance = 1mm]above:$a^p_3$}] (an3) {};
    \node [avertex, right = \spacer of an3, label = {[label distance = 1mm]above:$a^p_4$}] (an4) {};

    \node [bvertex, right = \spacer of be1, yshift = .5cm, label = 20:$\alpha(3)$] (al3) {};
    \node [bvertex, right = \spacer of be1, yshift = -.5cm, label = {[label distance = 1mm]20:$\alpha(4)$}] (al4) {};
    \node [bvertex, right = \spacer of be2, yshift = .5cm, label = {[label distance = 1mm]20:$\alpha(5)$}] (al5) {};
    \node [bvertex, right = \spacer of be2, yshift = -.5cm, label = {[label distance = 1mm]20:$\alpha(6)$}] (al6) {};

    \node [bvertex, right = \spacer of al3, label = {[label distance = 1mm]right:$\beta(3)$}] (be3) {};
    \node [bvertex, right = \spacer of al4, label = {[label distance = 1mm]right:$\beta(4)$}] (be4) {};
    \node [bvertex, right = \spacer of al5, label = {[label distance = 1mm]right:$\beta(5)$}] (be5) {};
    \node [bvertex, right = \spacer of al6, label = {[label distance = 1mm]right:$\beta(6)$}] (be6) {};

    \graph{
      (vs) -- {(al1), (al2)};

      (an1) -- [bend right = 0] (alvl1);

      (alvl1) -- (blvl1);

      (an2) -- [bend right = 0] (blvl1);
      (be1) -- (be2);

      (be1) -- {(al3), (al4)};
      (be2) -- {(al5), (al6)};

      (an3) -- [bend right = 0] {(al3), (al4), (al5), (al6)};

      (an4) -- [bend right = 0] {(be3), (be4), (be5), (be6)};
      {(be3), (be5)} -- {(be4), (be6)};
      (be4) -- (be5);
      (be3) -- [bend right = 15] {(be5), (be6)};
      (be4) -- [bend right = 15] (be6);

      {(al3), (al4), (al5), (al6)} -- {(be3), (be4), (be5), (be6)};
    };

    \newcommand\dialr{3mm}
    \newcommand\forbr{3mm}

    \begin{pgfonlayer}{bg}
      \draw[forb]
      \hedgeiii{al1}{al2}{vs}{\forbr}
      \hedgeiii{be1}{al3}{al4}{\forbr}
      \hedgeiii{be2}{al5}{al6}{\forbr};

      \draw[dial]
      \hedgei{an1}{\dialr}
      \hedgeiii{an2}{be1}{be2}{\dialr}
      \hedgei{an3}{\dialr}
      \hedgeiii{an4}{be3}{be6}{\dialr};
    \end{pgfonlayer}

  \end{tikzpicture}
  \caption{Left: An example for the tree $T$ used in $\selcon(p, q)$
    for $q = 4$.
    Right: Parts of the selection gadget constructed by $\selcon(p, q)$ using $T$ where $q = 4$ and $M$ is a $P_3$.
    Vertex~$v^\star$ is the activator vertex and vertices $\beta(3)$ through $\beta(6)$ are the choice vertices.
    Copies of $M$ are highlighted with a red region with solid outline, that is, the corresponding vertices have been made exclusive.
    The parts of the dials of the anchors that are used in the construction are highlighted with a blue region with dashed outline.
    Not shown are the gadgets used for fixing the anchors in $A$ and parts of the dials that possibly were previously constructed.}
  \label[figure]{fig:selcons}
\end{figure}

We introduce a new vertex~$v^*$.
Our goal is to construct a structure in which, starting from a trivial \pipartition~$(A, B)$, putting $v^* \in B$ triggers an avalanche of pushes according to a path in a binary-tree-like structure.
To this end, fix a rooted binary tree $T$ with $q$~leaves (corresponding to the $q = t$~instances of \pCIS\ for the instance-selection gadget).
Say a vertex in~$T$ is on \emph{level} $i \in [\log q]$ if its distance from the root is~$i$.
For $i \in [\log q]$, $L_i$ denotes the set of vertices at level~$i$.
The tree~$T$ will not be part of the constructed graph; we use it only as a scaffold to define the actual vertices in the graph.

For each vertex $v \in V(T)$ except the root, proceed as follows.
Introduce two vertices~$\alpha(v), \beta(v)$ into~$G$.
Let $i$ be the level of~$v$.
Connect $\alpha(v)$ to both~$a^p_{2i - 1}$ and~$\beta(v)$.
Make~$\beta(v)$ join~$D_{2i}^p$.
Next, for each vertex~$u \in L_i$, $i \in \{0, \ldots, \log q\}$, let $v$, $w$ be the two children of~$u$ in~$T$ and make $\beta(u), \alpha(v), \alpha(w)$ exclusive.
If $i = 0$, then let $v, w$ be the two vertices in level~$1$ in $T$ and make $v^*, \alpha(v), \alpha(w)$ exclusive instead.
This completes the construction of the selection gadget.
Vertex~$v^*$ is a dial vertex.
Each $\alpha(v)$, $v \in V(T)$, is a volatile vertex.
Each $\beta(v)$, $v \in V(T)$, is a dial vertex.

We now verify that the above construction maintains all invariants.
Observe that \cref{inv:anchor-helper} is maintained as none of the previously introduced helper vertices receive new edges.
\Cref{inv:dials}-(i) is maintained: Each dial still induces a clique, because each introduced dial vertex joined some dial, except for $v^*$ which is not adjacent to any other dial vertex.
\Cref{inv:dials}-(ii) is maintained as well, because each dial vertex is made adjacent either only to some vertices of one specific dial, or to non-dial vertices.
\Cref{inv:volatile}-(i) is maintained since no vertex joins the referenced dials.
The only volatile vertices that have been introduced are the vertices $\alpha(v)$, $v \in V(T)$.
These vertices have been made adjacent to only one anchor $a_\ell^p$ where $\ell$ is odd.
By \cref{inv:volatile}-(i) $a_\ell^p$'s dial is a singleton and thus \cref{inv:volatile}-(ii) is maintained.

Denote the constructed gadget as $\selcon(p, q)$, and say that $v^*$ is the \emph{activator vertex}, and that the vertices in $\{\beta(v) \mid v \in L_{\log q}\}$ are the \emph{choice vertices}.
We fix an arbitrary order of the choice vertices, so that we may speak of the $i$th choice vertex without confusion.

\begin{lemma}\label[lemma]{lem:selcons}
  Let $G'$ be the graph before applying $\selcon(p, q)$ and $G$
  the graph afterwards.

  \begin{lemenum}
  \item If \pipartition~$(A, B)$ has at most~$d$ clusters in~$G[A]$
    and the activator vertex is in~$B$, then at least one choice
    vertex is in~$B$.
  \item%
    If there is a \pipartition~$(A', B')$ for~$G'$ with~$d$ clusters in~$G'[A']$, then there is a \pipartition~$(A, B)$ for~$G$ with $d + 1$ clusters, where the activator vertex is a singleton cluster and each choice vertex is in~$A$.
    If $(A', B')$ is \friendly{\cal D} for some dial set $\cal D$, then $(A, B)$ is \friendly{\cal D}.
  \item If $G'$ has a \pipartition~$(A', B')$ that is \friendly{D_i^p}
    and such that $G'[A']$ contains at most~$d$ clusters, then, for
    each $i \in [q]$, there is a \pipartition~$(A, B)$ of~$G$, such
    that graph~$G[A]$ contains at most~$d$ clusters, and out of all
    choice vertices only the~$i$th one is in~$B$ (and, necessarily,
    the activator vertex is in~$B$). Moreover, the choice vertex that
    is contained in~$B$ is isolated in~$G[B]$.
\end{lemenum}
\end{lemma}
\begin{proof}
  \proofparagraphf{(i).} %
  Note that there are
  $d$~anchors and each anchor is in~$A$. Hence, each cluster in~$G[A]$
  consists of an anchor and possibly further vertices. By assumption,
  we have $v^* \in B$. We now prove by induction that for each
  $i \in [\log q]$, there is at least one vertex~$v \in L_i$ with
  $\beta(v) \in B$, yielding the statement. Consider the case $i = 1$.
  Let $u, v \in L_1$. As $v^* \in B$, we have that either $\alpha(u)$
  or $\alpha(v)$ is in~$A$; say $\alpha(u) \in A$ and the other case is
  symmetric. Since $\alpha(u)$ is adjacent to both $a^2_{2i - 1}$
  and $\beta(u)$, we have $\beta(u) \in B$ as, otherwise,
  vertices~$a^2_{2i - 1}, \alpha(u), \beta(u)$ would form an
  induced~$P_3$ in~$G[A]$. That is, the statement holds if $i = 1$.
  Now suppose that for some $u \in L_{i - 1}$, $i > 1$, we have
  $\beta(u) \in B$. Consider the children~$v, w$ of~$u$ in~$T$.
  Since $\beta(u), \alpha(v), \alpha(w)$ are made exclusive, either
  $\alpha(v)$ or~$\alpha(w)$ is in~$A$. Say $\alpha(v) \in A$ and the
  other case is symmetric. Note that $\alpha(v)$ is adjacent to both
  $a^2_{2i - 1}$ and $\beta(v)$. Hence, $\beta(v) \in B$ since,
  otherwise, $a^2_{2i -1 }, \alpha(v), \beta(v)$ would induce a $P_3$
  in~$G[A]$. Thus, indeed, for some $v \in L_i$ we
  have~$\beta(v) \in B$.

  \proofparagraph{(ii).} For the second statement, let $(A', B')$ be a \pipartition\
  for~$G'$. %
  Construct a \pipartition~$(A, B)$ for~$G$ as follows.
  Put $(A, B) = (A', B')$.
  Put $v^* \in A$.
  For each $v \in T$ at level $i > 0$, put $\alpha(v) \in B$ and $\beta(v) \in A$.
  This concludes the construction.
  Clearly, each choice vertex is in~$A$, as required.

  We claim that $G[A]$ is a cluster graph with~$d + 1$ clusters.
  Note that~$v^*$ is not adjacent to any vertex in~$A$ and hence constitutes a singleton cluster.
  By \cref{inv:volatile}, each anchor~$a_i^j$ whose dial~$D_i^j$ is not a singleton is contained in a cluster in~$G[A]$ whose vertex set is contained in~$D_i^j$.
  Apart from~$v^*$, the only vertices from the construction placed into~$A$ are contained in dials which are not singletons, and hence, $G[A]$ is a cluster graph with $d + 1$~clusters.
  From this fact it is also immediate that, if $(A, B)$ is a \pipartition\ and $(A', B')$ is friendly with respect to the dials $\cal D$, then $(A, B)$ is friendly with respect to the dials $\cal D$.

  To conclude the proof of statement~(ii), we apply \cref{lem:exclusive-completeness} to show that $G[B] \in \Pi$.
  Note that all vertices in $B \setminus B'$ are part of a triple of vertices that has been made exclusive by \selcon.
  Furthermore, no two vertices between two different triples have been made adjacent by \selcon.
  Thus, it is enough to show that the conditions in \cref{lem:exclusive-completeness} are satisfied.
  Note that each triple of exclusive vertices contains one vertex from~$A$.
  Furthermore, for each vertex $v \in T$, $\alpha(v)$ is connected in~$B$ only to $\alpha(w)$ where $w$ is the sibling of~$v$ in~$T$.
  Thus, $G[B] \in \Pi$.

  \proofparagraph{(iii).}
  For the third statement, let $(A', B')$ be a \pipartition\ for~$G'$.
  Given $i \in [q]$, we construct a \pipartition~$(A, B)$ for~$G$ as follows (as before, we ignore helper vertices).
  Set $A = A'$, $B = B'$, and note that $v^* \in B$.
  Pick a path~$P$ in~$T$ from the root~$r$ to the leaf~$v_\ell$ corresponding to the $i$th choice vertex, call it~$\beta(v_\ell)$.
  For each vertex~$v \in V(T) \setminus \{r\}$, if $v \in V(P)$, put $\alpha(v) \in A$ and $\beta(v) \in B$.
  Otherwise, if $v \notin V(P)$, put $\alpha(v) \in B$ and $\beta(v) \in A$.
  Clearly, $\beta(v_\ell) \in B$ and $\beta(v_\ell)$ is isolated in~$G[B]$, as required.

  We first show that $G[A]$ is a cluster graph with at most~$d$ clusters.
  Suppose that $G[A]$ contains an induced~$P_3$, say~$Q$.
  Clearly, $Q$ contains at least one vertex introduced by the construction \selcon.
  As all helper vertices are in~$B$, path $Q$ does not involve helper vertices.
  By \cref{inv:dials}, $Q$ involves a volatile vertex; that is, $\alpha(v) \in V(Q)$ for some~$v \in V(T)$.
  Moreover, $v \in V(P)$ as otherwise $\alpha(v) \in B$.
  By construction, apart from helper vertices $\alpha(v)$ is adjacent in $G$ only to $\beta(v)$, $\alpha(w)$ (where $w$ is $v$'s sibling in~$T$), and $a_{2i -1}^p$, where $i$ is $v$'s level in~$T$.
  As $\beta(v), \alpha(w) \in B$ by definition of~$(A, B)$, path~$Q$ contains~$a_{2i - 1}^p$.
  As $D_{2i - 1}^p$ is a singleton by \cref{inv:volatile}, that is, $D_{2i - 1}^p = \{a_{2i - 1}^p\}$, and since $(A', B')$ is \friendly{D_{2i - 1}^p}, we have that $D_{2i - 1}^p$ is a singleton cluster in~$G[A']$.
  Recall that, by construction, the only new vertices adjacent to $a_{2i - 1}^p$ are vertices~$\alpha(x)$ for $x \in L_i$.
  By definition of $(A, B)$, only one of these vertices~$\alpha(x)$ is in~$A$, namely~$\alpha(v)$.
  Hence, $\alpha(v)$ is the only neighbor of $a_{2i - 1}^p$ in $G[A]$, a contradiction to~$Q$ being an induced~$P_3$ in $G[A]$.
  To see that there are at most $d$~clusters, observe that each vertex in~$A$ is adjacent to one of the anchors and thus, there are at most~$d$ connected components.

  To show that $G[B] \in \Pi$ by \cref{lem:exclusive-completeness} it remains to show that each triple of three vertices that were made exclusive include one vertex in~$A$, and that they are adjacent in~$G[B]$ only to some subset of themselves (apart from helper vertices).
  By construction, the only triples of exclusive vertices are~$\beta(u)$, $\alpha(v)$, $\alpha(w)$ for some $u \in V(T)$ and its children~$v, w$.
  (The case of $v^*$ is analogous.)
  Either $v$ or~$w$ is not in~$V(P)$, and hence, either~$\alpha(v)$ or~$\alpha(w)$ is in~$A$, that is, at least one vertex in the triple is in~$A$, as required.
  It remains to show the condition on their adjacencies.
  If~$\beta(u) \in B$, then $\alpha(u) \in A$ and, hence, regardless of whether $\beta(u) \in B$, a possible connection outside of the triple must involve~$\alpha(v)$ or~$\alpha(w)$, say it involves $\alpha(v)$.
  Vertex~$\alpha(v)$ is only adjacent to $\beta(u)$, to some anchor, and to $\beta(v)$.
  If $\alpha(v) \in B$, then $\beta(v) \in A$.
  Thus, indeed $\beta(u)$, $\alpha(v)$, $\alpha(w)$ are only connected within themselves in~$G[B]$ (apart from helper vertices).
  This shows that $(A, B)$ is a \pipartition\ with $d$ clusters in~$G[A]$.
\end{proof}

\subsubsection*{Instance Selection}
As mentioned before, to construct the \emph{instance-selection gadget}, we apply $\selcon(2, t)$.
Fix a bijection~$\phi$ from the set of instances $[t]$ to the choice vertices produced by the construction.
We use $\phi$ later to denote the choice vertex corresponding to an instance.

\subsection{Vertex Selection}\label{sec:nopk-vertsel}

We apply~$\selcon$ to create vertex-selection gadgets for each instance and each color.
Each vertex-selection gadget selects one vertex of the gadget's color into the independent set when activated by putting its activator vertex into~$B$ (which will be effected by the instance-selection gadget).
The vertex-selection gadgets for each instance are distinct, but they use dials which are shared by all instances.
See \cref{fig:composition-struc} for illustration.

In the first part of the construction of the vertex-selection gadgets, for each instance $r \in [t]$ and color $i \in [k]$, we carry out $\selcon(3 + i, n)$ to introduce the \emph{vertex-selection gadget} for instance~$r$ and color~$i$.
Let $\psi^*_{r, i}$ be the corresponding activator vertex and fix a bijection~$\psi_{r, i}$ from the vertices $V(G_r)$ of color~$i$ to the choice vertices.
Make $\psi^*_{r, i}$ join $D^{3}_{1 + i}$.
Intuitively, if the activator vertex~$\psi^*_{r, i}$ is put into~$B$, the subgraph constructed by~$\selcon(3 + i, n)$ enforces the push of a choice vertex into~$B$, which by bijection~$\psi_{r, i}$ corresponds in a one-to-one fashion to the vertices of color~$i$ in instance~$r$.
This is how the selection of an independent-set vertex is modelled.

\begin{figure}[t]
  \centering
  \begin{tikzpicture}
    \input{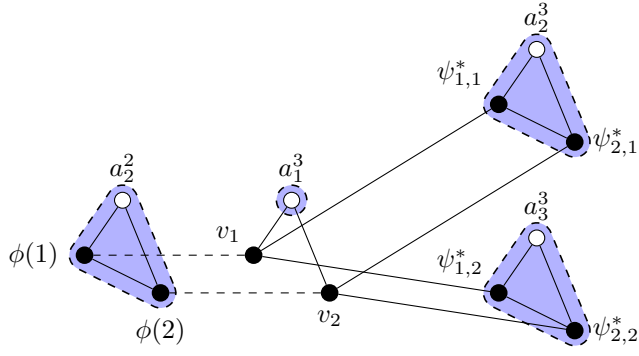}
    \newcommand\dialvdiffa{0.5cm}
    \newcommand\dialhdiffa{2cm}
    \newcommand\shiftdiff{.5cm}

    \node[avertex, label = above:$a_{2}^2$] (a2) at (0,0) {};
    \node[bvertex, below = \dialvdiffa of a2, xshift = -\shiftdiff, label = {[label distance = 1mm]left:$\phi(1)$}] (phi1) {}
    edge (a2);
    \node[bvertex, below = 2 * \dialvdiffa of a2, xshift = \shiftdiff, label = {[label distance = 1mm]below:$\phi(2)$}] (phi2) {}
    edge (phi1)
    edge (a2);

    \node[avertex, right = \dialhdiffa of a2, label = above:$a_{1}^3$] (a31) {};
    \node[bvertex, below = \dialvdiffa of a31, xshift = -\shiftdiff, label = {[label distance = -.5mm]above left:$v_1$}] (v1) {}
    edge[dashed] (phi1)
    edge (a31);
    \node[bvertex, below = 2 * \dialvdiffa of a31, xshift = \shiftdiff, label = below:$v_2$] (v2) {}
    edge[dashed] (phi2)
    edge (a31);

    \node[avertex, right = 1.5 * \dialhdiffa of a31, yshift = 4*\dialvdiffa, label = above:$a_2^3$] (a32) {};
    \node[bvertex, below = \dialvdiffa of a32, xshift = -\shiftdiff, label = {[label distance = 0mm]above left:$\psi_{1, 1}^*$}] (psi11) {}
    edge (a32)
    edge (v1);
    \node[bvertex, below = 2 * \dialvdiffa of a32, xshift = \shiftdiff, label = right:$\psi_{2, 1}^*$] (psi21) {}
    edge (a32)
    edge (psi11)
    edge (v2);

    \node[avertex, right = 1.5 * \dialhdiffa of a31, yshift = -1*\dialvdiffa, label = above:$a_3^3$] (a33) {};
    \node[bvertex, below = \dialvdiffa of a33, xshift = -\shiftdiff, label = {[label distance = 0mm]above left:$\psi_{1, 2}^*$}] (psi12) {}
    edge (a33)
    edge (v1);
    \node[bvertex, below = 2 * \dialvdiffa of a33, xshift = \shiftdiff, label = right:$\psi_{2, 2}^*$] (psi22) {}
    edge (a33)
    edge (psi12)
    edge (v2);

    \begin{pgfonlayer}{bg}
      \newcommand\dialr{2mm}
      \draw[dial]
      \hedgeiii{a2}{phi2}{phi1}{\dialr}
      \hedgei{a31}{\dialr}
      \hedgeiii{a32}{psi21}{psi11}{\dialr}
      \hedgeiii{a33}{psi22}{psi12}{\dialr};
    \end{pgfonlayer}

  \end{tikzpicture}
  \caption{%
    Illustration of the second part of the construction of the vertex-selection gadgets in which we connect them to the instance-selection gadget.
    In this example, there are two instances, represented by the choice vertices $\phi(1)$ and $\phi(2)$ of the instance-selection gadget.
    Each instance has two colors; $\psi_{1,1}^*$, $\psi_{2,1}^*$, $\psi_{1,2}^*$, and $\psi_{2,2}^*$ are the activator vertices of the corresponding vertex-selection gadgets.
    Anchor vertices are white.
    Dials are shown by blue regions with dashed outlines.
    Dashed edges mean that the endpoints have been made exclusive.}
  \label{fig:vertex-select}
\end{figure}

In the second part of the construction of the vertex-selection gadgets, we introduce a way to activate the vertex-selection gadgets of all colors if some instance $r \in [t]$ has been chosen.
See \cref{fig:vertex-select} for an illustration.
To achieve this, for each $r \in [t]$, we carry out the following steps.
Introduce a volatile vertex~$v_r$.
Make $\phi(r)$ and $v_r$ exclusive.
Make $v_r$ adjacent to $a^3_1$ and, for each $i \in [k]$, make $v_r$ adjacent to $\psi^*_{r, i}$.
This concludes the construction of the vertex-selection gadgets.

Intuitively, the selection of instance~$r$ is indicated by placing~$\phi(r) \in B$.
Since $\phi(r)$, and $v_r$ are exclusive, $v_r \in A$.
Vertex~$v_r$ forms a $P_3$ with $a^3_1$ and each~$\psi^*_{r, i}$.
Hence, the activator vertices~$\psi^*_{r, i}$ of each vertex-selection gadget for instance~$r$ are in~$B$.
This enforces the selection of an independent-set vertex from each color.

We now verify after this construction that the invariants are maintained.
\Cref{inv:anchor-helper} is maintained because it is maintained by the operations of $\selcon$ and making vertices exclusive.
\Cref{inv:dials,inv:volatile} is maintained in the first part of the construction because $\selcon$ maintains these invariants.
In the second part of the construction, no dial vertices are added, giving \cref{inv:dials} and \cref{inv:volatile}-(i).
\Cref{inv:volatile}-(ii) holds for~$a^3_1$ since $D^{3}_{1}$ is a singleton.
For all the other anchors \cref{inv:volatile}-(ii) holds because the invariant was satisfied before the second part of the construction, and because each volatile vertex~$v_r$ is only made adjacent to the single anchor~$a^3_1$.
Thus, the construction of the vertex-selection gadgets maintains \cref{inv:anchor-helper,inv:dials,inv:volatile}.

\begin{lemma}\label[lemma]{lem:vertsel}
  Let $G$ be the graph after constructing the vertex-selection gadgets.
  \begin{lemenum}
  \item If $G$ admits a \pipartition~$(A, B)$ with $d$ clusters
    in~$G[A]$, then there is an instance~$r \in [t]$ such that, for
    each color $i \in [k]$, there is at least one vertex $v \in V(G_r)$
    of color~$i$ satisfying that $\psi_{r, i}(v) \in B$.
  \item \looseness=-1 For each instance $s \in [t]$ and each vertex
    subset~$V' \subseteq V(G_s)$ containing exactly one vertex of each
    color, there is a \pipartition~$(A, B)$ for $G$ such that $G[A]$
    contains at most~$d$ clusters, $\psi_{s, i}(V') \subseteq B$, and
    all other choice vertices of each vertex-selection gadget are
    in~$A$. Moreover, the choice vertices that are contained in~$B$ are
    isolated in~$G[B]$.
\end{lemenum}
\end{lemma}
\begin{proof}
  \proofparagraphf{(i).}
  There are $d$ anchors in $G[A]$ and the activator vertex of the instance-selection gadget is not adjacent to any of the anchors.
  Thus, the activator vertex is in $B$. By \cref{lem:selcons}-(i), it follows that for at least one instance $r \in [t]$, we have $\phi(r) \in B$.
  Since $\phi(r)$ and $v_r$ are exclusive, we have $v_r \in A$.
  Since, for each $i \in [k]$, vertices $a^3_1$, $v_r$, and $\psi^*_{r, i}$ form a~$P_3$, we have that $\psi^*_{r, i} \in B$ for each $i \in [k]$.
  By \cref{lem:selcons}-(i), it follows that, for each $i \in [k]$, there is one choice vertex of the $i$th vertex-selection gadget that is in~$B$.
  Thus, for each $i \in [k]$, there is a vertex~$v \in V(G_r)$ of color~$i$ such that~$\psi_{r, i}(v) \in B$, as required.

  \proofparagraph{(ii).}
  Without loss of generality, assume that the instance-selection gadget has been constructed first, and the vertex-selection gadgets have been constructed in ascending order of instances and then colors.
  We start by showing that a partial \pipartition\ with the required properties exists after the first part of the construction, and then proceed to treat the second part.
  For the first part, we show that a suitable partial \pipartition\ exists after each call to $\selcon$.

  Let $G_0$ be the graph
  obtained after introducing the instance-selection gadget. Before
  introducing any selection gadget, the graph has a trivial
  \pipartition\ with $d$ clusters in~$A$ that is friendly with respect to each
  dial. By \cref{lem:selcons}-(iii), there is a
  \pipartition~$(A_0, B_0)$ for $G_0$ such that $G_0[A_0]$ has
  $d$~clusters, and out of all choice vertices only the $r$th one is
  in~$B$. Furthermore, this \pipartition\ is \friendly{D^{3 + i}_\ell}, $i \in [k]$.

  In the following, let $s \in [t]$ be the instance for which we want
  to construct a \pipartition. Let $G_{s -1}$ be the graph obtained
  after introducing all vertex-selection gadgets for instances
  in~$[s - 1]$. By iteratively applying \cref{lem:selcons}-(ii),
  starting with $G_0$ and $(A_0, B_0)$, we obtain that there is a
  \pipartition~$(A_{s - 1}, B_{s - 1})$ for~$G_{s - 1}$ such that, for
  each vertex-selection gadget, each activator vertex is in~$A$ (in a
  cluster together with the dial it joined) and each choice
  vertex is in~$A$. Since we joined the activator vertices to some
  dials, $G_{s - 1}[A_{s - 1}]$ has $d$ clusters. Moreover, since
  $(A_0, B_0)$ is \friendly{D^{3 + i}_\ell},
  $i \in [k]$, $(A_{s - 1}, B_{s - 1}$) is friendly with respect to these dials as well.

  Let $V' \subseteq V(G_s)$ as in the statement of the lemma.
  For each $i \in [k]$, denote by $v'_i \in V'$ the vertex of color~$i$ in~$V'$ and let $G_{s, i}$ be the graph obtained after introducing the vertex-selection gadget for instance~$s$ and color~$i$ (in the first part of the construction of the vertex-selection gadgets).
  By induction on~$i$ and by \cref{lem:selcons}-(iii), we obtain that $G_{s, i}$ admits a \pipartition~$(A_{s, i}, B_{s, i})$ with $d$~clusters in~$G_{s, i}[A_{s, i}]$ such that, for each $j \in [i]$, we have $\psi^*_{s, j} \in B$, $\psi_{s, j}(v'_j) \in B$, and such that all other choice vertices in any vertex-selection gadget are in~$A$.
  Moreover, $(A_{s, i}, B_{s, i})$ is \friendly{D^{3 + j}_\ell}, $j \in \{i, i + 1, \ldots, k\}$ (whence we can apply induction).

  Let $G_{t}$ be the graph obtained after introducing all vertex-selection gadgets for instances in~$[t] \setminus [s - 1]$.
  By applying iteratively \cref{lem:selcons}-(ii) to $G_{s, k}$ and $(A_{s, k}, B_{s, k})$ we obtain a \pipartition~$(A_t, B_t)$ for $G_t$ analogously to the \pipartition\ for $G_{r - 1}$.
  Hence, the statement of the lemma holds after the first part of the construction of the vertex-selection gadgets.
  It remains to incorporate the second part, that is, to incorporate vertices $v_r$, $r \in [t]$, into~$(A_t, B_t)$.
  Construct a \pipartition~$(A, B)$ for~$G$ from~$(A_t, B_t)$ as follows.
  Put $A = A_t$, $B = B_t$.
  For each $r \in [t] \setminus \{s\}$, put $v_r \in B$.
  Finally, put $u_s \in B$ and $v_s \in A$.

  We claim that $G[A]$ is a cluster graph with at most $d$ clusters.
  Recall that $G_t[A_t]$ contains at most~$d$ clusters (corresponding to the $d$ anchors).
  Thus, $G[A]$ has at most $d$~connected components since, for each $r \in [t]$, vertex~$v_r$ is connected to some anchor.
  To show that~$G[A]$ does not contain an induced~$P_3$, it is enough to show that, for each $r \in [t]$, either $v_r \in B$ or, for all $i \in [k]$, $\psi^*_{r, i} \in B$.
  The fact that $v_r \in B$ is trivial for~$r \neq s$; otherwise, if $r = s$, we have $\psi^*_{r, i} \in B$ by the construction of~$(A_{s, i}, B_{s, i})$.

  Note that, for each $r \in [t]$, either $\phi(r)$ or $v_r$ is in~$A$.
  Hence, by \cref{lem:exclusive-completeness} (and \cref{def:two-exclusive}), to show that $G[B] \in \Pi$, it suffices to prove, for each~$r \in [t]$, the property that vertices $\phi(r)$ and $v_r$ are adjacent in~$G[B]$ only to each other (apart from helper vertices).
  If $r \neq s$, we have $v_r \in A$ and, thus, by construction of $(A_0, B_0)$ according to \cref{lem:selcons}~(iii), that $\phi(r)$ is an isolated vertex in~$G_0[B_0]$, giving the required property.
  If $r = s$, then $\phi(s) = \phi(r) \in A$.
  For the incident edges of $v_r$, by construction of~$(A_{s, i}, B_{s, i})$, $i \in [k]$, according to \cref{lem:selcons}~(ii), for each $i \in [k]$ we have $\psi^*_{s, i} \in A$, that is, none of the non-helper neighbors of $v_r$ is in $B$.
  Thus indeed by \cref{lem:exclusive-completeness} $G[A] \in \Pi$, finishing the proof.
\end{proof}

\subsection{Verification}\label{sec:nopk-verif} %

We now construct the verification gadgets. It is again crucial that gadgets share clusters (anchors) in order to keep the overall number of clusters in~$A$ small.
How the clusters are shared is indicated in \cref{fig:composition-struc}.
In essence, we fix a vertex ordering for each color and each instance, and an edge ordering for each instance.
Then, for each color, we use one vertex gadget that represents all the first vertices of that color, one vertex gadget that represents all the second vertices of that color, and so on.
Similarly for the edges: The first edge gadget represents all the first edges of each instance. The second edge gadget represents all the second edges and so on.

\begin{figure}[t]
  \centering
  \begin{tikzpicture}[on grid]
    \input{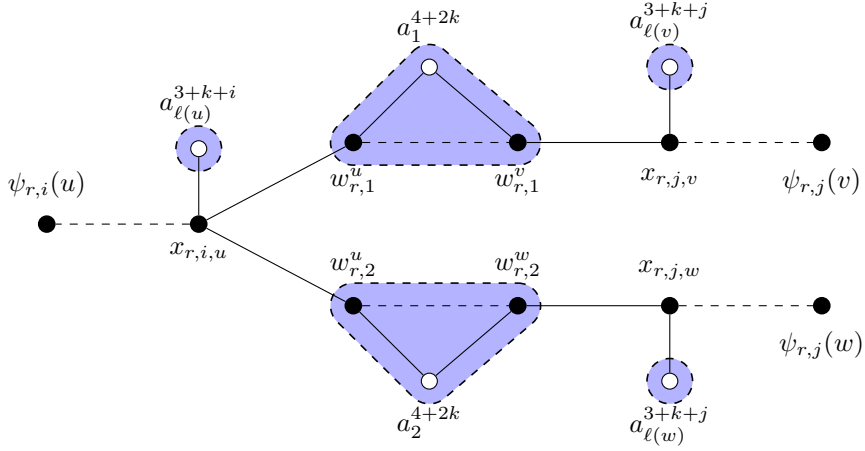}
    \coordinate (origin) at (0,0);

    \node[bvertex, above left = of origin, label = {[label distance = 1mm]below:$w^u_{r, 1}$}] (wru1)  {};
    \node[bvertex, above right = of origin, label = {[label distance = 1mm]below:$w^v_{r, 1}$}] (wrv1) {};

    \node[bvertex, below left = of origin, label = {[label distance = 1mm]above:$w^u_{r, 2}$}] (wru2) {};
    \node[bvertex, below right = of origin, label = {[label distance = 1mm]above:$w^w_{r, 2}$}] (wrw2) {};

    \node[avertex, above right = of wru1, label = {[label distance = 1mm]above:$a^{4 + 2k}_{1}$}] (ae1) {};
    \node[avertex, below right = of wru2, label = {[label distance = 1mm]below:$a^{4 + 2k}_{2}$}] (ae2) {};

    \node[bvertex, left = 3cm of origin, label = {[label distance = 0mm]below:$x_{r, i, u}$}] (u) {};
    \node[bvertex, right = 2cm of wrv1, label = {[label distance = 1mm]below:$x_{r, j, v}$}] (v) {};
    \node[bvertex, right = 2cm of wrw2, label = {[label distance = 1mm]above:$x_{r, j, w}$}] (w) {};

    \node[avertex, above = of u, label = {[label distance = 1mm]above:$a^{3 + k + i}_{\ell(u)}$}] (au) {};
    \node[avertex, above = of v, label = {[label distance = 1mm]above:$a^{3 + k + j}_{\ell(v)}$}] (av) {};
    \node[avertex, below = of w, label = {[label distance = 1mm]below:$a^{3 + k + j}_{\ell(w)}$}] (aw) {};

    \node[bvertex, left = 2cm of u, label = {[label distance = 1mm]above:$\psi_{r, i}(u)$}] (psiu) {};
    \node[bvertex, right = 2cm of v, label = {[label distance = 1mm]below:$\psi_{r, j}(v)$}] (psiv) {};
    \node[bvertex, right = 2cm of w, label = {[label distance = 1mm]below:$\psi_{r, j}(w)$}] (psiw) {};

    \graph{
      (ae1) -- {(wru1), (wrv1)};
      (ae2) -- {(wru2), (wrw2)};

      (wru1) --[dashed] (wrv1);
      (wru2) --[dashed] (wrw2);

      (u) -- {(wru1), (wru2), (au)};
      (v) -- {(wrv1), (av)};
      (w) -- {(wrw2), (aw)};

      {(u), (v), (w)} -- [dashed] {(psiu), (psiv), (psiw)};
    };

    \begin{pgfonlayer}{bg}
      \newcommand\dialr{3mm}
      \draw[dial] \hedgei{au}{\dialr};
      \draw[dial] \hedgei{av}{\dialr};
      \draw[dial] \hedgei{aw}{\dialr};
      \draw[dial] \hedgeiii{wru1}{ae1}{wrv1}{\dialr};
      \draw[dial] \hedgeiii{wru2}{wrw2}{ae2}{\dialr};
    \end{pgfonlayer}
  \end{tikzpicture}
  \caption{%
    Parts of the verification gadget for three vertices and two edges in instance~$r$.
    There are three vertices~$u$, $v$, and $w$ in instance~$r$.
    Vertex $u$ is of color $i$, and vertices $v$ and~$w$ of color $j$ (their indices that are used in the construction are denoted by $\ell(u), \ell(v), \ell(w)$, respectively).
    There are two edges, $e_1 = \{u, v\}$ and $e_2 = \{u, w\}$.
    Dials are highlighted with blue regions with dashed outline.
    A dashed edge means that its endpoints have been made exclusive---only one of the endpoints can be in~$B$.}
  \label[figure]{fig:verification}
\end{figure}

The working principle of the gadgets is as follows.
See \cref{fig:composition-struc,fig:verification} for illustration.
Each vertex gadget consists of a singleton dial and a vertex for each instance that could be pushed into that dial.
Selecting a vertex~$v$ via a vertex-selection gadget will make it necessary to push the vertex corresponding to~$v$ into the cluster containing the dial of its vertex gadget.
Next, each edge gadget consists of a dial and, for each instance, two vertices corresponding to the endpoints of an edge in that instance.
The push of a vertex to the dial of a vertex gadget creates a $P_3$ in~$A$ for each incident edge~$e$, necessitating further pushes.
Namely, we are required to push a vertex out of the dial of the edge gadget in~$A$ representing~$e$.
Pushing the corresponding vertex for the other endpoint of~$e$ into~$B$ will complete a forbidden induced subgraph, yielding that no two endpoints of an edge are selected.
This is achieved by making the two corresponding vertices in the constructed graph exclusive.

The formal construction is as follows.
See \cref{fig:verification} for an illustration.
For each $r \in [t]$, let $E(G_r) = \{e_1, \ldots, e_m\}$.
(If there are fewer than $m$ edges, duplicate an arbitrary edge as needed.)
For each $j \in [m]$, perform the following steps towards constructing the $j$th \emph{edge gadget}.
Let $e_j = \{u, v\}$.
Introduce two vertices $w^u_{r, j}, w^v_{r, j}$ into $G$.
Make $w^u_{r, j}$ and $w^v_{r, j}$ exclusive.
Make $w_{r, j}^u$ and $w^v_{r, j}$ join $D^{4 + 2k}_{j}$ (they are thus dial vertices).

We furthermore need for each vertex a \emph{vertex gadget}, which is constructed for each instance~$r \in [t]$, and each color $i \in [k]$ as follows.
Fix an arbitrary ordering of the vertices of color~$i$ in $G_r$ and say the \emph{index} of a vertex is its index in that ordering.
For each vertex $v \in V(G_r)$ of color~$i$, introduce a vertex~$x_{r, i, v}$ into $G$.
Make $\phi_{r, i}(v)$ and $x_{r, i, v}$ exclusive.
Make $x_{r, i, v}$ adjacent to~$a^{3 + k + i}_\ell$, where $\ell$ is the index of~$v$.
Vertex $x_{r, i, v}$ is a volatile vertex.

Finally, connect the edge gadgets and vertex gadgets as follows.
For each instance~$r \in [t]$, perform the following steps.
Recall that $E(G_r) = \{e_1, \ldots, e_m\}$.
For each $j \in [m]$, let $i_1, i_2 \in [k]$ be the colors of the endpoints~$v_1, v_2 \in V(G_r)$ of~$e_j$.
Make $x_{r, i_1, v_1}$ adjacent to $w^{v_1}_{r, j}$ and make $x_{r, i_2, v_2}$ adjacent to~$w^{v_2}_{r, j}$.
This finishes the construction of the verification gadgets and concludes the construction of the graph~$G$ in our instance of \picluster.

Let us now verify that the above construction maintains the invariants.
Clearly, \cref{inv:anchor-helper} remains valid since making vertices exclusive maintains \cref{inv:anchor-helper} and otherwise no new helper vertices are introduced.
\Cref{inv:dials} remains valid since only the vertices $w_{r, j}^u$ and $w^v_{r, j}$ join a dial and each pair joins the same dial.
\Cref{inv:volatile}-(i) remains valid since it was valid before and no vertex joins any dial of the form $D^{3 + k + i}_\ell$.
\Cref{inv:volatile}-(ii) remains valid since the only vertices made adjacent to anchors are $x_{r, i, v}$ and the corresponding dial is a singleton by \cref{inv:volatile}-(i).

\begin{lemma}
  Let $G$ be the graph constructed above. The graph $G$ admits a
  \pipartition~$(A, B)$ with $d$ clusters in~$G[A]$ if and only if
  there exists an instance~$s \in [t]$ such that~$G_s$ has an
  independent set with exactly one vertex of each color.
\end{lemma}
\begin{proof}
  Assume that $G$ admits a \pipartition~$(A, B)$ with $d$ clusters in~$G[A]$.
  Note that \cref{lem:vertsel} refers to a subgraph of~$G$, the graph resulting from constructing all the vertex-selection gadgets.
  By restricting $(A, B)$ to that subgraph, from \cref{lem:vertsel}-(i) we infer that there is an instance $s \in [t]$ such that, for each color~$i \in [k]$, there is a vertex~$v_i \in V(G_s)$ such that $\psi_{s, i}(v_i) \in B$.
  We claim that $V' := \{v_i \mid i \in [k]\}$ is an independent set in~$G_s$.
  Suppose $V'$ is not an independent set and let $e_j \in E(G_s)$ be such that $e_j \subseteq V'$.
  Let $e_j = \{u, v\}$ and let $i, i'$ be the colors of $u$ and $v$, respectively.
  Since $\psi_{s, i}(u)$ and $x_{s, i, u}$ are exclusive, we have $x_{s, i, u} \in A$.
  Thus, $w^{u}_{s, j} \in B$ as, otherwise, $a^{3 + k + i}_\ell$, $x_{s, i, u}$, and $w^{u}_{s, j}$, would form an induced~$P_3$ in $G[A]$, where $\ell$ is the index of~$u$.
  Similarly, $w^{v}_{s, j} \in B$.
  However, $w^{v}_{s, j}$ and $w^{u}_{s, j}$ are exclusive and each of them is contained in~$B$.
  This contradicts the fact that~$G[B] \in \Pi$.
  Hence, $V'$ is an independent set.
  Clearly, $V'$ contains exactly one vertex of each~color.

  Now assume that for some instance $s \in [t]$, there is an
  independent set~$V' = \{v_i \mid i \in [k]\} \subseteq V(G_s)$ with
  exactly one vertex~$v_i$ of each color~$i \in [k]$. Let $G'$ be the
  graph obtained in the construction before constructing the verification gadgets. By
  \cref{lem:vertsel}-(ii), there is a \pipartition~$(A', B')$ for~$G'$
  with $d$ clusters in~$G'[A']$ such that, for each~$i \in [k]$, we
  have $\psi_{s, i}(v_i) \in B'$ (and these vertices are isolated
  in~$G[B']$), and all other choice vertices of each vertex-selection
  gadget are in~$A'$.

  We now construct a \pipartition~$(A, B)$ for~$G$ from $(A', B')$.
  Put $A = A'$ and $B = B'$.
  For each instance $r \in [t]$ including~$s$, and for each $v \in V(G_r)$, let $i$ be the color of~$v$. If $v$ is not in the independent set~$V'$, then put $x_{r, i, v} \in B$ and if $v \in V'$, then put $x_{r, i, v} \in A$ instead. For each edge $e_j \in E(G_r)$, and each of its endpoints, $v \in e_j$, if $v \notin V'$, then put $w^v_{r, j} \in A$ and if $v \in V'$, then put $w^u_{r, j} \in B$.
  Clearly, not both endpoints can be in the independent set.

  Observe that $(A, B)$ is a bipartition of~$V(G)$.
  We claim that $(A, B)$ is a \pipartition\ for $G$ with at most $d$~clusters in~$G[A]$.
  We first show that $G[A]$ is a cluster graph.
  Suppose that $G[A]$ contains an induced~$P_3$, say~$Q$.
  Since $G'[A']$ is a cluster graph, $Q$ contains a vertex in~$V(G) \setminus V(G')$.
  By \cref{inv:dials}, $Q$ involves a non-dial vertex, that is, a vertex~$v$ from a vertex gadget.
  Since $v \in A$, by definition of~$(A, B)$, we have $v = x_{s, i, v_i} \in V(Q)$ for some $v_i \in V'$.
  The only neighbors of $x_{s, i, v_i}$ in~$G$ are $\psi_{r, i}(i)$ and $w^{v_i}_{r, j}$, where $j \in J$
  for some set $J \subseteq [m]$ (apart from helper vertices).
  By definition of~$(A, B)$, each of these vertices is in~$B$, a contradiction to the existence of~$Q$.
  Hence, $G[A]$ is a cluster graph.
  To see that $G[A]$ contains at most $d$~connected components, observe that $G'[A']$ has at most $d$~connected components, one for each anchor, and each vertex in~$A \setminus A'$ is connected to an anchor in~$G[A]$.

  It remains to show that $G[B] \in \Pi$.
  Recall that $G'[B'] \in \Pi$.
  The only edges in $G$ between vertices in $V(G')$ and newly-introduced vertices in $V(G) \setminus V(G')$ are incident with either an anchor or some choice vertex of some vertex-selection gadget.
  The anchors are in~$A$ and if some of the choice vertices are in $B'$, then they are isolated in~$G'[B']$ by~\cref{lem:vertsel}~(ii).
  Thus, it is enough to show that these choice vertices and the newly-introduced vertices induce a subgraph of~$G$ that satisfies~$\Pi$.
  Since all of these vertices have been made exclusive, it is enough to show that the conditions of \cref{lem:exclusive-completeness} are satisfied for each pair that has been made exclusive.
  Each such pair has the form (i) $(w^u_{r, j}, w^v_{r, j})$ or (ii) $(\psi_{r, i}(v), x_{r, i, v})$.
  By definition of~$(A, B)$, out of each pair, at least one vertex is in~$A$.
  Thus, it remains to prove the adjacency condition of \cref{lem:exclusive-completeness}.
  As $\psi_{r, i}(v)$, if contained in~$B$, is a singleton in~$G'[B]$, by construction, there is no edge in $G[B]$ between any two pairs of form~(ii).
  There is no edge between two pairs of form~(i) since, by definition of~$B$, for each edge gadget $j \in [m]$, there is exactly one pair of form~(i) containing a vertex in~$B$ and there is no edge between any two pairs of form~(i) for distinct edge gadgets~$j$.
  Finally, there is no edge in $G[B]$ between two pairs of form~(i) and (ii): Assume there is such an edge $e$ and let $v \in V(G_r)$ and $j \in [m]$ correspond to the two pairs.
  By construction, $e$ is between two pairs of form~(i) and (ii) that correspond to the same instance~$r$ (otherwise, no edge has been introduced between them).
  Moreover, $v \in e_j$ for $e_j \in E(G_r)$.
  That is, $e = \{w^v_{r, j}, x_{r, i, v}\}$.
  We have $w^v_{r, j} \in B$ only if $v \in V'$.
  However, $x_{r, i , v} \in A$ by definition and, thus, $e \not\subseteq B$.
  Thus, the conditions of \cref{lem:exclusive-completeness} are satisfied, meaning that $G[B] \in \Pi$.
 It follows that $(A, B)$ is the required \pipartition.
\end{proof}

It is not hard to verify that the construction can be carried out in
polynomial time. Since
$d \leq \poly(\log t + \max_{i = 1}^t |V(G_i)|)$, we thus have shown
that all the conditions of cross-compositions are satisfied, yielding
\cref{thm:lowerbound}.

\section{Kernels for Parameterization by the Size of One of the Parts}
In this section, we prove that \textsc{$(\Pi_{A},\Pi_{B})$-Recognition} has a polynomial kernel parameterized by the size of one of the parts of the bipartition when $\Pi_A$ and $\Pi_B$ satisfy certain general technical conditions. To simplify the presentation, we pick $B$ to be the part whose size is at most the parameter $k$. We then consider the conditions that $\Pi_A$ is characterized by forbidden induced subgraphs, each of size at most $d$, and $\Pi_B$ is hereditary (closed under taking induced subgraphs). In the first subsection, we give a polynomial kernel with $\Oh(d!\, (k+1)^d)$ vertices in this general setting. In the second subsection, we consider the restricted setting of \deltacluster: $\Pi_A$ is the set of all cluster graphs ($P_3$-free graphs) and $\Pi_B$ a hereditary property that contains only graphs of degree at most~$\Delta$. Although the result of the first subsection implies a kernel with $\Oh(k^{3})$ vertices in this setting, we prove that {\deltacluster} actually has a kernel with $\Oh((\Delta^2+1) k^2)$~vertices.

\subsection{A Kernel in the Generic Setting}
In this subsection, we prove that \textsc{$(\Pi_{A},\Pi_{B})$-Recognition} has a polynomial kernel with $\Oh((d+1)!\, (k+1)^d)$ vertices,  when $\Pi_A$ can be characterized by forbidden induced subgraphs, each of size at most $d$, and $\Pi_B$ is hereditary. We obtain the kernel by applying a powerful lemma of Fomin~\etal~\cite{FominSV13} that is based on the Sunflower Lemma (see~\cite{FG06,CFK+15}).

Let $\mathcal{U}$ be a universe and let $\mathcal{F}$ be a set of subsets of $\mathcal{U}$. Recall that $X \subseteq \mathcal{U}$ is a \emph{hitting set} of $\mathcal{U}$ if $X \cap F \not= \emptyset$ for every $F \in \mathcal{F}$.

\begin{lemma}[{\cite[Lemma~2]{FominSV13}}]\label[lemma]{lem:sunf}
Let $\mathcal{U}$ be a universe and let $k$ be an integer. Let $\mathcal{F}$ be a set of subsets of $\mathcal{U}$, each of size at most $d$. Then in $\Oh(|\mathcal{F}|\, (k+|\mathcal{F}|))$ time, we can find a set $\mathcal{F}' \subseteq \mathcal{F}$ of size at most $d!\, (k+1)^d$ such that for every $X \subseteq \mathcal{U}$ of size at most $k$, $X$ is a minimal hitting set of $\mathcal{F}$ if and only if $X$ is a minimal hitting set of~$\mathcal{F}'$.
\end{lemma}

We also need the following observation, inspired by a similar observation of Kratsch~\cite[Lemma~3]{Kratsch12}.

\begin{proposition} \label[proposition]{prp:sunf:obs}
Let $\mathcal{U}$ be a universe and let $\mathcal{F}',\mathcal{F}^*,\mathcal{F}$ be sets of subsets of $\mathcal{U}$ such that $\mathcal{F}' \subseteq \mathcal{F}^* \subseteq \mathcal{F}$. Suppose that for every $X \subseteq \mathcal{U}$ of size at most $k$, $X$ is a minimal hitting set of $\mathcal{F}$ if and only if $X$ is a minimal hitting set of $\mathcal{F}'$. Then for every $X \subseteq \mathcal{U}$ of size at most $k$, $X$ is a minimal hitting set of $\mathcal{F}$ if and only if $X$ is a minimal hitting set of $\mathcal{F}^{*}$.
\end{proposition}
\begin{proof}
Let $X \subseteq \mathcal{U}$ be of size at most $k$. Suppose that $X$ is a minimal hitting set of $\mathcal{F}$. Then $X$ is a hitting set of $\mathcal{F}^*$, as $\mathcal{F}^* \subseteq \mathcal{F}$.
If a set $X' \subset X$ would be a hitting set of $\mathcal{F}^*$, then $X'$ would also be a hitting set of $\mathcal{F}'$, because $\mathcal{F}' \subseteq \mathcal{F}^*$. However, using the assumption in the proposition statement, $X$ is already a minimal hitting set of $\mathcal{F}'$, a contradiction. Hence, $X$ is a minimal hitting set of $\mathcal{F}^*$.

Suppose that $X$ is a minimal hitting set of $\mathcal{F}^*$. Then $X$ is a hitting set of $\mathcal{F}'$, as $\mathcal{F}' \subseteq \mathcal{F}^*$. Suppose that $X' \subset X$ is a minimal hitting set of $\mathcal{F}'$. This implies that $X'$ is a minimal hitting set of $\mathcal{F}$ by the assumption in the proposition statement. Then $X'$ is also a hitting set of $\mathcal{F}^*$, as $\mathcal{F}^* \subseteq \mathcal{F}$, contradicting the minimality of $X$. Hence, $X$ is a minimal hitting set of $\mathcal{F}'$. Then the assumption in the proposition statement implies that $X$ is a minimal hitting set of $\mathcal{F}$.
\end{proof}

In the remainder, let $(G,k)$ be an instance of \textsc{$(\Pi_{A},\Pi_{B})$-Recognition} where $\Pi_A$ can be characterized by a collection $\mathcal{H}$ of forbidden induced subgraphs, each of size at most $d$, and $\Pi_B$ is hereditary. Throughout, for a graph $G'$, let $\mathcal{F}(G')$ be the set of subsets of $V(G')$ that induce a subgraph of $G'$ isomorphic to a member of $\mathcal{H}$. We observe the following.

\begin{proposition} \label[proposition]{prp:sunf:obs2}
If $B \subseteq V(G)$ is a hitting set of $\mathcal{F}(G)$ and $G[B] \in \Pi_B$, then $(V(G) \setminus B, B)$ is a $(\Pi_A,\Pi_B)$-partition of $G$.
\end{proposition}
\begin{proof}
Suppose that a subgraph $H$ of $G-B$ is isomorphic to a member of~$\mathcal{H}$. Then $V(H)$ is contained in $\mathcal{F}(G)$. Hence, $H$ contains a vertex of $B$, a contradiction. Therefore, $G-B \in \Pi_A$.
\end{proof}

\begin{proposition} \label[proposition]{prp:sunf:minimal}
If $G$ admits a $(\Pi_A,\Pi_B)$-partition, then $G$ admits a $(\Pi_A,\Pi_B)$-partition $(A,B)$ with $|B|$ minimum such that $B$ is a minimal hitting set of $\mathcal{F}(G)$.
\end{proposition}
\begin{proof}
Let $(A,B)$ be a $(\Pi_A,\Pi_B)$-partition of $G$ such that $|B|$ is minimum. Clearly, $B$ is a hitting set of $\mathcal{F}(G)$, or $G[A]$ would still contain a subgraph isomorphic to a member of $\mathcal{H}$. Suppose there exists a set $B' \subset B$ such that $B'$ is still a hitting set of $\mathcal{F}(G)$. Note that $G[B'] \in \Pi_B$, because $\Pi_B$ is hereditary. Hence, \cref{prp:sunf:obs2} implies that $(A \cup (B \setminus B'), B')$ is a $(\Pi_A,\Pi_B)$-partition of $G$. However, this contradicts the choice of $(A,B)$, particularly the minimality of $|B|$. The proposition follows.
\end{proof}

We now describe the single reduction rule of the kernel.

\begin{rrule} \label[rrule]{rule:bss}
Apply the algorithm of \cref{lem:sunf} with $\mathcal{U} = V(G)$, $\mathcal{F} = \mathcal{F}(G)$, and $k$, and let $\mathcal{F}'$ be the resulting set. Let $T = \bigcup_{F \in \mathcal{F}'} F$ be the set of vertices contained in $\mathcal{F}'$ and let $R = V(G) \setminus T$. Remove $R$ from $G$.
\end{rrule}
\begin{proof}
Let $G' = G[T] = G-R$. We prove that $(G,k)$ is a yes-instance if and only if $(G',k)$ is.

Suppose that $(G,k)$ is a yes-instance, and let $(A,B)$ be a $(\Pi_A,\Pi_B)$-partition of~$G$ such that $|B| \leq k$. Since $\Pi_A$ can be characterized by a collection of forbidden induced subgraphs, it is hereditary. Recall that $\Pi_B$ is hereditary as well. Hence, $G[A \setminus R] \in \Pi_A$ and $G[B \setminus R] \in \Pi_B$. Therefore, $(A \setminus R, B \setminus R)$ is a $(\Pi_A,\Pi_B)$-partition of $G'$, and $(G',k)$ is a yes-instance.

Suppose that $(G',k)$ is a yes-instance, and let $(A',B')$ be a $(\Pi_A,\Pi_B)$-partition of~$G'$ such that $|B'| \leq k$. By \cref{prp:sunf:minimal}, we may assume that $B'$ is a minimal hitting set of $\mathcal{F}(G')$. Recall that by \cref{lem:sunf}, for every $X \subseteq \mathcal{U} = V(G)$ of size at most $k$, $X$ is a minimal hitting set of $\mathcal{F} = \mathcal{F}(G)$ if and only if $X$ is a minimal hitting set of $\mathcal{F}'$. Also note that by the definition of $T$ and $G'$, it follows that $\mathcal{F}' \subseteq \mathcal{F}(G') \subseteq \mathcal{F}(G) = \mathcal{F}$. Combined with \cref{prp:sunf:obs}, all this implies that $B'$ is a (minimal) hitting set of $\mathcal{F}(G)$. But then \cref{prp:sunf:obs2} implies that $(V(G) \setminus B', B') = (A' \cup R, B')$ is a $(\Pi_A,\Pi_B)$-partition of $G$ such that $|B'| \leq k$. Therefore, $(G,k)$ is a yes-instance.
\end{proof}

\begin{proof}[Proof of \cref{thm:general-kernel}]
Let $(G,k)$ be an instance of \textsc{$(\Pi_{A},\Pi_{B})$-Recognition}, let $\Pi_A$ be characterized by a collection $\mathcal{H}$ of forbidden induced subgraphs, each of constant size, and let $\Pi_B$ be hereditary. Apply \cref{rule:bss}. Since the number of sets in $\mathcal{F}(G)$ is $\Oh(d|V(G)|^d)$, both constructing $\mathcal{F}(G)$ and the algorithm of \cref{lem:sunf} take time polynomial in $|V(G)|^d$, $d$, and $k$. Moreover, the produced set $\mathcal{F}'$ has size at most $d!\,(k+1)^d$, implying that $|V(G')| \leq d \cdot d!\, (k+1)^d \leq (d+1)!\, (k+1)^d$ where~$G'$ is the graph produced by the rule.  By the correctness of \cref{rule:bss} and the fact that the number of edges in~$G'$ is at most quadratic in~$|V(G')|$, this is indeed a polynomial kernel.
\end{proof}

\subsection{Smaller Kernels for a Restricted Setting: \deltaclustertitle}
In this subsection, we prove that {\deltacluster}, parameterized by the size~$k$ of $B$, has a kernel with $\Oh((\Delta^2+1) k^2)$ vertices. This improves on \cref{thm:general-kernel}, which implies a kernel with $\Oh((k+1)^3)$ vertices. Recall that {\deltacluster} is the restriction of {\picluster} to the case when all graphs containing a vertex of degree at least~$\Delta+1$ are forbidden induced subgraphs of $\Pi$. Throughout, we say that a cluster-$\Delta$ partition of $G$ is \emph{valid} if $|B| \leq k$.

The first step of the kernel is to compute a maximal set~$\mathcal{P}$ of vertex-disjoint induced~$P_3$s. We call~${\cal P}$ a \emph{$P_3$-packing}. We let~$V({\cal P})$ denote the set of vertices of the~$P_3$s in~${\cal P}$.

\newcommand{\Gr}{\ensuremath{G-V({\cal P})}}

\begin{rrule}
\label[rrule]{rule:is-packing-size}
Let $(G, k)$ be an instance of \deltacluster, and let~${\cal P}$ be a $P_3$-packing. If~$|{\cal P}|>k$, then reject.
\end{rrule}
\begin{proof}%
For each~$P_3$, at least one vertex must be in~$B$. Therefore, if $|{\cal P}|>k$, then~$|B|> k$ for any cluster-$\Pi_\Delta$ partition $(A, B)$ of $G$.
\end{proof}
Since~${\cal P}$ is a maximal set of~$P_3$s, $G-V({\cal P})$ is a cluster graph. The first step of the kernelization is to identify vertices of~$V({\cal P})$ that are in~$B$ in every valid cluster-$\Pi_\Delta$ partition. %

\begin{definition}\label[definition]{def:fixed}
For a vertex~$u\in V({\cal P})$, we say that $u$ is \emph{fixed} if either:

\begin{enumerate}
\item $u$ has neighbors in at least~$k+2$ different clusters of~$\Gr$; or
\item there is a cluster $C$ in~$G-V({\cal P})$ such that~$u$ has (at least) $\Delta + 2$ neighbors and (at least) $\Delta+2$ nonneighbors in~$C$.
\end{enumerate}

A fixed vertex $u$ is said to be \emph{heavy} if it has neighbors in at least~$k+2$ different clusters of~$\Gr$ (\ie, satisfies condition~1 above); otherwise, $u$ is \emph{nonheavy}.
\end{definition}

\begin{lemma}
\label[lemma]{lem:fixed}
  Let $(G, k)$ be an instance of \deltacluster, let~${\cal P}$ be a $P_3$-packing, and let~$u$ be a fixed vertex. If~$G$ has a valid cluster-$\Pi_\Delta$ partition~$(A,B)$, then~$u\in B$.
\end{lemma}
\begin{proof}%
  \emph{Case~1:~$u$  is heavy.} If~$u\in A$, then there is at most one cluster~$C$
  of~$G-V({\cal P})$ such that~$A$ contains vertices of~$N(u)\cup C$. Therefore,~$B$ contains
  vertices of~$k+1$ clusters of~$G-V({\cal P})$ and thus~$|B|>k$.

  \emph{Case~2: $u$ is nonheavy.} Since~$u$ is fixed, there is a cluster~$C$ in~$G-V({\cal P})$ such
    that~$u$ has (at least) $\Delta + 2$ neighbors and (at least) $\Delta +2$ nonneighbors in~$C$. Let~$v_1,v_2, \ldots, v_{\Delta+2}$ be $\Delta+2$ neighbors
  of~$u$ in~$C$, and let~$w_1,w_2, \ldots, w_{\Delta+2}$ be $\Delta+2$ nonneighbors of~$u$ in~$C$. Assume, towards a
  contradiction, that there is a cluster-$\Pi_\Delta$ partition~$(A,B)$ with~$u\in A$. Since each of $G[\{v_1,v_2, \ldots, v_{\Delta+2}\}]$ and $G[\{w_1,w_2, \ldots, w_{\Delta+2}\}]$ is a clique on $\Delta+2$ vertices (and hence of degree $\Delta+1$), $A$ must contain at least one vertex $v_i \in \{v_1,v_2, \ldots, v_{\Delta+2}\}$ and at least one vertex $w_j \in \{w_1,w_2, \ldots, w_{\Delta+2}\}$. But then $(u, v_i, w_j)$ forms an induced~$P_3$ in~$A$.
\end{proof}
Next, we label certain vertices in~$V\setminus V({\cal P})$ as \emph{important} using the following scheme.\\

\noindent {\bf Labeling Scheme}
\begin{itemize}
\item[(i)] For each (fixed) heavy vertex~$u$ of~$V({\cal P})$, pick
  $k+2$ (distinct) clusters \\ $C_1, \ldots, C_{k+2}$ in $\Gr$ such
  that, for each $i \in [k+2]$, $C_i$ contains a neighbor $v_i$ of $u$
  and label~$v_1, v_2, \ldots, v_{k + 2}$ as important.

\item[(iii)] For each (fixed) nonheavy vertex~$u$ of~$V({\cal P})$, pick an arbitrary cluster~$C$ of~$G-V({\cal P})$ such that~$u$
  has $\Delta +2$ neighbors $v_1,v_2, \ldots, v_{\Delta+2}$ and $\Delta +2$ nonneighbors~$w_1,w_2, \ldots, w_{\Delta+2}$~in~$C$, and label~$v_1,v_2, \ldots, v_{\Delta+2}, w_1, w_2, \ldots, w_{\Delta+2}$ as important.

\item[(iii)]   For each nonfixed vertex~$u$ of~$V({\cal P})$, and each cluster~$C$ of~$G-V({\cal P})$
  containing at least one neighbor of~$u$, label $\min\{\Delta+2, |N(u) \cap C|\}$ (arbitrary) neighbors of $u$ in $C$ and $\min\{\Delta+2, |C- N(u)|\}$ (arbitrary) nonneighbors of $u$ in $C$ as important.
\end{itemize}
Any vertex in $V\setminus V({\cal P})$ that was not labeled in this scheme is called \emph{unimportant}.

\begin{observation}\label[observation]{obs:important-verts}
  If~$(G,k)$ is reduced with respect to \cref{rule:is-packing-size}, then the number of
  vertices that are marked as important is~$\Oh((\Delta +1)\cdot k^2)$.
\end{observation}

\begin{proof}
After \cref{rule:is-packing-size}, we have $|V({\cal P})| \leq 3k$. Each heavy vertex in $V({\cal P})$ labels $k+2$ vertices in $V\setminus V({\cal P})$ as important, according to condition (i) of the labeling scheme. Therefore, the total number of vertices in $V\setminus V({\cal P})$ labeled as important by heavy vertices is $\Oh(k^2)$. Each fixed nonheavy vertex in $V({\cal P})$ labels $2\Delta +4$ vertices in $V\setminus V({\cal P})$ as important, according to condition (ii) of the labeling scheme. Therefore, the total number of vertices in $V\setminus V({\cal P})$ labeled as important by fixed nonheavy vertices is $\Oh(\Delta \cdot k +k)$. Each nonfixed vertex $v \in V({\cal P})$ is adjacent to at most $k+1$ clusters in $G - V({\cal P})$  (otherwise $v$ would be fixed), and can label at most $2\Delta +4$ vertices in each adjacent cluster as important (according to condition (iii) of the labeling scheme). Therefore, a nonfixed vertex $v$ labels $\Oh(\Delta \cdot k+k)$ many vertices in $V\setminus V({\cal P})$ as important. It follows that the at most $3k$ vertices in $V({\cal P})$ label $\Oh(\Delta \cdot k^2 + k^2)=\Oh((\Delta+1) \cdot k^2)$ vertices of $V\setminus V({\cal P})$ as important.
\end{proof}

We now present several reduction rules that use the above labeling scheme.

\begin{rrule}\label[rrule]{rule:is-unimportant-clique}
  If there is a cluster~$C$ in~$G - V({\cal P})$ such that all vertices in~$C$ are unimportant, then remove~$C$ from~$G$.
\end{rrule}
\begin{proof}%
 If $G$ has a valid cluster-$\Pi_\Delta$ partition $(A, B)$, then obviously so does $G-C$. To prove the converse, suppose that $(A, B)$ is a valid cluster-$\Pi_\Delta$ partition of
$G-C$. We claim that $(A \cup C, B)$ is a cluster-$\Pi_\Delta$ partition, which obviously satisfies $|B| \leq k$, and hence, is valid.

Suppose not. Then there must exist a vertex $u \in A$ that has a neighbor in $C$. Clearly, $u \in V({\cal P})$ because $G - V({\cal P})$ is a cluster graph containing cluster $C$ and $u \notin C$.
Vertex $u$ cannot be fixed; otherwise, since no vertex in $C$ is important, $u$~would remain fixed in $G-C$, and hence, $u$ would not belong to $A$ by \cref{lem:fixed}. Since $u$ is adjacent to $C$, it follows from condition (iii) of the labeling scheme that at least $\min\{\Delta+2, |N(u) \cap C|\} > 0$ (since $u$ is adjacent to $C$) neighbors of $u$ in $C$ are labeled important. This, however, contradicts the assumption of the reduction rule.
\end{proof}

\begin{rrule}\label[rrule]{rule:is-many-unimportant}
  If there is a cluster~$C$ in~$G - V({\cal P})$ such that~$C$ contains (at least) $\Delta +3$ unimportant vertices, then remove one of these unimportant vertices.
\end{rrule}
\begin{proof}%
Let $w$ be an unimportant vertex in~$C$~that is removed by an application of this rule. If $G$ has a valid cluster-$\Pi_\Delta$ partition $(A, B)$, then clearly so does $G-w$. To prove the converse, suppose that $G-w$ has a valid cluster-$\Pi_\Delta$ partition $(A, B)$. We claim that $(A \cup \{w\}, B)$ is a cluster-$\Pi_\Delta$ partition of $G$, which obviously will be valid.

Since $C$ contains
 $\Delta+2$ neighbors $w_1, \ldots, w_{\Delta+2}$ of $w$ that are unimportant and the maximum degree of $G[B]$ is at most $\Delta$, at least one of these vertices, say $w_1$, belongs to a cluster $C'$ in $A$. Every vertex in $C'$ that is in $V \setminus V({\cal P})$ must be in $C$, and hence, is adjacent to $w$. Now suppose that a vertex $u \in V({\cal P})$ is in $C'$. We will show that $u$ must be adjacent to $w$. Suppose, towards a contradiction, that $u$ is not adjacent to $w$. Since $w$ is unimportant, $u$ cannot be fixed (otherwise, $u$ would be fixed in $G-w$, and would belong to~$B$ by \cref{lem:fixed}).
 Since $u$ is adjacent to $w_1 \in C$, and $u$ is nonfixed, condition (iii) of the labeling scheme applies to $u$, and in particular, $\min\{\Delta+2, |C- N(u)|\}$ nonneighbors of $u$ in $C$ are labeled as important. Since $w$ is a nonneighbor of $u$ in $C$, and $w$ is unimportant, it follows that there are $\Delta+2$ nonneighbors of $u$ in $C$ that are different from $w$, and that are labeled important. At least one of these vertices, say $x$, must be in $A$. But then $(u, w_1, x)$ forms an induced $P_3$ in $A$ (note that $w_1$ is adjacent to $x$ since both of them are in~$C$). This is a contradiction. It follows that each vertex in~$C'$ is adjacent to $w$, and hence, $C' \cup \{w\}$ is a cluster in $A \cup \{w\}$.

 To conclude that $G[A \cup \{w\}]$ is a cluster graph, it remains to show that no vertex~$u$ that belongs to another cluster $C'' \neq C'$ in $G[A \cup \{w\}]$ is adjacent to $w$. Suppose not. Then clearly $u \in V({\cal P})$, and by the same arguments as above, $u$ cannot be fixed. Since $u$ is adjacent to $w \in C$, and $u$ is nonfixed, condition (iii) of the labeling scheme applies to $u$, and in particular, $\min\{\Delta+2, |C \cap N(u)|\}$ neighbors of $u$ in $C$ are labeled as important. Since $w$ is unimportant, it follows that there are $\Delta+2$ neighbors of $u$ in~$C$ that are different from $w$, and that are labeled important. One of these neighbors, say $x$, must be in $A$, and hence, must belong to the same cluster as both $u$ and $w_1$ (because $w_1 \in C$). But then this implies that $C' = C''$, contradicting our assumption that $u$ belongs to a different cluster than $C'$.

 It follows that $(A \cup \{w\}, B)$ is a valid cluster-$\Pi_\Delta$ partition of $G$.
\end{proof}

\begin{lemma}\label[lemma]{lem:kernelbound}
  Let~$(G,k)$ be an instance of \deltacluster~that is reduced with respect to the above reduction rules, then~$G$ has~$\Oh((\Delta^2 +1) \cdot k^2)$ vertices.
\end{lemma}
\begin{proof}
  Since~$(G,k)$ is reduced with respect to \cref{rule:is-packing-size},~$|V({\cal
    P})|\le 3k$. By \cref{obs:important-verts}, the number of important vertices
  in~$V\setminus V({\cal P})$ is~$\Oh((\Delta +1) \cdot k^2)$. Thus, to show the upper bound on the kernel size, it
  remains to upper bound the number of unimportant vertices in~$V\setminus V({\cal P})$.

  To this end, we first upper bound the number of clusters in~$\Gr$. Since~$(G,k)$ is reduced
  with respect to \cref{rule:is-unimportant-clique}, every cluster in~$\Gr$ contains
  at least one important vertex. By \cref{obs:important-verts}, the number of
  important vertices in~$G$ is~$\Oh((\Delta +1)\cdot k^2)$. Thus, the total number of clusters in~$\Gr$
  is~$\Oh((\Delta+1) \cdot k^2)$.

  Now, observe that since~$(G,k)$ is reduced with respect to
  \cref{rule:is-many-unimportant}, there are at most $\Delta+3$ unimportant vertices in
  each cluster, and thus~$\Oh((\Delta^2 +1) \cdot k^2)$ unimportant vertices overall.
\end{proof}

{
\begin{theorem}
\deltacluster, parameterized by the size $k$ of $B$, has a polynomial kernel with $\Oh((\Delta^2 +1) \cdot k^2)$ vertices, that is computable in time $\Oh(k\cdot (m+n))$, where $n$ and $m$ are the number of vertices and edges, respectively, in the graph.
\end{theorem}
}
\begin{proof}
Given an instance $(G, k)$ of \deltacluster, we start by computing a $P_3$-packing ${\cal P}$. Afterwards, we apply \cref{rule:is-packing-size}--\cref{rule:is-many-unimportant}. If after the application of these reduction rules the instance $(G, k)$ is not rejected, then these reduction rules result in an equivalent instance $(G', k)$ of \deltacluster~satisfying $|V(G')| = \Oh((\Delta^2 +1) \cdot k^2)$ by \cref{lem:kernelbound}. Clearly, this implies also that the size of~$G'$ is polynomial in~$k$. Therefore, what is left is analyzing the running time taken to apply \cref{rule:is-packing-size}--\cref{rule:is-many-unimportant}.

First, it is important to observe that each reduction rule is applied exhaustively once, meaning that we apply a particular reduction rule exhaustively, but no more after any of the other reduction rules have been applied. In particular, after applying any of the reduction rules, $\Gr$ is still a cluster graph, because the reduction rules only remove vertices. %
Moreover, the reduction rules leave unchanged the status of each vertex $u \in V({\cal P})$ as (fixed) heavy, (fixed) nonheavy, or nonfixed, because only (edges to) unimportant vertices are removed and the important vertices maintain the status of $u$. The reduction rules also leave unchanged the label of each vertex in $V\setminus V({\cal P})$ as important or unimportant, for the same reason. Therefore, it suffices to analyze the running time of a single, exhaustive application of each of the reduction rules.

To apply \cref{rule:is-packing-size}, we observe that, as is well known, a $P_3$ in $G$ can be recognized in $\Oh(m+n)$ time. (For instance, this can done by computing the connected components of $G$, and the degree of each vertex in $G$. We can then identify a connected component that is not a clique, which must exist if a $P_3$ exists. A $P_3$ in such a component can then be computed in linear time.) Therefore, ${\cal P}$ can be greedily computed in time $\Oh(k\cdot (m+n))$ (note that if more than $k$ $P_3$'s are identified in $G$, then the instance can be immediately rejected). It follows from the preceding that \cref{rule:is-packing-size} can be applied in $\Oh(k\cdot (m+n))$ time.

Next, we show that we can classify the vertices in $V({\cal P})$ into fixed heavy, fixed nonheavy, and nonfixed in $\Oh(m+n)$ time. To do so, we first compute the clusters in $G-V({\cal P})$, and
color the vertices of different clusters with different colors, \ie, each vertex in the $i$-th cluster receives color $i$, for some arbitrary numbering of the clusters. We then iterate through the vertices in $V({\cal P})$, and for each vertex $v \in V({\cal P})$, we iterate through its neighbors in $G-V({\cal P})$. If $v$ has at least $k+2$ neighbors in $G-V({\cal P})$ with different colors (this can be determined in time $\Oh(\deg(v))$ by sorting the colors of the neighbors of $v$ using Counting Sort), then we
define $v$ to be fixed and heavy. For each vertex in $V({\cal P})$ that has not been classified yet, we iterate through its neighbors in $G-V({\cal P})$, and partition its neighbors into subsets, such that all neighbors in the same subset have the same color (belong to the same cluster); for each such subset of neighbors of size $s \geq \Delta +2$ that belong to a cluster $C$, we check if $|C| \geq s+\Delta +2$, and if it is, we classify $v$ as fixed but nonheavy. All the remaining vertices in $V({\cal P})$ are defined to be nonfixed. Clearly, this whole process can be done in $\Oh(m+n)$ time.

Afterwards, we label the vertices in $G-V({\cal P})$ as important or unimportant. To do so, for each heavy vertex $v$ in $V({\cal P})$, we iterate through its neighbors in $G-V({\cal P})$ to pick $k+2$ neighbors of distinct colors, and label them important. This can be done in time $\Oh(\deg(v))$, and hence, $\Oh(m+n)$ time overall. For each fixed nonheavy vertex $v$ in $V({\cal P})$, we iterate through its neighbors to determine a cluster $C$ such that $v$ has $\Delta +2$ neighbors and $\Delta+2$ nonneighbors in $C$, and label those vertices as important. Again, this can be done in time $\Oh(\deg(v))$, and hence, $\Oh(m+n)$ time overall.
Finally, for each nonfixed vertex $v$ in $V({\cal P})$, we iterate through its neighbors to partition them into subsets of the same color; for each subset of neighbors of the same color that belong to a cluster $C$, we label $\min\{\Delta+2, |N(u) \cap C|\}$ (arbitrary) neighbors of~$u$ in~$C$ and $\min\{\Delta+2, |C- N(u)|\}$ (arbitrary) nonneighbors of $v$ in $C$ as important. This can be done in time $\Oh(\Delta + \deg(v))$, and hence in time $\Oh(\Delta \cdot (m+n))=\Oh(k \cdot (m+n))$ overall.

To apply \cref{rule:is-unimportant-clique}, we go over every cluster $C$ in $G-V({\cal P})$, checking if it contains any important vertices, and if not, we remove $C$ from $G$. This can be done in $\Oh(m+n)$ time.

Finally, to apply \cref{rule:is-many-unimportant}, we again go over every cluster $C$ in $G-V({\cal P})$, and remove all but $\Delta +2$ unimportant vertices from $C$. Again, this can be done in $\Oh(m+n)$ time. It follows that the kernelization algorithm runs in $\Oh(k \cdot (m+n))$ time.
\end{proof}

\section{Conclusion and Outlook}
As we have seen in this paper, the pushing process is not only useful for finding
efficient algorithms for \textsc{$(\Pi_A,\Pi_B)$-Recognition} as demonstrated by Kanj \etal~\cite{KKSL17}, but can
also be used to classify when or when not such problems admit polynomial kernels. Herein, we focused on the well-motivated case when
$\Pi_A$ is the set of cluster graphs on the first level above
triviality; when $\Pi_A$ is characterized by forbidden induced
subgraphs of order at least three. A natural next step is to
check to which extent our results carry over to other sets of
forbidden subgraphs for~$\Pi_A$. %

The lower bound given in \cref{thm:main} should in a straightforward
manner extend to graph classes~$\Pi_A$ that are closed under disjoint
union, have neighborhood diversity\footnote{The neighborhood diversity
  of a graph is the number of different open neighborhoods.} at
most~$k$, and contain cluster graphs. A more challenging avenue is to
try to apply our techniques to related partitioning problems such as
\textsc{Rectangle Stabbing}.

Finally, when we parameterized by the size~$k$ of one of the parts, we
obtained a kernel with $\Oh(k^d)$ vertices (\cref{thm:general-kernel}), where
$d$ is the largest order of a forbidden subgraph of the other part.
Since the techniques used herein are similar to the ones for
$d$-\textsc{Hitting Set}, it is natural to ask whether this upper bound
can be improved to~$\Oh(k^{d - 1})$.

\bibliographystyle{abbrvnat}
\bibliography{kernel}

\end{document}